\documentclass[11pt]{article}
\usepackage{graphicx} % Required for inserting images
\usepackage{xspace}
\usepackage{verbatim}
\usepackage[dvipsnames]{xcolor}
\usepackage{graphicx}
\usepackage{tabularx}

\usepackage{amsmath}
\usepackage{bbm}
\usepackage[showonlyrefs]{mathtools}
\usepackage{mathstyle}
\usepackage{empheq}
\usepackage{amstext,amssymb,amsfonts}
\usepackage{nicefrac}
\usepackage{bm}
\usepackage{derivative}

\usepackage{algorithm}
\usepackage{algpseudocode}
\usepackage{caption}
\usepackage{subcaption}

\usepackage{tikz}
\usepackage[normalem]{ulem}
\usepackage{soul}
\usetikzlibrary{arrows.meta,positioning,calc}

\usepackage{todonotes}

\newcounter{mynotes}
\usepackage{amsthm}
\usepackage{thmtools}
\usepackage{thm-restate}

\usepackage{textcomp,setspace}
\usepackage[top=1in, bottom=1in, left=1in, right=1in]{geometry}
\usepackage{nameref}
\usepackage{color}
\usepackage[pagebackref,colorlinks,linkcolor=blue,filecolor = blue,citecolor = ForestGreen, urlcolor
= blue, hyperfootnotes=false,hypertexnames=false]{hyperref}
\usepackage[capitalize, noabbrev, nameinlink]{cleveref}
\declaretheorem[within=section]{theorem}
\declaretheorem[sibling=theorem]{corollary}

\declaretheorem[sibling=theorem]{lemma}
\declaretheorem[sibling=theorem]{claim}

\declaretheorem[sibling=theorem]{definition}

\declaretheorem[sibling=theorem]{remark}

\declaretheorem[sibling=theorem]{fact}

\declaretheorem[sibling=theorem]{example}

\crefname{claim}{Claim}{Claims}
\crefname{fact}{Fact}{Facts}
\newcounter{termcounter}
\renewcommand{\thetermcounter}{\Alph{termcounter}}
\crefname{term}{term}{terms}
\creflabelformat{term}{\textup{(#2#1#3)}}

\makeatletter
\def\term{\@ifnextchar[\term@optarg\term@noarg}%]
\def\term@optarg[#1]#2{%
  \textup{(#1)}%
  \def\@currentlabel{#1}%
  \def\cref@currentlabel{[][2147483647][]#1}%
  \cref@label[term]{#2}}
\def\term@noarg#1{%
  \refstepcounter{termcounter}%
  \textup{(\thetermcounter)}%
  \cref@label[term]{#1}}
\makeatother
\newcommand{\ignore}[1]{}

\newcommand{\bits}{\{0,1\}}
\newcommand{\zo}{\bits}

\newcommand{\sbits}{\{1,-1\}}

\newcommand{\varvarTV}{\mathrm{TV}}

\newcommand{\poly}{\mathrm{poly}}

\newcommand{\dist}{\mathrm{dist}}

\newcommand{\Brac}[1]{\left[#1 \right]}
\newcommand{\abs}[1]{\lvert#1\rvert}
\newcommand{\Abs}[1]{\left\lvert#1\right\rvert}

\definecolor{DSred}{rgb}{1,0,0}

\renewcommand{\leq}{\leqslant}
\renewcommand{\geq}{\geqslant}
\renewcommand{\ge}{\geqslant}
\renewcommand{\le}{\leqslant}
\renewcommand{\epsilon}{\varepsilon}
\newcommand{\eps}{\epsilon}

\newcommand{\R}{\mathbb{R}}

\newcommand{\N}{\mathbb{N}}

\newcommand{\cA}{\mathcal A}
\newcommand{\cB}{\mathcal B}

\newcommand{\cD}{\mathcal D}

\newcommand{\cF}{\mathcal F}

\newcommand{\cH}{\mathcal H}

\newcommand{\cL}{\mathcal L}

\newcommand{\cN}{\mathcal N}

\newcommand{\cP}{\mathcal P}

\newcommand{\cR}{\mathcal R}

\newcommand{\cT}{\mathcal T}

\newcommand{\cX}{\mathcal X}
\newcommand{\cY}{\mathcal Y}

\newcommand{\bbR}{\mathbb R}

\newcommand{\bbD}{\mathbb D}
\newcommand{\bbT}{\mathbb T}

\newcommand{\Esymb}{{\bf E}}

\newcommand{\Psymb}{{\bf Pr}}

\DeclareMathOperator*{\E}{\Esymb}

\DeclareMathOperator*{\ProbOp}{\Psymb}

\renewcommand{\Pr}{\ProbOp}

\newcommand{\argEx}[2]{\E_{#1}\Brac{#2}}

\newcommand{\Ind}{\mathbb I}

\newcommand{\bern}{\mathrm{Bern}}
\allowdisplaybreaks

\newif\ifshowcomments
\showcommentsfalse

\ifshowcomments
  \newcommand{\RD}[1]{{\color{blue} (Rathin: #1)}}
    \newcommand{\RDR}[1]{{\color{red} (Rathin: #1)}}
  \newcommand{\RF}[1]{{\color{ForestGreen} (Renato: #1)}}
  \newcommand{\MB}[1]{{\color{purple} (Mark: #1)}}

\else
  \newcommand{\RD}[1]{}
  \newcommand{\RF}[1]{}
  \newcommand{\MB}[1]{}
  \newcommand{\RDR}[1]{}

\fi

\newcommand\numberthis{\addtocounter{equation}{1}\tag{\theequation}}

\newcommand{\ee}{\E}

\DeclareMathOperator{\err}{err}
\DeclareMathOperator{\opt}{opt}

\newcommand{\bin}{\mathrm{Bin}}

\newcommand{\define}{\vcentcolon=}

\newcommand{\bbN}{\mathbb N}
\renewcommand{\sbits}{\{\pm 1\}}

\newcommand{\Unif}{\mathrm{Unif}}

\newcommand{\tv}{\mathrm{TV}}
\newcommand{\sset}{\cH_{P,Q}}
\newcommand{\ind}{\mathbbm{1}}
\renewcommand{\Ind}{\ind}

\newcommand{\refdist}[1]{{#1}_{\mathsf{ref}}}
\newcommand{\tgtdist}[1]{{#1}}

\newcommand{\comp}{\mathrm{comp}}
\newcommand{\sound}{\mathrm{sound}}

\newcommand{\TL}{\mathcal{TL}}

\title{Testing Distributions Against Bounded Distinguishers}
\author{
    Mark Bun \\ Boston University \and
    Rathin Desai \\ Boston University \and
    Renato Ferreira Pinto Jr. \\ Columbia University
}
\date{\today}

\begin{document}

\maketitle

\begin{abstract}
    Motivated by the challenge of testing distributions over high-dimensional or
    continuous domains, we study distribution testing with respect to
    \emph{bounded classes of distinguishers}. %In a nutshell, the goal is to use
    A representative task is to use samples from an unknown distribution $P$ over a very large domain to decide
    between two cases: $P = \refdist{P}$ for a fixed reference distribution
    $\refdist{P}$, or there exists a distinguisher $f$ in a bounded class $\cF$
    which witnesses the separation $\abs{\E_P[f] - \E_{\refdist{P}}[f]} >
    \epsilon$. This is the task of \emph{identity testing} with respect to
    \emph{fooling distance}, a name inspired by the conceptual connection with
    pseudorandomness. (Formally, our model instantiates integral probability
    metrics from Boolean classes of bounded expressivity.)

    We show that testing with respect to fooling distance is not only a natural
    computational problem that admits sample-efficient algorithms even in
    high-dimensional settings, but also one that reveals and underlies connections
    between three seemingly unrelated areas of study: testable
    learning~\cite{rubinfeld2022testingdistributionalassumptionslearning},
    verification of learning algorithms~\cite{GoldwasserRSY21}, and testing of
    structured distributions~\cite{DKN14} (whose ``$\cA_k$-testing'' model our
    framework extends). These connections yield new results for all of these
    models, including:
    \begin{enumerate}
        \item Testable \emph{proper} learners using membership queries for
        halfspaces and decision trees.
    \item A lower bound for \emph{testable} PAC verification in terms
        of Rademacher complexity, and a distribution-free verification protocol for disjoint
        unions of $k$ multidimensional rectangles.
        \item Identity testers (with respect to total variation distance) for
        decision tree distributions
        and distributions with low-degree polynomial densities,
        over Boolean and continuous hypercube domains.
    \end{enumerate}
\end{abstract}

\thispagestyle{empty}
\setcounter{page}{0}
\newpage
{
\setcounter{tocdepth}{2}
\tableofcontents
}
\thispagestyle{empty}
\setcounter{page}{0}
\newpage
\setcounter{page}{1}

\section{Introduction}

Distribution testing \cite{BatuFRSW00,GR00} captures a basic class of problems that
arise both as self-contained
scientific questions and as building blocks for larger computational tasks. One central
distribution testing problem is \emph{identity testing}: given a reference distribution $\refdist{P}$ over
domain $\cX$, and sample access to unknown input distribution $P$, the goal is
to distinguish between 1)~$P=\refdist{P}$ and 2)~$\dist_\varvarTV(P,\refdist{P})
> \epsilon$, where $\dist_\tv$ denotes total variation distance.
Unfortunately, well-known sample complexity lower bounds rule out standard distribution testing as a
feasible task in most natural \emph{high-dimensional} settings. For example, testing uniformity of a
distribution over the Boolean hypercube $\bits^n$ requires $\Theta(2^{n/2})$
samples \cite{paninski2008coincidence}.
A
related challenge afflicts distribution testing over \emph{continuous} domains, in which case
standard distribution testing is impossible using a finite number of samples.

Broadly speaking, two approaches have emerged in the literature to address these
challenges and bypass basic infeasibility results: invoking a \emph{stronger
sampling oracle}, or assuming that the input comes from a class of
\emph{structured distributions}. For example, given access to \emph{conditional}
samples, many properties of distributions over $\bits^n$ are efficiently testable
\cite{canonne2015testingprobabilitydistributionsusing,Bhattacharyya_2018,CanonneCKLW21,ChenJLW21,AdarFL24a,ChenM24,BlancaCSV25,chakrabarty2025monotonicitytestinghighdimensionaldistributions}.
In the second category, distributions classes permitting efficient testing
include Ising models and Bayesian networks
\cite{DaskalakisD019,CanonneDKS20}. Identity of
distributions over the continuous domain $[0,1]$ can be tested efficiently when
their probability density functions are
piecewise constant, low-degree polynomials, or log concave \cite{DKN14,DiakonikolasKN15}. Uniformity of monotone distributions over $\bits^n$ and $[0,1]^n$ is also
efficiently testable \cite{RubinfeldS09,AdamaszekCS10}. Identity testing, closeness testing, and distance estimation between $k$-modal distributions also have efficient algorithms \cite{daskalakis2013testing,gs009}.

What if these stronger assumptions are not applicable? In this paper, we study a third relaxation of
the problem which not only suffices to make distribution testing feasible in new challenging
settings, but also unlocks new applications to existing problems. Specifically, we propose to relax
the \emph{success criterion} of distribution testing from its standard formulation under \emph{total
variation} (TV) distance to a criterion that only contemplates a fixed class of
\emph{distinguishers}. Before making this notion precise, we motivate the approach via the following
example:

\begin{example}
    \label{example:intro}
    Given samples of the medical features of patients presenting a certain disease, which we model
    as a distribution $P$ over $\bits^n$, distinguish with high probability between the following
    two cases:
    \begin{enumerate}
        \item The distribution $P$ is uniform over $\bits^n$.
        \item There is a simple rule $f : \bits^n \to \bits$, representable as a
            size-$s$ decision tree, which captures an overrepresented subgroup in the sense that
            $\E_{x \sim P}[f(x)] > \E_{x \sim \bits^n}[f(x)] + \epsilon$.
    \end{enumerate}
    In the scenario where the features are uniformly distributed in the background population,
    the first case indicates that the collected features have no predictive value for the
    disease, whereas the second case points toward a simple predictive rule connecting the
    collected features to disease incidence.
\end{example}

We make the notion above precise via the notion of \emph{fooling distance}.
For a domain $\cX$ and a family $\cF$ of Boolean functions $f : \cX \to
\bits$, called the \emph{distinguishers}, the fooling distance between
probability distributions $P$ and $Q$ over $\cX$ with respect to $\cF$ is
\begin{equation}
    \label{eq:fooling-distance}
    \dist_\cF(P, Q) \define \sup_{f \in \cF} \abs{\E_P[f] - \E_Q[f]} \,.
\end{equation}
For comparison, we may see TV distance as the fooling distance with respect to the
family of \emph{all} Boolean functions. The measure \eqref{eq:fooling-distance}
is the \emph{integral probability metric} (IPM) instantiated using the bounded
family $\cF$. IPMs capture many other  natural notions of distance including
the Kolmogorov and $1$-Wasserstein metrics \cite{zolotarev1984probability,muller1997integral}.

We study distribution testing problems where fooling distance
takes the role traditionally played by TV distance, as in \cref{example:intro}. This newly defined problem guiding our studies is what we call \emph{identity testing against fooling distance}, or
\emph{$\cF$-identity testing} for short. The problem is defined as follows:

\begin{definition}[$\cF$-identity testing]
    \label{def:intro-identity-testing}
    Given a domain $\cX$, a family $\cF$ of Boolean functions $f : \cX \to
    \bits$, an explicit reference distribution $\refdist{P}$ over $\cX$, and
    parameters $\epsilon, \delta \in (0,1)$, algorithm $\cA$ is an
    \emph{$(\epsilon, \delta)$-identity tester to $\refdist{P}$ against $\cF$}
    if, given sample access to input probability distribution $P$ over $\cX$,
    with probability at least $1-\delta$ the following conditions hold:
    \begin{itemize}
        \item \textbf{(Completeness)} If $P=\refdist{P}$, then $\cA$ accepts.
        \item \textbf{(Soundness)} If $\dist_\cF(P,\refdist{P}) > \epsilon$,
        then $\cA$ rejects.
    \end{itemize}
    When the parameter $\delta$ is not important or can be inferred from
    context, we omit it from the notation and call $\cA$ an $\epsilon$-identity
    tester to $\refdist{P}$ against $\cF$. For brevity, we also call $\cA$ an
    $\cF$-identity tester (to $\refdist{P}$).
\end{definition}

We study both sample and time complexity bounds for such
algorithms\footnote{For the design of computationally-efficient algorithms, we
consider infinite sequences of testing problems over arbitrarily large finite
domains.}; some of our general results are sample efficient but not
computationally efficient, and for many concrete settings we give computationally efficient algorithms as well. 
We stress that while
\cref{def:intro-identity-testing} serves to guide our initial study,
fooling
distance has broader applicability to distribution testing problems. 
We also study a version of \emph{equivalence} (or closeness) testing
with respect to fooling distance, where both distributions $P$ and $\refdist{P}$ are unknown and only
accessible via samples. 
A natural intermediate goal, which
also arises in our applications,
is to have a single algorithm $\cA$ that is an identity
tester for \emph{any} distribution $\refdist{P}$ whose explicit description is given as an input.
More generally, many classic distribution testing problems may be studied under the lens of fooling
distance, in which case we call them $\cF$-distribution testing problems. 
While the $\cF$-distribution testing problems are natural by themselves, we show that $\cF$-distribution testing is an important unifying concept related to other areas of theoretical computer science such as testable learning \cite{rubinfeld2022testingdistributionalassumptionslearning}, PAC verification of learning algorithms \cite{GoldwasserRSY21}, and testing of structured distributions \cite{DKN14}. We show that $\cF$-distribution testing is a core idea in each of these models unlocking new results for each one. We now give a brief overview of our results:

\begin{enumerate}
    \item \textbf{Sample complexity of $\cF$-identity testing} (\cref{sec:intro-sample-complexity}).
        A distribution-free baseline follows immediately from uniform convergence:
        if family $\cF$ has VC dimension $d$, then for any distribution $P$, the empirical
        expectations $\hat{\ee}_P[f]$ over a sample of size $O(d/\epsilon^2)$ are good estimates
        of the population expectations $\ee_P[f]$ for all $f \in \cF$.
        However, since $\cF$-identity testing is defined with respect to
        a fixed reference distribution $\refdist{P}$, it is natural to ask for bounds that exploit
        this structure.  We show that a sharper, distribution-specific bound is achievable and
        necessary, by characterizing the sample complexity in terms of the Rademacher complexity of
        $\cF$ with respect to $\refdist{P}$, up to a quadratic gap between our upper and lower
        bounds.

\item \textbf{Testable learning and testable proper learning} (\cref{subsec:TL}). In many settings,
    efficient agnostic learning algorithms are known for a family $\cF$ under a specific ``nice''
    marginal distribution, but not for arbitrary marginals. The testable learning framework
    \cite{rubinfeld2022testingdistributionalassumptionslearning} enables the design of learning
    algorithms that succeed whenever the data marginal is nice, even if the marginal is unknown
    and niceness is hard to verify, by optionally declining to produce an output when the
    marginal is ``detectably not nice.'' We argue that $\cF$-identity testing is the natural
    testing concept underlying testable learning. We show that testable learners imply
    $\cF$-identity testers (efficiently), and conversely that $\cF$-identity testers imply
    testable \emph{query} learners. As an application, this yields testable \emph{proper}
    learning algorithms for halfspaces and decision trees, a capability not previously available
    without queries.

    \item \textbf{PAC verification} (\cref{sec:pac-verification}). In the PAC verification model
        \cite{GoldwasserRSY21,MutrejaS23,GurJKRSS24}, a computationally weak verifier $V$ interacts
        with a powerful but untrusted prover to agnostically learn a family $\cF$. The prover's role
        is to help the verifier find a good hypothesis using fewer labeled samples than it would
        need on its own, but the verifier must remain protected against a dishonest prover who might
        attempt to certify a bad hypothesis. We show a mild equivalence between $\cF$-identity
        testing and PAC verification: $\cF$-identity testers imply verification protocols with
        sample-efficient verifiers, and testable verification protocols
        imply $\cF$-identity testers
        for a restricted class of distributions. As corollaries, we obtain a Rademacher-based lower
        bound for testable verification and new distribution-free verification protocols for unions of
        intervals and axis-aligned rectangles.

    \item \textbf{Testing structured distributions} (\cref{sec:intro-structured}).
            As discussed, distribution testing with respect to TV distance is intractable in
            high-dimensional and continuous settings. Meanwhile, our results show that
            $\cF$-identity is often tractable if $\cF$ is a family of bounded expressivity.
            While $\cF$-identity testing is a syntactically weaker goal than TV testing,
            it is possible to identify circumstances where performing the former suffices to
            accomplish the latter. One such setting is when the input distribution comes from a
            \emph{structured class} of distributions, within which fooling distance against some set of distinguishers $\cF$
            already captures TV distance.
        This principle has been exploited in the $\cA_k$-testing framework of \cite{DKN14}. As a
        consequence of our broader systematization, we obtain new sample-efficient TV testers for
        structured classes of high-dimensional distributions (including low-degree polynomial and
        decision tree distributions over $[0,1]^n$), as well as computationally efficient
        testers for degree-$2$ polynomial and decision tree distributions over $\bits^n$.
\end{enumerate}

Next, we give a more detailed discussion of each class of results and overview the main proof ideas. We
finish the introduction with a discussion of interesting directions for
further investigation.

\subsection{VC and Rademacher bounds for the sample complexity of \texorpdfstring{$\cF$}{F}-identity testing}
\label{sec:intro-sample-complexity}

Uniform convergence implies that $O(d/\epsilon^2)$ samples suffice to
perform identity testing to any explicit distribution $\refdist{P}$ against a family $\cF$ of VC dimension
$d$; the algorithm simply estimates the expectations $\E_P[f]$ for each $f \in \cF$ using the
samples, and rejects if any estimate is too far from the corresponding value $\E_{\refdist{P}}[f]$. This observation has appeared before in the distribution testing
literature \cite{diakonikolas2023testingclosenessmultivariatedistributions}. Note that even this simple observation already has
intriguing conceptual implications: to test identity against polynomial-time computation or
polynomial-size circuits (for any fixed $\poly(n)$ bound),
polynomially (rather than exponentially) many samples suffice.

This result is distribution-free in the sense that it
holds for all reference distributions $\refdist{P}$ simultaneously. We show that
more refined bounds may be obtained in terms of the \emph{Rademacher complexity}
$\cR_m(\cF, \refdist{P})$, a
measure that accounts for both the family $\cF$ and the distribution
$\refdist{P}$. Specifically, Rademacher complexity upper bounds the sample complexity of
$\cF$-identity testing to $\refdist{P}$, and
moreover this bound is tight up to a quadratic gap. 

\begin{theorem}[Rademacher upper bound for $\cF$-identity testing]
    \label{thm:intro-testing-upper-bounds}
    Let $\cF$ be a family of $\cX \to \bits$ functions, $\refdist{P}$ be a
    reference distribution over $\cX$, and $m$ be a parameter such that
    $\cR_m(\cF, \refdist{P}) \le \epsilon/16$. Then there exists an $(\epsilon,
    \delta)$-identity tester to $\refdist{P}$ against $\cF$ with sample complexity
    $m + O\left(\frac{\log(1/\delta)}{\epsilon^2}\right)$.
\end{theorem}

\begin{restatable}[Rademacher lower bound for $\cF$-identity testing]{theorem}{thmrademacherlowerbound}
    \label{thm:rademacher-lower-bound}
    There exists an absolute constant $C > 0$ such that the following holds. Let
    $\cF$ be a family of $\cX \to \bits$ functions, $\refdist{P}$ be a reference
    distribution over $\cX$, and $m \ge C/\epsilon^2$ be a parameter such that
    $\cR_m(\cF, \refdist{P}) > 4\epsilon$. Then every $(\epsilon,
    0.99)$-identity tester to $\refdist{P}$ against $\cF$ has sample complexity
    $\Omega(\sqrt{m})$.
\end{restatable}

These results parallel similar bounds on the sample complexity of testable
learning in terms of Rademacher complexity by \cite{GollakotaKK23}\footnote{Indeed, our upper bound could also
be inferred via our reduction from $\cF$-identity testing to testable learning previewed in the next
subsection; for clarity, we give direct proofs of both upper and lower bounds.}\!\!, and reveal
another characterization of Rademacher complexity in terms of a natural statistical task. The idea
for the upper bound is that the Rademacher complexity of $\cF$ with respect to the unknown input
distribution $P$ can itself be estimated from samples --- if it is too high, the algorithm may
reject, and otherwise, empirical frequencies suffice to estimate $\dist_\cF(P, \refdist{P})$.
The lower bound holds because, informally, when $\cR_m(\cF, \refdist{P}) > \epsilon$, a random
nearly-balanced partition of a size-$m$ sample has non-trivial correlation with some $f \in
\cF$, which implies that the uniform distribution over one of the parts is far
from the distribution $\refdist{P}$ under $\dist_\cF$;
yet, by a birthday problem bound,
$o(\sqrt{m})$ samples cannot detect this.

\subsection{Testable learning}\label{subsec:TL}

The testable learning framework \cite{rubinfeld2022testingdistributionalassumptionslearning}
addresses the situation where an efficient agnostic learning algorithm is known for a family $\cF$
when the input marginal distribution is a ``good'' distribution $\cD_\cX$, but not for general
distributions. For example, halfspaces are efficiently agnostically learnable under Gaussian
marginals by polynomial regression \cite{kalai2008agnostically}, but the distribution-free case is
computationally intractable \cite{feldman2006new,guruswami2009hardness,daniely2016complexity}.
Testable learners enjoy the computational efficiency guarantees of distribution-specific learning,
but are required to produce a good output \emph{regardless} of what the true input
distribution is, in the following sense: either they successfully learn a good hypothesis (which in
particular they must do whenever the true marginal is indeed $\cD_\cX$), or they flag and abort
because the distributional assumption was violated. Crucially, this framework bypasses basic
impossibility results for directly \emph{testing} the distributional assumption (for example, it is
impossible to test Gaussianity of an arbitrary distribution with respect to TV distance). For
instance, the following positive results are known from the testable learning literature:

\begin{enumerate}
    \item Halfspaces (and, in fact, fixed functions of a constant number of halfspaces) can be
        testably learned with respect to the uniform distribution over $\bits^n$ or the standard
        Gaussian distribution over $\bbR^n$ in time $n^{\tilde{O}(1/\epsilon^2)}$ \cite{rubinfeld2022testingdistributionalassumptionslearning,GollakotaKK23}.
    \item Size-$s$ decision trees can be testably learned with respect to the uniform distribution
        over $\bits^n$ in time $n^{O(\log(s/\epsilon))}$ \cite{GollakotaKK23,KlivansSV24,GoelSSV24}.
\end{enumerate}

Let us formally define the model. Write the \emph{classification error} of hypothesis $h : \cX \to
\bits$ with respect to joint distribution $\cD$
as $\err(h,\cD) \define \Pr_{(x,y)\sim \cD}[h(x)\neq y]$, and write $\opt(\cF, \cD)
\define \inf_{h \in \cF} \err(h, \cD)$ for the \emph{optimal error} of any
hypothesis in $\cF$ with respect to the joint distribution $\cD$. Then, the
model is defined as follows. We extend the definition of
\cite{rubinfeld2022testingdistributionalassumptionslearning} to encompass testable \emph{query} learning and
testable \emph{semiagnostic} learning, which will feature in our applications.

\begin{definition}[Testable learning]
    \label{def:testable-learning}
    We say algorithm $\cA$ is an \emph{$(\epsilon, \delta)$-testable (agnostic) learner} for class
    $\cF$ with respect to marginal distribution $\cD_\cX$ on $\cX$ if, given sample access to input
    distribution $\cD$ on $\cX\times \bits$, algorithm $\cA$ either \emph{rejects} or outputs a
    hypothesis $h : \cX \to \bits$ such that, with probability at least $1-\delta$, the following
    properties hold.
    \begin{itemize}
        \item \textbf{(Soundness/composability)} If $\cA$ outputs a hypothesis $h$, then
            $\err(h,\cD) \le \opt(\cF,\cD) + \epsilon$.
        \item \textbf{(Completeness)} If $\cD$ truly has marginal $\cD_\cX$ on $\cX$, then $\cA$
            does not reject.
    \end{itemize}
    When the parameter $\delta$ is not important or can be inferred from
    context, we omit it from the notation and say $\cA$ is an
    $\epsilon$-testable learner.
    
    When the output hypothesis $h$ of $\cA$ is required to satisfy $h \in \cF$,
    we say $\cA$ is \emph{proper}.

    When the guarantee $\err(h,\cD) \le \opt(\cF,\cD) + \epsilon$ in the
    soundness condition is replaced with $\err(h,\cD) \le \beta \cdot
    \opt(\cF,\cD) + \epsilon$ for a parameter $\beta > 1$, we say that $\cA$ is
    \emph{$\beta$-semiagnostic}.
    
    When we assume that the input distribution $(x, y) \sim \cD$ is of the form
    $(x \sim P, y = f(x))$ for unknown marginal $P$ and
    labeling function $f : \cX \to \bits$, and moreover algorithm $\cA$
    is given \emph{query access} to the function $f$ at arbitrary points $x \in
    \cX$, we say that $\cA$ is a testable \emph{query} learner.
\end{definition}

Much of the success in the rapidly growing literature on testable learning \cite{rubinfeld2022testingdistributionalassumptionslearning,GollakotaKK23,diakonikolas2023efficient,gollakota2023efficient,KlivansSV24,slot2024testably,goel2025testing} has come from
harnessing analytic properties of distributions (such as Gaussian and uniform over $\bits^n$) and
hypothesis classes (such as halfspaces, intersections of halfspaces, decision trees, and small
formulas) to obtain tester-learner pairs that, in some form, combine two basic kinds of tasks:
1)~moment estimation, to check that the input distribution roughly matches the reference
distribution; and 2)~polynomial regression, to learn a good hypothesis when the moment-matching
condition is satisfied. The instantiation of this framework by~\cite{GollakotaKK23} crucially relied on the notion of \emph{fooling}.
 Informally, they gave a duality-based equivalence between
1)~low-degree moment matching between marginal distributions $P$ and $\cD_\cX$ being a sufficient
condition for ensuring that $\dist_\cF(P, \cD_\cX)$ is small; and 2)~the family $\cF$ admitting
low-degree \emph{sandwiching polynomials}, which is a stronger guarantee than $L_1$ polynomial
approximation that allows one to reason about the quality of polynomial regression when the true
marginal is $P$ rather than the reference marginal $\cD_\cX$.

\paragraph*{Testable learning implies $\cF$-identity testing.}
Our first result shows that the connection between testable learning and fooling can be made even
more concrete, and indeed independent of specific analytic considerations such as moment matching: given
a testable learner for family $\cF$ with respect to reference distribution $\cD_\cX = \refdist{P}$, we can construct an identity tester to $\refdist{P}$ against $\cF$ with the same sample and time complexity (up to
simulation overhead). That is, $\cF$-identity testing is no harder than testable learning:

\begin{theorem}[Testable learning implies $\cF$-identity testing]
    \label{thm:intro-testable-learning-implies-identity-testing}
    Suppose there exists an $\epsilon$-testable learner $\cA$ for family $\cF$ with respect to
    marginal $\cD_\cX = \refdist{P}$. Then there exists an $\epsilon$-identity tester to $\refdist{P}$ against $\cF$
    with the same sample complexity as $\cA$ (plus $O(1/\epsilon^2)$ additional samples), and
    with running time polynomial in the running time of $\cA$.
\end{theorem}

Notably, this reduction yields \emph{computationally efficient} $\cF$-identity testers for all the settings
in which we have testable learners. In comparison, our VC and Rademacher-based sample complexity
upper bounds described in \cref{sec:intro-sample-complexity} are information-theoretic and do not
yield computationally efficient algorithms.

The idea behind the proof of \cref{thm:intro-testable-learning-implies-identity-testing} is to
create a mixture joint distribution $M$ which with probability $1/2$ draws a sample from the input
marginal $P$ labeled $1$, and with probability $1/2$ draws a sample from the reference
marginal $\refdist{P}$ labeled $0$, and feed samples from $M$ as input to the testable learner $\cA$. The key observation is that
\begin{align*}
    \opt(\cF,M)=\frac{1}{2}-\frac{1}{2}\dist_{\cF}(P,\refdist{P}),
\end{align*}
since the best $f\in\cF$ for separating the two halves of $M$ is precisely the function maximizing the fooling distance between the distributions.
If
$\cA$ rejects, this means that $P \ne \refdist{P}$ and the tester may reject. Otherwise, $\cA$ produces some hypothesis $h$. Then, since
%one can verify that
$\dist_\cF(P, \refdist{P})$ is large if and only if the optimal classification error $\opt(\cF, M)$ is
bounded away from $1/2$, the tester may accept or reject depending on the classification error of the
hypothesis $h$.

\paragraph*{$\cF$-identity testing implies testable query learning.}
As a partial converse to the result above, we show that the existence of an
identity tester to $\refdist{P}$
against $\cF$, along with the existence of an efficient distribution-specific agnostic learner for
$\cF$ under marginal $\cD_\cX = \refdist{P}$, implies the existence of a testable \emph{query} learner for
$\cF$ with respect to marginal $\cD_\cX = \refdist{P}$. %Intuitively, the reason for requiring queries in this
%reduction is that, by using queries, the learning algorithm may draw its own labeled samples of the
The reason our reduction uses queries is so that the learning algorithm can draw its own labeled samples of the
form $(x \sim \refdist{P}, y = f(x))$. Hence, the task of producing a good hypothesis for the input
distribution $(x \sim P, y = f(x))$ reduces to a question of \emph{domain adaptation} \cite{mansour2009domain,chandrasekaran2024efficient}:
does a hypothesis obtained by training on $\refdist{P}$ generalize to the marginal $P$?

Concretely, our result is the following:

\begin{theorem}[$\cF$-identity testing implies testable query learning; informal]
    \label{thm:intro-testable-query-learning-from-identity-testing}
    Suppose there is an identity tester $\cT$ to reference distribution $\refdist{P}$ against family $\cF$,
    and an agnostic query learner $\cL$ for family $\cF$ under marginal $\cD_\cX = \refdist{P}$. Then there
    exists a testable query learner $\cA$ for $\cF$ with respect to marginal $\cD_\cX = \refdist{P}$ using the
    combined sample and time complexities of $\cT$ and $\cL$. Moreover, $\cA$ is proper if $\cL$ is.
\end{theorem}

As suggested, the proof relies on a domain adaptation guarantee: if we could use the
$\cF$-identity tester $\cT$ to ensure that a near-optimal hypothesis learned by a distribution-specific
agnostic learner on the reference marginal $\refdist{P}$ will perform similarly to a near-optimal hypothesis
for the true marginal $P$, then we could simulate the agnostic learner on marginal $\refdist{P}$ by drawing
unlabeled samples from $\refdist{P}$ and issuing queries to $f$, and return the learned hypothesis. On
the other hand, if the $\cF$-identity tester $\cT$ rejects, then we know that $P \ne \refdist{P}$ and the testable
learning model allows us to reject.

The key observation is that the domain adaptation guarantee above holds as long as $P$ and $\refdist{P}$ are close
under fooling distance with respect to the \emph{related} family $f \oplus \cF \define \{f \oplus g
: g \in \cF\}$, where $\oplus$ denotes the XOR operation. This is true because the classification
error of a hypothesis $h$ under marginal $D$ is precisely $\err_D(h, f) = \E_D[f \oplus h]$, so when
$\dist_{f \oplus \cF}(P, \refdist{P})$ is small, the classification error of each hypothesis $h \in \cF$ is
similar under marginals $P$ and $\refdist{P}$. To complete the picture, our key lemma
(\cref{lemma:implicit-tester-reduction})
shows that, in the query
setting, there is a black-box reduction from identity testing against the class $f \oplus \cF$ to
identity testing against the class $\cF$ itself; the reduction proceeds by a careful analysis of
certain synthetic distributions obtained by mixing $P$ and $\refdist{P}$ according to the label values
assigned by $f$.

\paragraph*{Connection to TDS learning.} A related insight on the role of fooling distance metrics for
domain adaptation appears in recent work on \emph{testable learning with distribution
shift} (TDS learning) \cite{KlivansSV24,chandrasekaran2024efficient}, where the goal is to learn a hypothesis using labeled samples from
a joint distribution $\cD_{\text{train}}$, and use unlabeled samples from another marginal
distribution $\cD_{\text{test}}$ to either \emph{reject} (when the marginals differ ``too much''), or
produce a hypothesis $h$ that performs well against the test distribution (up to an
information-theoretic barrier that depends on the train and test distributions). In \cite{chandrasekaran2024efficient}, the
authors defined the notion of \emph{localized discrepancy} between two distributions. That
definition is fairly general but, in a natural specialization that appears in
their applications,
it precisely matches the fooling distance $\dist_{f \oplus \cF}$. Building on the
``analytic'' theme discussed above, the authors showed that, by estimating the Chow parameters of a
learned hypothesis $h$ under the marginals $P$ and $\refdist{P}$, one can test whether $\dist_{h \oplus
\cF}(P, \refdist{P})$ is small. This suffices for TDS learning, which is conceptually related to, but a
different problem from testable learning. In contrast, here we give a purely combinatorial result
showing that, given a tester against $\dist_\cF$, we can obtain a tester against $\dist_{f \oplus
\cF}$ using queries to $f$, which in turn suffices for testable query learning.

\paragraph*{Applications to testable proper learning.}
Prior work on testable learning relied on polynomial regression \cite{kalai2008agnostically} to produce a
hypothesis $h$, which means that in general we have $h \not\in \cF$ --- the learning algorithm is
\emph{improper}. Since our upper bound for testable query learning in
\cref{thm:intro-testable-query-learning-from-identity-testing} uses any agnostic learning algorithm as a
black box, we may use recent advances on agnostic \emph{proper} learning (using samples or queries)
to obtain testable proper query learners. Moreover, thanks to the reduction in the \emph{opposite}
direction from \cref{thm:intro-testable-learning-implies-identity-testing}, we may fulfill the
requirement that an $\cF$-identity tester exist by leveraging an existing \emph{improper} testable learner.
Combining these observations allows us to prove results in the following spirit:
\begin{quote}
    \centering
    \emph{In testable learning, queries help to afford properness.}
\end{quote}

We give two applications of this principle. First, using recent advances on agnostic proper learning
of halfspaces \cite{DiakonikolasKKT24} under the Gaussian distribution, together with
testable learning results for halfspaces with respect to the Gaussian distribution
\cite{GollakotaKK23}, we obtain a testable proper query learner for halfspaces with respect to the
Gaussian distribution over $\bbR^n$ with running time $n^{\tilde O(1/\epsilon^2)} +
2^{\poly(1/\epsilon)}$, matching the dependence on $n$ of standard testable learning.

Second, using a result of \cite{BlancLQT22} on agnostic proper learning of size-$s$ decision trees
with respect to the uniform distribution over $\bits^n$ using membership queries, along with known
results for testable learning of size-$s$ decision trees (which may be inferred from
\cite{GollakotaKK23} and \cite{KlivansSV24}, or as directly stated in
\cite{GoelSSV24}), we
obtain a testable proper query learner for size-$s$ decision trees with respect to the uniform
distribution over $\bits^n$ with running time
$(n/\epsilon)^{O(\log(s/\epsilon))}$, essentially matching the known bound
for standard testable learning.

\paragraph*{Testable semiagnostic proper learning using samples.}
It is natural to ask whether query access, which is a somewhat onerous condition, is a hard
requirement for the proper learning results stated above. We observe that, in the sample-based
setting, we may still use the insights above to obtain testable \emph{semiagnostic} proper learners.
Specifically, we show that bounds such as
those described above may still be attained using labeled
samples for the task of testable $3$-semiagnostic proper learning, i.e., with error guarantee of the
form $3 \cdot \opt +\ \epsilon$. The proof essentially proceeds by first learning a hypothesis $h$ using a standard
(improper) testable learner, and then applying the query-based results using query access to $h$; the factor of $3$ comes from the triangle inequality.

\paragraph*{Recent work on the limitation of queries in testable learning.} In independent
concurrent work, \cite{lange2026limitations} showed that testable learning with membership queries
cannot have smaller time complexity than sample-based distribution-specific agnostic learning. That
is, the time complexity of agnostic learning places a barrier to how much queries may speed up
testable learning. In contrast, our results above show that queries can help in a different way:
rather than beating the time complexity of sample-based testable learning, we may obtain a proper
testable learner from an improper one, as long as we also have access to a
distribution-specific agnostic proper learner.

\subsection{PAC verification}\label{sec:pac-verification}

In the PAC verification model \cite{GoldwasserRSY21,MutrejaS23,GurJKRSS24}, a computationally weak verifier $V$ wishes to agnostically learn a family $\cF$ with assistance from
powerful but untrusted prover $P$. The goal is to design a protocol
in which the verifier $V$ uses fewer labeled samples than would be required for it to agnostically
learn $\cF$ by itself (i.e., without the prover's help). This model may be defined in the
distribution-specific setting, where the underlying marginal distribution is assumed to be a fixed
distribution $\cD_\cX$ (or to come from a fixed class $\mathbb{D}$ of distributions), or in the
distribution-free sense where the marginal distribution is allowed to be arbitrary.
In analogy with agnostic learning, there is also a natural \emph{testable} setting that sits
between distribution-free and distribution-specific, following
\cite{rubinfeld2022testingdistributionalassumptionslearning}.
We now present the formal definition of PAC verification, suitably augmented with its
testable setting: 

\begin{definition}[PAC Verification of Hypothesis Class]\label{def:ver}
    Let $\cX$ be a set, let $\mathbb{D}_\comp, \mathbb{D}_\sound
    \subseteq\Delta(\cX\times\{0,1\})$ be classes of distributions and let
    $\cF\subseteq\{0,1\}^{\cX}$ be a class of hypotheses. We say $\cF$ is PAC verifiable with
    respect to $(\mathbb{D}_\comp, \mathbb{D}_\sound)$ if there exists an interactive proof
    system consisting of a verifier $V$ and an honest prover $P$ and parameters
    $m_V,m_P\in\N$ such that the following conditions are satisfied:
\begin{itemize}
    \item (Completeness) For all $\cD \in \mathbb{D}_\comp$, the following holds. Let $S_V
        \sim \cD^{m_V}$ be the sample drawn by $V$, let $S_P \sim \cD^{m_P}$ be the sample drawn by $P$,
        and let the random variable
        \[
            h_V=[V(S_V,\eps,\delta),P(S_P,\eps,\delta)]\in\cF\cup\{\text{reject}\}
        \]
        denote the output of $V$ after receiving input $(S_V,\eps,\delta)$ and interacting with $P$,
        which received input $(S_P,\eps,\delta)$. Then, with probability at least $1-\delta$ over
        $S_V$, $S_P$, and the internal randomness of $V$ and $P$,
        \[
            (h_V\neq \text{reject})\land \left(\err(h_V,\cD)\leq \opt(\cF,\cD)+\eps\right).
        \]

    \item \sloppy (Soundness) For all $\cD \in \mathbb{D}_\sound$, and for any (possibly malicious and
        computationally unbounded) prover $P'$ (which may depend on $\cD,\eps,\delta$), the verifier
        outputs $h_V=[V(S_V,\eps,\delta),P']$ such that, with probability at least $1-\delta$ over
        $S_V \sim \cD^{m_V}$ and the internal randomness of $V$ and $P'$,
        \[
            (h_V= \text{reject})\lor \left(\err(h_V,\cD)\leq \opt(\cF,\cD)+\eps\right).
        \]
\end{itemize} 
    When $\mathbb{D}_\comp = \mathbb{D}_\sound = \Delta(\cX\times \{0,1\})$, we refer to
    this setting as distribution-free verification. When $\mathbb{D}_\comp = \mathbb{D}_\sound =
    \{\cD \in \Delta(\cX \times \bits) : \cD \text{ has marginal $\cD_\cX$ on $\cX$}\}$
    for some fixed marginal $\cD_\cX$, we
    refer to this setting as distribution-specific verification. When $\mathbb{D}_\comp = \{\cD \in
    \Delta(\cX \times \bits) : \cD \text{ has marginal $\cD_\cX$ on $\cX$}\}$
    for some fixed marginal $\cD_\cX$, but
    $\mathbb{D}_\sound = \Delta(\cX \times \bits)$, we refer to this setting as \emph{testable}
    verification (with respect to marginal $\cD_\cX$).
\end{definition}

In this paper, we actually consider a simplification of \cref{def:ver} whereby the
honest prover receives an explicit description of the input distribution, rather than just samples
from it. This simplification is benign for us because we only focus on verifier sample
complexity, and do not attempt to optimize prover or communication cost in our general reductions.
We suspect that our concrete applications should enjoy doubly-efficient protocols. (See
\cref{rem:verification-efficiency}.)

We show that PAC verification is closely related to $\cF$-identity testing via a mild equivalence between the
two problems. Specifically, the equivalence hinges on identity testing against the ``lifted class''
$\tilde{\cF}$ over the augmented domain $\cX \times \bits$, given by
$\tilde{\cF}\;:=\;\bigl\{\;\tilde{f}(x,y):=\Ind[f(x)\neq y] :\ f\in \cF\;\bigr\}$. The
intuition behind this definition is that it captures error rates: for every joint probability
distribution $\cD$ over $\cX\times \bits$ and function $f \in \cF$, we have $\E_{(x,y) \sim
\cD}[\tilde{f}(x,y)]=\err(f,\cD)$.

\paragraph*{$\cF$-identity testing implies distribution-free verification.} We first show that
algorithms for $\tilde{\cF}$-identity testing imply distribution-free verification protocols with
sample-efficient verifiers\footnote{We focus on understanding verifier sample complexity, and do not
attempt to construct protocols with efficient communication. See
\cref{rem:verification-efficiency}.}\!\!. The prover receives an explicit description 
and the verifier gets sample access to the input distribution $\cD\in\Delta(\cX\times\bits)$.
The protocol proceeds by the honest prover computing and
sending an approximately optimal hypothesis $h^*\in \cF$ along with an explicit distribution
$\cD'=\cD$ to the verifier. The verifier, using sample access to the distribution $\cD$,
can check if $\cD'$ and $\cD$ are close in fooling distance with respect to the class of
distinguishers $\tilde{\cF}$, where closeness would imply that the error of all functions in the
function class $\cF$ are close with respect to distributions $\cD'$ and $\cD$. If they are far
away, the verifier rejects, otherwise the verifier checks if $h^*$ is indeed the optimal hypothesis
with respect to $\cD'$ (without suffering the uniform convergence bound). In case the prover is
honest, the verifier successfully outputs $h^*$, and if the prover is dishonest, the verifier either
outputs ``reject'' or outputs a hypothesis which is near-optimal, as required. 

Moreover, a ``lifting'' lemma shows that $\tilde{\cF}$-identity testing in fact
reduces to the more intuitive problem of $\cF$-identity testing (by carefully working with conditional
probabilities to relate $\dist_{\tilde{\cF}}$ to $\dist_\cF$). Therefore, we conclude that
$\cF$-identity testing yields distribution-free verification:

\begin{theorem}[Informal]
    \label{thm:informal-identity-testing-implies-verification}
    Suppose there is an algorithm $\cT$ that, given as input the explicit description of any
    probability distribution $\refdist{P}$ over $\cX$, performs identity testing to $\refdist{P}$ against $\cF$. Then
    there exists an interactive proof system in the distribution-free setting which verifies $\cF$ in which the verifier $V$ has the
    same sample complexity as $\cT$.
\end{theorem}

A quick corollary of this result is that, by leveraging the identity tester of \cite{DKN14} against the
class $\cA_k$ of unions of $k$ intervals on domain $[0,1]$, we obtain a distribution-free
verification protocol for
unions of $k$ intervals with the optimal verifier sample complexity $O(\sqrt{k}/\epsilon^2)$,
slightly improving upon the bound of $O(\sqrt{k}/\epsilon^{2.5})$ shown by \cite{MutrejaS23}. Moving beyond the
one-dimensional regime, we can also leverage recent results of \cite{diakonikolas2023testingclosenessmultivariatedistributions} for equivalence testing
against the \emph{multidimensional} class $\cA_k$ --- which corresponds to unions of $k$ disjoint
rectangles in $[0,1]^n$ --- to obtain a distribution-free verification protocol for disjoint unions of
multidimensional rectangles with sample-efficient verifier (i.e., smaller sample complexity than
required for learning) for every constant dimension $n$:

\begin{restatable}{theorem}{thmrect}
\label{thm:rectangles}
    There exists a distribution-free verification protocol for the family of unions of $k$ disjoint rectangles in
    $[0,1]^n$ with verifier sample complexity $k^{6/7} \log^n(k) \poly_n(1/\epsilon)$, where the
    dependence on $\epsilon$ is a large polynomial that depends only on $n$.
\end{restatable}

\paragraph*{Testable verification implies $\tilde{\cF}$-identity testing}
We show a partial converse to the implication above: given a testable verification protocol for family
$\cF$ with respect to a fixed marginal distribution $\cD_\cX$, we may obtain an identity tester
for the distribution $\refdist{\cD} \define \cD_{\cX} \times \Unif(\bits)$ over $\cX \times \bits$ against the
lifted family $\tilde{\cF}$, with essentially the same complexity as the verifier. 

The key idea is to use the verifier as a testing primitive by simulating a prover that has
access only to the reference distribution $\refdist{\cD}$. Recall that
$\dist_{\tilde{\cF}}(\cD, \refdist{\cD}) > \epsilon$ means some $f \in \cF$ achieves
$|\err(f, \cD) - 1/2| > \epsilon$, i.e., the labels in $\cD$ are not close to uniform.
The identity tester runs the verification protocol with a prover instantiated on
$\refdist{\cD}$, while the verifier receives sample access to the unknown $\cD$. If
$\cD = \refdist{\cD}$, labels are uniform and every $f \in \cF$ has error close to
$1/2$; moreover, in this case our prover is honest and the marginal is indeed
$\cD_\cX$, so the verifier accepts a hypothesis with error near $1/2$. If,
instead, $\cD$ is far
from $\refdist{\cD}$ in fooling distance, then $\opt(\cF, \cD) < 1/2 - \epsilon$, and by soundness, either the verifier rejects or any
accepted hypothesis must have error bounded away from $1/2$. The identity tester therefore
accepts or rejects based on whether the output hypothesis has error close to $1/2$.

While this
``lifted'' identity tester may not have a natural interpretation in terms of $\cF$-identity testing, it is
easy to show that the Rademacher complexity of $\tilde{\cF}$ over
$\refdist{\cD}$ is the same as that of
$\cF$ over $\cD_{\cX}$. Combining this with our sample complexity lower bound for $\cF$-identity testing from
\cref{thm:rademacher-lower-bound}, we conclude a lower bound for testable
verification:

\begin{restatable}{theorem}{thmintroverificationlowerbound}
    \label{thm:intro-verification-lower-bound}
    There exists an absolute constant $C > 0$ such that the following holds for
    all sufficiently small constants $\delta \in (0,1)$. Let $\cF$ be a family
    of $\cX \to \bits$ functions, $\cD_\cX$ be a reference marginal
    distribution, and $m \ge C/\epsilon^4$ be a parameter such that $\cR_m(\cF,
    \cD_\cX) > 4\epsilon$. Then every testable verification protocol for $\cF$ to error
    $\epsilon$ and confidence $1-\delta$ with respect to marginal $\cD_\cX$ has
    verifier sample complexity $\Omega(\sqrt{m})$.
\end{restatable}

This lower bound complements a similar $\Omega(\sqrt{d}/\epsilon^2)$ lower bound by \cite{MutrejaS23} for
PAC verification in the distribution-free setting, where $d$ is the VC dimension of family $\cF$.

\subsection{Testing structured distributions}
\label{sec:intro-structured}
Given a structured class of distributions $\bbD$, a natural question is whether the TV distance between any two distributions in the class is effectively captured by fooling distance with respect to a function class $\cF$ of ``low complexity.'' If so, $\cF$-identity testing, which is tractable for low complexity $\cF$, would suffice for \emph{TV distance} testing within the class of distributions $\bbD$, thus bypassing the general impossibility of TV testing in high-dimensional and continuous settings. A series of works \cite{DKN14,DiakonikolasKN15,DiakonikolasKN17,DiakonikolasKP19,diakonikolas2023testingclosenessmultivariatedistributions} in the literature under the name of \emph{$\cA_k$-testing} has exploited this connection --- instantiated against the family $\cA_k$ of unions of $k$
disjoint rectangles (or intervals, in one dimension)  --- to obtain efficient identity and
equivalence testers for distributions whose density functions are, for example, piecewise low-degree
polynomials or log-concave, in one-dimension \cite{DKN14}; or piecewise-uniform over unions of disjoint
rectangles, in higher dimensions \cite{diakonikolas2023testingclosenessmultivariatedistributions}. These works focused on obtaining testers with $o(k)$
sample complexity, i.e., fewer samples than required for uniform convergence. We give further results in this spirit, both extending the range of structured classes that admit sample-efficient TV testers, and obtaining computationally efficient testers via our connections to the testable learning framework. We now formally define the problem of testing structured distributions:
\begin{definition}
\label{def:identity-testing-structured-intro}
Let $\bbD \subseteq\Delta(\cX)$ be a (possibly infinite) class
of distributions defined over a (possibly infinite) domain $\cX$. We say algorithm
$\cA$ is an $(\epsilon, \delta)$-identity tester for distribution class
$\bbD$ if,
given an explicit description of any reference distribution $\refdist{P}\in \bbD$
and sample access to a distribution $\tgtdist{P}\in \bbD$, with probability
at least $1-\delta$ the following conditions hold:
\begin{itemize}
    \item \textbf{(Completeness)} If $\refdist{P}=\tgtdist{P}$, then
    $\cA$ accepts.
    \item \textbf{(Soundness)} If $\dist_\tv(\refdist{P},\tgtdist{P})>\eps$,
    then $\cA$ rejects.
\end{itemize}
We omit the parameter $\delta$ when it is not important or can be inferred from context.
\end{definition}

As already mentioned,
algorithms for $\cF$-identity testing may be used to
satisfy the requirements of \cref{def:identity-testing-structured-intro},
specifically when the structural properties of $\bbD$ guarantee that $\dist_\cF(P, \refdist{P}) \approx
\dist_\varvarTV(P, \refdist{P})$ for all $P, \refdist{P} \in \bbD$. 

The key concept is that of \emph{Scheff\'{e}} sets (also known as \emph{Yatracos} sets which were
inspired by \cite{scheffe1947useful}, implicitly defined in \cite{yatracos1985rates}, and formally
defined by \cite{devroye2001combinatorial}): for distributions $P,\refdist{P}$ with densities
$p,\refdist{p}$ respectively, the Scheff\'{e} set is defined by
$S_{P,\refdist{P}}=\{x:p(x)>\refdist{p}(x)\}$. It is easy to see that fooling distance -- with the
class of distinguishers being the indicator of whether a point is present in the Scheff\'{e} set --
equates to the total variation distance. (The Scheff\'{e} set approach to bounding the total
variation distance via a restricted class of distinguishers originates in the density estimation
literature
\cite{DBLP:conf/colt/MahalanabisS08,DBLP:conf/colt/BousquetKM19,DBLP:conf/icml/AamandACINS23}, where it underlies the classical minimum distance estimator and modern refinements.)
Therefore, to perform TV testing within class $\bbD$, it suffices to
include the indicator of the Scheff\'{e} set between every pair of distributions in $\bbD$ in the family $\cF$ of distinguishers. When the
set $\cF$ obtained in this manner has bounded VC dimension, sample complexity
upper bounds immediately follow.
In particular, distributions whose density
functions are low-degree polynomials, as well as  size-$s$ decision
tree distributions, admit sample-efficient identity testers against TV distance
over $[0,1]^n$. (Recall that, without structural guarantees, TV testing over
continuous domains is impossible.) However, the resulting testers are not computationally efficient.

By further leveraging our reduction from $\cF$-identity testing to testable learning from
\cref{thm:intro-testable-learning-implies-identity-testing}, we give \emph{computationally
efficient} uniformity and identity testers for structured classes: 
\begin{enumerate}
    \item For the class of distributions over $\bits^n$ whose probability mass functions are
        degree-$2$ multivariate polynomials, there is a \emph{uniformity} tester (against TV
        distance) with sample and time complexity at most that of testably learning the class of
        degree-$2$ polynomial threshold functions  with respect to the uniform distribution over
        $\bits^n$, which is $n^{\tilde O(1/\epsilon^9)}$ \cite{GollakotaKK23}.
    \item For the class of distributions over $\bits^n$ that are piecewise-uniform over the leaves
        of some size-$s$ decision tree over $\bits^n$ (i.e., decision tree distributions \cite{FeldmanOS08,BlancLMT23}), there is
        an \emph{identity}
        tester (against TV distance) with sample and time complexity at most that of
        testably learning the class of size-$s^2$ decision trees with respect to the uniform
        distribution over $\bits^n$, which is $n^{O(\log(s/\epsilon))}$
        \cite{GollakotaKK23,KlivansSV24,GoelSSV24}.
\end{enumerate}

The first result works by a reduction to testable learning; while we could unpack the
analytic techniques therein to express the tester in terms estimating low-degree moments, we find it
compelling that the reduction works in a combinatorial, black-box way to connect the seemingly very
different models of testable learning and testing of structured distributions. The
second result for \emph{identity} testing performs extra work on top of a testable learning
reduction, and thus carries distribution testing content above and beyond the analytic machinery of
testable learning.

\subsection{Discussion}

In this paper, we show that distribution testing against bounded distinguishers,
operationalized by the fooling distance, is a natural statistical task that
1)~makes feasible an intuitively interpretable version of distribution testing
in high-dimensional and continuous settings; and 2)~connects previously
unrelated areas of property testing and learning theory, uncovering new results
in those areas.

We now mention a few more lines of research related to our work, and then conclude
with open questions for future research.

\paragraph*{Fooling distance in the wild.} Fooling distance, as defined here, is
a natural concept in computer science that appears in many contexts. Formally, it is an
integral probability metric (IPM), a family of metrics with broad applications in statistics
and machine learning (see e.g.,
\cite{sriperumbudur2012empirical,arjovsky2017wasserstein,amit2022integral,kong2023covariate}).
We favor our terminology due to the prominence of the concept of fooling in
computer science, not only, e.g., in testable learning but also in pseudorandomness
and its applications \cite{trevisan2009regularity,vadhan2012pseudorandomness}.
Similar distance metrics also appear in the settings of
multicalibration and multiaccuracy
\cite{hebert2018multicalibration,gopalan2021omnipredictors,dwork2023pseudorandomness,casacuberta2024complexity,casacuberta2025global}, which are concerned with notions
of multi-group fairness measured via a bounded family $\cF$ representing the
subpopulations of concern, on each of which the learned hypothesis $h$ should
be a good approximator.

In fact, multicalibration proved to be a key technical ingredient in recent work of
\cite{marcussen2025characterizing}, who studied the computational aspects of the problem of
distinguishing from which of two distributions, say $P$ or $Q$, a sequence of $m$ independent
samples was drawn. Information-theoretically, the best achievable advantage is $\dist_\tv(P^m,
Q^m)$, which is controlled by the Hellinger distance between $P$ and $Q$. In
\cite{marcussen2025characterizing}, the authors asked about the \emph{computational}
distinguishability of $P$ and $Q$ by a bounded class $\cF$, with emphasis on the class of small
circuits. They showed that the distinguishability of $P$ and $Q$ by $\cF$ is tightly characterized
by the information-theoretic distinguishability of related distributions $\tilde{P}$ and
$\tilde{Q}$, such that $P, \tilde{P}$ are indistinguishable by $\cF$ and similarly for
$Q, \tilde{Q}$. While the setting of \cite{marcussen2025characterizing} is related to ours, and it
would be interesting to study concrete connections between the two works, we also highlight
two important conceptual differences:
\begin{enumerate}
    \item At a surface level, \cite{marcussen2025characterizing} studied the problem of
        distinguishing between two fixed distributions $P$ and $Q$, whereas we study the identity
        testing problem of distinguishing between a fixed reference distribution $\refdist{P}$, and
        any distribution that is far from it (in fooling distance).
    \item In \cite{marcussen2025characterizing}, the bounded family of distinguishers $\cF$ is part
        of the computational modeling of the problem at hand: informally, how well can we
        distinguish between $P$ and $Q$ using a fixed number of samples and computational budget? In
        contrast, in our work the bounded family $\cF$ is part of the specification of the
        distinguishing problem: the task is to tell apart between $\refdist{P}$ and any distribution
        $Q$ with $\dist_\cF(\refdist{P}, Q) > \epsilon$, and one can ask about the statistical and
        computational aspects of this task --- how many samples are needed, and how efficiently the
        computation can be performed? (We give both types of results.)
\end{enumerate}

\paragraph*{Distribution testing with alternative metrics.}
Metrics other than TV distance have appeared in the property testing
and related literature --- 
either incidentally to the analysis of algorithms (e.g., $\ell_2$ distance
and $\chi^2$ divergences, as described in the survey \cite{Canonne22}),
or as targets of direct study (see e.g., \cite{DaskalakisKW18}).
Prior works have studied distribution testing against
the \emph{earthmover distance} \cite{IT03,BaNNR11,FH24}, which is an
IPM that also allows for a form of testing over continuous
settings --- but offers qualitatively
different guarantees that incorporate geometric content of the domain.

\paragraph*{Future work.}
We identify two directions for further research. First, we ask for a \emph{tight
characterization} of the sample complexity of $\cF$-identity testing, improving upon the quadratic gap in our
Rademacher characterization. One challenge we note is that,
in the limit where $\cF$
is the family of all Boolean functions and $\cF$-identity testing becomes standard
identity testing, such a characterization would need to encompass known
characterizations of the sample complexity of identity testing in terms of
subtle $\ell_{2/3}$-norms or interpolation norms of the reference distribution
\cite{valiant2017automatic,BlaisCG19}.

More broadly, we observe that three related statistical tasks are currently all
characterized by Rademacher complexity up to a quadratic gap: $\cF$-identity testing (via
\cref{thm:intro-testing-upper-bounds,thm:rademacher-lower-bound}), testable learning (by
\cite{GollakotaKK23}), and testable verification (the upper bound by reducing to
testable learning, and the lower bound by
\cref{thm:intro-verification-lower-bound}). Moreover,
$\cF$-identity testing is no harder than testable learning
by \cref{thm:intro-testable-learning-implies-identity-testing}
(while it can be strictly easier\footnote{Unions of $k$ intervals witness a
maximal quadratic separation: identity testing against the family $\cA_k$ can be done using $O(\sqrt{k})$
samples \cite{DKN14}, but learning this class requires $\Theta(k)$ samples.}), and both bounds for $\cF$-identity testing
(\cref{thm:intro-testing-upper-bounds,thm:rademacher-lower-bound}) can be nearly tight for some classes (see
\cref{rem:rademacher-tight}).

Second, it would be interesting to characterize which distributions
$\refdist{P}$ and families $\cF$ admit
\emph{computationally efficient} identity testers; our VC and Rademacher bounds are information-theoretic, and although we
give computationally efficient bounds for specific applications, we do not have a general
characterization. One promising avenue might be to adapt the notion of \emph{refutation complexity}
\cite{vadhan2017learning,kothari2018improper}, which takes the role of Rademacher complexity for characterizing efficient agnostic
learnability, seeking a complexity measure that captures efficient distribution testing against fooling
distance.

\paragraph*{Organization.}
The rest of this paper is organized as follows. \cref{sec:prelim} introduces notation and
conventions used throughout the paper. \cref{sec:sample-complexity} gives our general sample
complexity bounds for $\cF$-identity testing. \cref{sec:testable-learning} gives our results relating
$\cF$-identity testing and testable learning. \cref{sec:verification} gives our results relating
$\cF$-identity testing and PAC verification. Finally, \cref{sec:structured} gives our results for TV testing
of structured distributions.

\section{Preliminaries}
\label{sec:prelim}
Throughout the paper, we write $\cX$ for a domain and $\cF \subseteq \zo^\cX$
for a family of Boolean functions over $\cX$. For simplicity, we assume
throughout the paper that the family $\cF$ is closed under complementation, that
is, that for each $f \in \cF$ it holds that $1-f \in \cF$.

We write $\Delta(\cX)$ for the set of probability distributions over $\cX$. We
typically denote reference distributions for identity testing with notation such
as $\refdist{P} \in \Delta(\cX)$, and write $P$ for the unknown input
distribution over $\cX$. In learning settings where the domain of interest is $\cX \times
\bits$, we often denote a joint distribution over $\cX \times \bits$ by $\cD$ and a marginal
distribution over $\cX$ by $\cD_\cX$.

Typically,
parameter $\epsilon$ denotes an error or distance parameter, and $\delta$
denotes a failure probability parameter. For a set $S$, we write $\Unif(S)$ for
the uniform distribution over $S$.

We write $\err(h,\cD) \define \Pr_{(x,y)\sim \cD}[h(x)\neq y]$ for the
\emph{classification error} of hypothesis $h : \cX \to \bits$ with respect to
joint distribution $\cD$ over $\cX \times \bits$, and we write $\opt(\cF, \cD)
\define \inf_{h \in \cF} \err(h, \cD)$ for the \emph{optimal error} of any
hypothesis in $\cF$ with respect to the joint distribution $\cD$.

Given a marginal distribution $P$ over $\cX$ and a Boolean function $f : \cX \to
\bits$, we write $P_f$ for the joint distribution over $\cX \times \bits$ whose
samples $(x,y) \sim P_f$ are distributed as $(x \sim P, y = f(x))$. Sometimes we
also place $f$ as a superscript and write $P^f$ for this distribution to avoid
cumbersome notation, e.g.\ $\refdist{P}^f = \left(\refdist{P}\right)_f$.

\subsection{Rademacher complexity}
Given a Boolean function $f : \cX \to \zo$, write $f^{\pm} : \cX \to \sbits$
for its corresponding signed function given by $f^{\pm}(x) \define 2f(x) - 1$.

\begin{definition}[Rademacher complexity]
    \label{def:rademacher-complexity}
    Let $m \in \bbN$ and let $S = (x_1, \dotsc, x_m) \in \cX^m$ be a tuple. The
    \emph{empirical Rademacher complexity} of family $\cF$ over $S$ is
    \begin{equation}
        \label{eq:empirical-rademacher}
        \cR(\cF, S) \define \argEx{\sigma \sim \sbits^m}{
            \sup_{f \in \cF} \Abs{\frac{1}{m} \sum_{i=1}^m \sigma_i f^{\pm}(x_i)}} \,.
    \end{equation}
    Let $D$ be a probability distribution over $\cX$. The \emph{distributional Rademacher
    complexity} of family $\cF$ over $D$ at size $m$ is the expected empirical Rademacher complexity
    over a sample from $D$ of size $m$:
    \[
        \cR_m(\cF, D) \define \argEx{S \sim D^m}{\cR(\cF, S)} \,.
    \]
\end{definition}

We recall three standard facts about Rademacher complexity: with high
probability, 1)~the generalization error of a sample is controlled by the
distributional Rademacher complexity; 2)~the empirical Rademacher complexity
approximates the distributional Rademacher complexity; and 3)~a random
assignment of $\sigma \sim \sbits^m$ suffices to approximate the empirical
Rademacher complexity in \eqref{eq:empirical-rademacher}. These are summarized
below.

\begin{fact}[See e.g.\ \cite{GollakotaKK23}]
    There exists a universal constant $C > 0$ such that the following holds.
    Let $D$ be a probability distribution over $\cX$. Then
    \begin{enumerate}
        \item Let $g : \cX \to \zo$ be any labeling function. With probability
        at least $1-\delta$ over a sample $S \sim D^{m'}$ for any $m' \ge m$, it
        holds that
            \begin{equation}
                \label{eq:rademacher-generalization}
                \sup_{f \in \cF} \Abs{\err(f, \Unif(S)_g) - \err(f, D_g)}
                \le \cR_m(\cF, D)
                    + C \sqrt{\frac{\log(1/\delta)}{m}} \,.
            \end{equation}
        \item With probability at least $1-\delta$ over a sample $S \sim D^m$,
        it holds that
            \begin{equation}
                \label{eq:rademacher-empirical-distributional}
                \Abs{\cR(\cF, S) - \cR_m(\cF, D)}
                \le C \sqrt{\frac{\log(1/\delta)}{m}} \,.
            \end{equation}
        \item Let $S = (x_1, \dotsc, x_m) \in \cX^m$ be a tuple. With
        probability at least $1-\delta$ over the choice of $\sigma \sim
        \sbits^m$, it holds that
            \begin{equation}
                \label{eq:rademacher-sigma}
                \Abs{
                    \cR(\cF, S)
                    - \sup_{f \in \cF} \Abs{\frac{1}{m} \sum_{i=1}^m \sigma_i f^{\pm}(x_i)}}
                \le C \sqrt{\frac{\log(1/\delta)}{m}} \,.
            \end{equation}
    \end{enumerate}
\end{fact}

\subsection{Uniform convergence}
We state the VC theorem (also referred to as the uniform convergence theorem),
which roughly captures the sample complexity rates for the empirical loss of any
function to be a good estimate of the true loss.

\begin{theorem}[\cite{vapnik2015uniform}]
\label{th:uc}
     Let $S=\{(x_1,y_1),\dots ,(x_m,y_m)\}$, where each
     $(x_i,y_i)\in \cX\times\{0,1\}$ is sampled i.i.d.\ from a
     distribution $\cD\in\Delta(\cX\times \bits)$. Let $\err(f,S) \define
     \err(f,\Unif(S)) = \frac{1}{m}\abs{\{i:f(x_i)\neq y_i\}}$. Let $\cF$
     be a class of Boolean functions over $\cX$ of VC dimension $d$. Then
     \[
        \Pr\left[
            \underset{f\in\cF}{\sup}\vert \err(f,D)-\err(f,S)\vert
                \leq O\left(\sqrt{\frac{d+\log(1/\delta)}{m}}\right)\right]
            \geq 1-\delta \,.
     \]
\end{theorem}

\subsection{Fooling distance}
Recall from  \eqref{eq:fooling-distance} that the fooling distance between
distributions $P$ and $Q$ against family $\cF$ is $\dist_\cF(P, Q) \define
\sup_{f \in \cF} \abs{\E_P[f] - \E_Q[f]}$. Formally, the pair $(\cX, \dist_\cF)$
only forms a \emph{pseudometric}: it satisfies $\dist_\cF(P, P) = 0$ for each
$P$, the symmetry condition $\dist_\cF(P, Q) = \dist_\cF(Q, P)$ for all $P, Q$,
and the triangle inequality $\dist_\cF(P, Q) \le \dist_\cF(P, R) + \dist_\cF(R,
Q)$. However, there may exist distinct distributions $P, Q$ for which
$\dist_\cF(P, Q) = 0$.

We also occasionally specialize the notation of fooling distance to a single
function $f : \cX \to \bits$, and write $\dist_f(P, Q) \define \abs{\E_P[f] - \E_Q[f]}$. Hence $\dist_\cF(P, Q) = \sup_{f \in \cF} \dist_f(P, Q)$.

\section{Sample Complexity of \texorpdfstring{$\cF$}{F}-Identity Testing}
\label{sec:sample-complexity}

\subsection{Rademacher bounds for specific distributions}
In this section, we show that the Rademacher complexity of the family $\cF$ of
distinguishers over the reference distribution $\refdist{P}$ characterizes, up
to a quadratic factor, the sample complexity of testing identity to
$\refdist{P}$ against $\cF$. (Later, in \cref{rem:rademacher-tight}, we observe that this
quadratic gap is necessary.)
Note that this analysis is information-theoretic
and does not contemplate the \emph{computational} complexity of this problem.

We now show that the distributional Rademacher complexity of the family $\cF$
over distribution $\refdist{P}$ gives an upper bound on the sample complexity of
identity testing to $\refdist{P}$ against $\cF$. In fact, the upper bound also
holds for \emph{equivalence} (or closeness) testing, which is similar to the
identity testing problem from \cref{def:intro-identity-testing}, but in the
setting where the algorithm is only given sample access to two unknown
distributions $P$ and $Q$, and must distinguish between $P=Q$ and
$\dist_\cF(P,Q) > \epsilon$ with probability at least $1-\delta$; such an
algorithm is called an \emph{$(\epsilon,\delta)$-equivalence tester} against
$\cF$.

Both upper bounds are consequences of a more general result on a \emph{testable
distribution learning} problem, where the algorithm must succeed in learning
(with respect to metric $\dist_\cF$) any input distribution with small
Rademacher complexity. We formally define this problem below.

\begin{definition}[Testably learning distributions of bounded Rademacher complexity]
    \label{def:testable-distribution-learning}
    Given parameters $m \in \bbN$ and $(\epsilon, \epsilon_\cR, \delta) \in
    (0,1)$, an algorithm $\cA$
    is an \emph{$(m, \epsilon, \epsilon_\cR, \delta)$-Rademacher testable distribution learner}
    against family $\cF$ if, given sample access to an unknown input
    distribution $P$, algorithm $\cA$
    either rejects or outputs a distribution $\hat{P}$ such that, with probability at least
    $1-\delta$, the following two conditions are satisfied:
    \begin{itemize}
        \item \textbf{(Soundness)} Either $A$ rejects or its output $\hat{P}$
        satisfies $\dist_\cF(P, \hat{P}) \le \epsilon$.
        \item \textbf{(Completeness)} If $\cR_m(\cF, P) \le \epsilon_\cR$, then
        $A$ does not reject.
    \end{itemize}
\end{definition}

Rademacher complexity arises naturally in the formulation above because it is a property that is
itself testable from samples. We have the following general result:

\begin{lemma}
    \label{lemma:testable-distribution-learning}
    For every family $\cF$ and parameters $m, \epsilon, \delta$, there exists an $(m, \epsilon,
    \epsilon/4, \delta)$-Rademacher testable distribution learner against $\cF$ with sample
    complexity $m + O\left(\frac{\log(1/\delta)}{\epsilon^2}\right)$.
\end{lemma}

Before proving the lemma, we observe that identity and equivalence testing of distributions against
$\cF$ follow immediately:

\begin{theorem}[Upper bounds; refinement of \cref{thm:intro-testing-upper-bounds}]
    \label{thm:testing-upper-bounds}
    The following problems admit algorithms with sample complexity $m +
    O\left(\frac{\log(1/\delta)}{\epsilon^2}\right)$:
    \begin{enumerate}
        \item $(\epsilon, \delta)$-Identity testing for explicit distribution
        $\refdist{P}$ against family $\cF$, where $\cR_m(\cF, \refdist{P}) \le
        \epsilon/16$.
        \item $(\epsilon, \delta)$-Equivalence testing between unknown distributions $P$ and $Q$
            against family $\cF$, under the \emph{promise} that $\cR_m(\cF, P) \le \epsilon/16$ or
            $\cR_m(\cF, Q) \le \epsilon/16$.
    \end{enumerate}
\end{theorem}
\begin{proof}
    Clearly the second result implies the first. The equivalence testing algorithm works as follows:
    \begin{enumerate}
        \item Use the algorithm from \cref{lemma:testable-distribution-learning} with input $P$ (and
            parameters $m, \epsilon/4, \delta/2$) to learn distribution $\hat{P}$ and with input $Q$
            (and same parameters) to learn distribution $\hat{Q}$, or \textbf{reject} if either
            execution rejects.
        \item \textbf{Accept} if $\dist_\cF(\hat{P}, \hat{Q}) \le \epsilon/2$, or \textbf{reject}
            otherwise.
    \end{enumerate}
    The sample complexity claim is immediate, so it remains to show correctness. We separately show
    completeness and soundness.

    \textbf{Completeness.} Suppose $P=Q$. Then $\cR_m(\cF, P) \le \epsilon/16$ and
    \cref{lemma:testable-distribution-learning} implies that, with probability at least
    $1-\delta/2$, the first step learns distribution $\hat{P}$ satisfying $\dist_\cF(P, \hat{P}) \le
    \epsilon/4$. Similarly, with probability at least $1-\delta/2$, the first step learns
    distribution $\hat{Q}$ satisfying $\dist_\cF(P, \hat{Q}) \le \epsilon/4$. Therefore, with
    probability at least $1-\delta$, the first step does not reject and moreover we have
    \[
        \dist_\cF(\hat{P}, \hat{Q})
        \le \dist_\cF(\hat{P}, P) + \dist_\cF(P, \hat{Q})
        \le \epsilon/2 \,,
    \]
     and hence the algorithm accepts.

    \textbf{Soundness.} Suppose $\dist_\cF(P, Q) > \epsilon$. By the union bound over the two
    executions of the algorithm from \cref{lemma:testable-distribution-learning}, with probability
    at least $1-\delta$ the first step either rejects or it learns distributions $\hat{P}$ and
    $\hat{Q}$ satisfying $\dist_\cF(P, \hat{P}) \le \epsilon/4$ and $\dist_\cF(Q, \hat{Q}) \le
    \epsilon/4$. In the first case we are done, and in the second case the triangle inequality gives
    \[
        \dist_\cF(\hat{P}, \hat{Q})
        \ge \dist_\cF(P, Q) - \dist_\cF(P, \hat{P}) - \dist_\cF(Q, \hat{Q})
        > \epsilon - \epsilon/4 - \epsilon/4
        = \epsilon/2 \,,
    \]
    and hence the second step rejects.
\end{proof}

We now prove \cref{lemma:testable-distribution-learning}. The idea is to use
\eqref{eq:rademacher-empirical-distributional} to estimate the Rademacher complexity of the input
distribution $P$ from a sample and either reject right away if this quantity too large, or use
\eqref{eq:rademacher-generalization} to argue that the empirical distribution is
close to $P$ (against distinguishers in $\cF$).

\begin{proof}[Proof of \cref{lemma:testable-distribution-learning}]
    Let $s = m + \frac{C \log(1/\delta)}{\epsilon^2}$ for sufficiently large
    universal constant $C > 0$. The algorithm works as follows:
    \begin{enumerate}
        \item Draw a sample $S \sim P^s$. If $\cR(\cF, S) > \epsilon/2$, \textbf{reject}.
        \item Otherwise, \textbf{return} $\hat{P} \define \Unif(S)$.
    \end{enumerate}
    We will show that the soundness and completeness conditions of
    \cref{def:testable-distribution-learning} are satisfied with probability at least $1-\delta/2$
    each, so that the claim will follow from a union bound.

    Completeness is easy to check: if $\cR_m(\cF, P) \le \epsilon/4$,
    so that $\cR_s(\cF, P) \le \epsilon/4$,
    then \eqref{eq:rademacher-empirical-distributional} implies that, with probability at least $1 -
    \delta/2$, we have $\cR(\cF, S) \le \epsilon/2$ and hence the algorithm does not reject.

    We now verify soundness. We consider two cases. First, suppose $\cR_s(\cF, P) > 2\epsilon/3$.
    Then, again by \eqref{eq:rademacher-empirical-distributional}, with probability at least
    $1-\delta/2$ we have $\cR(\cF, S) > \epsilon/2$ and hence the algorithm rejects, in which case
    the requirement is satisfied. Second, suppose $\cR_s(\cF, P) \le 2\epsilon/3$. Then
    \eqref{eq:rademacher-generalization} with $g=0$ gives that, with probability at least
    $1-\delta/2$, for each $f \in \cF$ we have
    \[
        \Abs{\argEx{S}{f} - \argEx{P}{f}}
        = \Abs{\err(f, \Unif(S)_g) - \err(f, P_g)}
        \le \epsilon \,,
    \]
    so that $\dist_\cF(P, \hat{P}) \le \epsilon$ and hence the requirement is satisfied.
\end{proof}

We now prove a lower bound for identity testing that essentially matches the upper bound from
\cref{thm:testing-upper-bounds} up to a quadratic gap. The idea is that, if $\cR_m(\cF, \refdist{P}) \gtrsim
\epsilon$, then by \eqref{eq:rademacher-sigma}, with high probability a random assignment $\sigma
\sim \sbits^m$ determines a (nearly-balanced) partition $S = T \cup R$ of the sample such that
$\Abs{\argEx{T}{f} - \argEx{R}{f}} \gtrsim \epsilon$ for some $f \in \cF$. This implies, via the
triangle inequality, that at least one of the distributions $\Unif(R)$ or $\Unif(T)$ is far from $\refdist{P}$ under
fooling distance, and yet indistinguishable from $\refdist{P}$ using
$o(\sqrt{m})$ samples by a birthday problem bound.

\thmrademacherlowerbound*
\begin{proof}
    Let $\cA$ be an $(\epsilon, 0.99)$-identity tester for $\refdist{P}$ against $\cF$ with sample complexity
    $c\sqrt{m}$, where $c > 0$ is a sufficiently small absolute constant, and consider the following
    experiment:
    \begin{enumerate}
        \item Draw a sample $S = (x_1, \dotsc, x_m) \sim \refdist{P}^m$ and a vector $\sigma \sim \sbits^m$.
        \item Define multisets $R \define \{x_i \in S : \sigma_i = +1\}$ and $T \define \{x_i \in S
            : \sigma_i = -1\}$. \textbf{Reject} if $|R| < m/3$.
        \item Otherwise, simulate $\cA$ with samples drawn uniformly and independently at random
            from $R$; \textbf{accept} if $\cA$ accepts, and \textbf{reject} if $\cA$ rejects.
    \end{enumerate}
    We first observe that the sets $R$ and $T$ are nearly balanced with good probability. Standard
    concentration of Rademacher random variables gives
    \[
        \Pr\Brac{\Abs{\Abs{R}-\Abs{T}} > 10\sqrt{m}}
        = \Pr\Brac{\Abs{\sum_{i=1}^m \sigma_i} > 10\sqrt{m}}
        \le 2e^{-10^2/2}
        < 0.01 \,.
    \]
    Using the fact that $\Abs{T} = m - \Abs{R}$, we conclude that, with probability at least $0.99$,
    we have
    \begin{equation}
        \label{eq:R-size}
        \Abs{2\Abs{R} - m} \le 10\sqrt{m}
        \implies
        \Abs{\Abs{R} - \frac{m}{2}} \le 5\sqrt{m} < \frac{m}{6}
        \implies
        \Abs{R} > \frac{m}{3} \,.
    \end{equation}
    Finally, using the assumption that $m \ge C/\epsilon^2$ for sufficiently large $C > 0$, we also
    obtain in this case that
    \begin{equation}
        \label{eq:R-size-eps}
        \Abs{\frac{\Abs{R}}{m/2} - 1} \le \frac{10}{\sqrt{m}} \le \frac{\epsilon}{32} \,.
    \end{equation}

    The rest of the proof proceeds in two parts. First, we show that $c\sqrt{m}$ independent samples
    drawn uniformly from $R$ are nearly indistinguishable from samples from $\refdist{P}$, and conclude that
    the experiment accepts with high constant probability. Then, we show that there is at least a
    constant probability that $\dist_\cF(\Unif(R), \refdist{P}) > \epsilon$, and conclude that the experiment
    rejects with some constant probability. This yields the desired contradiction.

    \textbf{Experiment accepts.}
    Suppose $\Abs{R} > m/3$. Then the probability that some element $x_i \in R$ is sampled as an
    input to $\cA$ more than once (i.e., a \emph{collision}, where we see multiset elements $x_i$
    and $x_j$ as distinct if $i \ne j$) is, by a union bound,
    \begin{equation}
        \label{eq:prob-collision}
        \Pr\Brac{\text{collision} \;|\; \Abs{R} > m/3}
        \le \frac{(c\sqrt{m})^2}{\Abs{R}} < 3c^2 \le 0.01 \,.
    \end{equation}
    Moreover, if no collisions occur, the samples from $R$ are identically distributed to
    independent samples from $\refdist{P}$. By the completeness of $\cA$, \eqref{eq:R-size},
    \eqref{eq:prob-collision}, and a union bound, we conclude that
    \begin{align*}
        &\Pr\Brac{\text{Experiment accepts}} \\
        &\qquad \ge 1 - \Pr\Brac{\Abs{R} \le m/3} - \Pr\Brac{\text{collision} \land \Abs{R} > m/3}
            - \Pr\Brac{\text{$\cA$ fails}} \\
        &\qquad \ge 1 - 0.03
        = 0.97 \,. \numberthis \label{eq:experiment-accepts}
    \end{align*}

    \textbf{Experiment rejects.}
    Combining \eqref{eq:rademacher-empirical-distributional} and \eqref{eq:rademacher-sigma}, with
    probability at least $0.99$ we have
    \begin{equation}
        \label{eq:exp-rad}
        \Abs{
            \cR_m(\cF, \refdist{P})
            - \sup_{f \in \cF} \Abs{\frac{1}{m} \sum_{i=1}^m \sigma_i f^{\pm}(x_i)}}
        \le \Theta\left(\sqrt{\frac{1}{m}}\right)
        \le \epsilon
    \end{equation}
    and hence, for some $f \in \cF$ (assuming the supremum is attained for simplicity),
    \[
        \Abs{\frac{1}{m} \sum_{i=1}^m \sigma_i f^{\pm}(x_i)}
        \ge \cR_m(\cF, \refdist{P}) - \epsilon
        \ge 3\epsilon \,.
    \]
    Write $\Abs{R} = \frac{m}{2}(1 + r)$ and $\Abs{T} = \frac{m}{2}(1 + t)$. We then have
    \begin{align*}
        3\epsilon
        &\le \Abs{\frac{1}{m} \sum_{i=1}^m \sigma_i f^{\pm}(x_i)}
        = \Abs{\frac{1}{m} \sum_{i=1}^m \sigma_i (2f(x_i) - 1)} \\
        &= \Abs{\frac{1}{m} \sum_{x \in R} (2f(x)-1) - \frac{1}{m} \sum_{x \in T} (2f(x)-1)} \\
        &= \Abs{\frac{\Abs{R}}{m/2} \argEx{R}{f} - \frac{\Abs{T}}{m/2} \argEx{T}{f}
                - \frac{\Abs{R}}{m} + \frac{\Abs{T}}{m}} \\
        &\le \Abs{(1+r) \argEx{R}{f} - (1+t) \argEx{T}{f}}
                + \frac{1}{2}\Abs{r} + \frac{1}{2}\Abs{t} \\
        &\le \Abs{\argEx{R}{f} - \argEx{T}{f}} + \frac{3}{2}\Abs{r} + \frac{3}{2}\Abs{t} \,,
    \end{align*}
    where in the last step we used the fact that $0 \le f \le 1$. Recall that, by
    \eqref{eq:R-size-eps}, we have $\Abs{r}, \Abs{t} \le \epsilon/32$ with probability at least
    $0.99$, and in this case we have
    \[
        \dist_\cF(\Unif(R), \Unif(T))
        \ge \Abs{\argEx{R}{f} - \argEx{T}{f}}
        \ge 3\epsilon - 3 \cdot \epsilon/32
        > 2\epsilon \,,
    \]
    \sloppy
    which implies via the triangle inequality that at least one of $\dist_\cF(\Unif(R), \refdist{P}) > \epsilon$
    or $\dist_\cF(\Unif(T), \refdist{P}) > \epsilon$ holds. Applying a union bound over \eqref{eq:R-size-eps} and
    \eqref{eq:exp-rad}, and using the symmetry between $R$ and $T$, we obtain that
    \[
        \Pr\Brac{\dist_\cF(\Unif(R), \refdist{P}) > \epsilon} \ge \frac{1 - 2 \cdot 0.01}{2} = 0.49 \,.
    \]
    By the soundness of algorithm $\cA$ and a union bound, we conclude that
    \begin{align*}
        \Pr\Brac{\text{Experiment rejects}}
        &\ge 1 - \Pr\Brac{\dist_\cF(\Unif(R), \refdist{P}) \le \epsilon} - \Pr\Brac{\text{$\cA$ fails}} \\
        &\ge 1 - 0.51 - 0.01
        = 0.48 \,,
    \end{align*}
    which contradicts \eqref{eq:experiment-accepts}.
\end{proof}

\subsection{VC bounds for all distributions}
\label{sec:vc-bounds}
As mentioned in the introduction, we are also interested in the sample
complexity of testers which, for any reference distribution $\refdist{P}$ whose
description is given explicitly as an input, and given sample access to an
unknown distribution $P$, test identity to $\refdist{P}$ against family $\cF$ on
input $P$. For the sake of clarity, we explicitly define the model below. Note
that, since this section is concerned with sample complexity rather than time
complexity, we do not discuss here the computational complexity of representing
and computing over the explicitly given reference distribution.

\begin{definition}[$\cF$-identity testing for all reference distributions]
    \label{def:testing-for-all}
    Given a domain $\cX$, a family $\cF$ of Boolean functions $f : \cX \to \bits$,
    and parameters $\epsilon, \delta \in (0,1)$, algorithm $\cA$ is an
    \emph{$(\epsilon, \delta)$-identity tester for all reference distributions
    against $\cF$} if, given the explicit description of a reference distribution
    $\refdist{P}$ over $\cX$ and sample access to a unknown input distribution
    $P$ over $\cX$, with probability at least $1-\delta$ the following conditions
    hold:
    \begin{itemize}
        \item \textbf{(Completeness)} If $P = \refdist{P}$, then $\cA$ accepts.
        \item \textbf{(Soundness)} If $\dist_\cF(P, \refdist{P}) > \epsilon$,
        then $\cA$ rejects.
    \end{itemize}
\end{definition}

It is an immediate consequence of uniform convergence that $O(d/\epsilon^2)$ samples
suffice to test identity for all reference distributions against a family $\cF$ whose
VC dimension is $d$, and this also follows from our distribution-specific upper
bound from \cref{thm:intro-testing-upper-bounds} by known VC bounds for Rademacher
complexity. Nevertheless, since the VC bound is convenient and indeed applied later
in \cref{sec:structured}, here we spell out the tester and its guarantees.

\begin{algorithm}[!ht]
\caption{$\cF$-Identity Tester against VC classes}
\label{alg:idtesting}
    \begin{algorithmic}[1]
    \State \textbf{Input:}
    Samples $S=\left( x_1,\dots,x_m\right)\sim \tgtdist{P}^m$ and reference
    distribution $\refdist{P}$ over $\cX$.
    \For {$f\in \cF$}
    \If{\(\left\lvert  \frac{1}{m} \sum_{i=1}^{m} f(x_i) - \mathbb{E}_{x\sim \refdist{P}}[f(x)] \right\rvert  > \varepsilon/2\)}
    \State \textbf{Output} ``Reject"
    \EndIf
    \EndFor
    \State \textbf{Output} ``Accept"
    \end{algorithmic}
\end{algorithm}

\begin{lemma}
    \label{lemma:vc-upper-bound}
    Let $\cX$ be a domain and $\cF$ be a family of Boolean functions over $\cX$
    of VC dimension $d$. Then \cref{alg:idtesting} is an
    $(\epsilon,\delta)$-identity tester for all reference distributions against
    $\cF$ with sample complexity $m = O\left(\frac{d + 
    \log(1/\delta)}{\epsilon^2}\right)$.
\end{lemma}
\begin{proof}
   \textbf{Completeness $(P=\refdist{P})$:} Since $P=\refdist{P}$, the uniform
   convergence theorem (\cref{th:uc}) implies that, with probability at least
   $1-\delta$, for all functions $f\in \cF$ we have
    \begin{align*}
        \left\lvert  \frac{1}{m} \sum_{i=1}^{m} f(x_i)
            - \mathbb{E}_{x\sim \refdist{P}}[f(x)] \right\rvert
        \le \varepsilon/2 \,.
    \end{align*}
    Hence the tester accepts with probability at least $1-\delta$.

    \textbf{Soundness ($\dist_{\cF}(P,\refdist{P})>\eps$):}
    By definition of fooling distance, there exists a function
    $f^*\in \cF$ such that $\left\lvert \ee_{x\sim P}[f^*(x)]  - \ee_{x\sim
    \refdist{P}}[f^*(x)] \right\rvert > \varepsilon$.
    By \cref{th:uc}, with probability at least $1-\delta$ it holds that
    $\lvert\E_{x\sim S}[f(x)]-\E_{x\sim P}[f(x)]\rvert\leq\varepsilon/2$.
    When this holds, the triangle inequality yields
    \begin{align*}
        \varepsilon&<\lvert\E_{x\sim \refdist{P}}[f^*(x)]-\E_{x\sim P}[f^*(x)]\rvert\\
        &\leq \lvert\E_{x\sim \refdist{P}}[f^*(x)]-\E_{x\sim S}[f^*(x)]\rvert+\lvert\E_{x\sim S}[f^*(x)]-\E_{x\sim P}[f^*(x)]\rvert\\
        &\leq \lvert\E_{x\sim \refdist{P}}[f^*(x)]-\E_{x\sim S}[f^*(x)]\rvert+\eps/2 \,,
    \end{align*}
    so that $\lvert\E_{x\sim \refdist{P}}[f^*(x)]-\E_{x\sim S}[f^*(x)]\rvert >
    \epsilon/2$ and the step of the tester with $f = f^*$ rejects.
\end{proof}

\paragraph*{Lower bounds.} It is possible to show that, for each family $\cF$ of
VC dimension $d$, any identity tester for all distributions against $\cF$
requires $\Omega(\sqrt{d})$ samples; otherwise, by a birthday problem bound, the
uniform distribution over a shattered set cannot be distinguished from the
uniform distribution over a randomly selected half of that set. For convenience,
we delay the statement of this bound until \cref{sec:vc-lower-bounds}, where we
obtain it via a black-box reduction from a lower bound for PAC verification by
\cite{MutrejaS23} instead. More interestingly, in that section we also leverage
a lower bound for PAC verification by \cite{GoldwasserRSY21} to show that there
exist classes $\cF$ of VC dimension $\tilde O(d)$ such that every identity
tester for all reference distributions against $\cF$ requires $\Omega(d)$
samples.

%\section{Computationally Efficient Testing and Applications to Testable Learning}
\section{Relation to Testable Learning}
\label{sec:testable-learning}

In this section, we prove our results relating $\cF$-identity testing to testable learning.
Recall the definition of testable learning from \cref{def:testable-learning}.

\subsection{Testable learning implies \texorpdfstring{$\cF$}{F}-identity testing}
\label{sec:testable-learning-implies-f-identity-testing}
%\begin{definition}[Testable agnostic learning]\label{def:testable}
%We say a tester-learner pair $(T,A)$ testably learns a class $\cF\subseteq \{\pm1\}^{\cX}$ w.r.t.\ a fixed marginal distribution $D_{\cX}$ on $\cX$ up to excess error $\eps>0$ if for any distribution $D$ on $\cX\times\{\pm1\}$, the following holds:
%\begin{itemize}
%    \item \textbf{(Soundness/composability)} If $D$ is such that the tester $T$ accepts with high probability over a sample drawn from $D$, then the learner $A$ succeeds in agnostically learning $\cF$ w.r.t.\ $D$, i.e., with probability at least $1-\delta$ it produces a hypothesis $h$ such that
%    \[
%      \err(h,D)\le\opt(\cF,D)+\eps,
%    \]
%    where, $\err(h,D):=\Pr_{(x,y)\sim D}[h(x)\neq y].$
%    \item \textbf{(Completeness)} Whenever $D$ truly has marginal $D_{\cX}$ on $\cX$, the tester $T$ accepts with probability at least $1-\delta$.
%\end{itemize}
%\end{definition}

We start with the following lemma relating the fooling distance between $P$ and
$Q$ to the optimal classification error of a certain mixture distribution.

\begin{lemma}\label{lemma:mixtureopt}
    Let $\cF\subseteq \bits^{\cX}$ be a class of Boolean functions and let
    $P,Q\in \Delta(\cX)$ be probability distributions. Define the mixture
    (joint) distribution $M = M(P,Q) \in \Delta(\cX\times\bits)$ by $M =
    \frac{1}{2}(P,1) + \frac{1}{2}(Q,0)$. Then
    $\opt(\cF,M)=\frac{1}{2}(1-\dist_{\cF}(P,Q))$.
\end{lemma}
\begin{proof}
    For each $f\in \cF$, the definition of classification error yields
    \begin{align*}
        \err(f,M)&=\Pr_{(x,y)\sim M}\left[f(x)\neq y\right]\\
        &=\frac{1}{2}\left(\Pr_{x\sim P}\left[f(x)\neq 1\right]+\Pr_{x\sim Q}\left[f(x)\neq 0\right]\right)\\
        &=\frac{1}{2}\left(1-\left(\Pr_{x\sim P}\left[f(x)= 1\right]-\Pr_{x\sim Q}\left[f(x)=1\right]\right)\right) \,.
    \end{align*}
    Minimizing over all $f \in \cF$ and using the assumption that $\cF$ is
    closed under complement we get
    \[
        \opt(\cF,M)=\frac{1}{2}\left(1-\dist_{\cF}(P,Q)\right) \,. \qedhere
    \]
\end{proof}

We now describe our construction showing how to convert a testable agnostic learner
(for function class $\cF$ with respect to the marginal distribution
$\refdist{P}$) to an $\cF$-identity tester to reference distribution $\refdist{P}$. Let
$m_{TL}$ denote the sample complexity of such a testable agnostic learning
algorithm.
\\

\textbf{Algorithm:}\label{alg:tl-to-it}
\begin{enumerate}
    \item Draw a sample $S_{TL}$ of size $m_{TL}$ from the distribution
    $M(\refdist{P},P)$. Additionally, draw a sample $S_V$ of size
    $m_V=O\left(\frac{\log (1/\delta)}{\eps^2}\right)$ from $M(\refdist{P},P)$.
    \item Simulate the testable learner with input $S_{TL}$. If that
    algorithm rejects, then \textbf{reject}.
    \item Otherwise, let $h$ be the hypothesis returned by the testable learner.
    If the empirical error of $h$ on $S_V$ is strictly less than
    $\frac{1}{2}-\frac{7\eps}{40}$, \textbf{reject}; else, \textbf{accept}. 
\end{enumerate}

We now prove the correctness of this reduction. Formally, we get the following guarantee:
\begin{theorem}[Testable learning implies $\cF$-identity testing; refinement of
    \cref{thm:intro-testable-learning-implies-identity-testing}]
    \label{thm:reduction}
    Let $\TL$ be an $(\epsilon/5, \delta/2)$-testable learner for $\cF$
    with respect to marginal $\cD_\cX=\refdist{P}$. Then there exists an $(\epsilon,
    \delta)$-identity tester to reference distribution $\refdist{P}$, with
    sample complexity equal to the sample complexity of $\TL$ plus an additional
    $O(\log(1/\delta)/\epsilon^2)$, and running time polynomial in the running
    time of $\TL$.
\end{theorem}
\begin{proof}
We analyze the algorithm described just above.

\textbf{Completeness ($\refdist{P}=P$).} Algorithm $\TL$ rejects only with
probability at most $\delta/2$ by its completeness condition. If $\TL$ produces
a hypothesis $h$, then since the mixture distribution in this case is
$M=\frac{1}{2}(P,1)+\frac{1}{2}(P,0)$, the distribution over the labels is
uniform, which implies that
\begin{equation}\label{eq:opterror}
    \err(h,M)=\frac{1}{2}.
\end{equation}
By Hoeffding's inequality, if we choose $\abs{S_V} =
\frac{C\log(1/\delta)}{\epsilon^2}$ for sufficiently large constant $C$,
with probability at least $1-\delta/2$ it holds that
\begin{equation}\label{eq:hoeffdingopt}
    \lvert\err(h,\Unif(S_V))-\err(h,M)\rvert\leq \eps/8.
\end{equation}
Taking the union bound over the events that $\TL$ rejects and that
%described by \eqref{eq:opterror} and
\eqref{eq:hoeffdingopt} fails to hold,
we get that with
probability at least $1-\delta$, algorithm $\TL$ outputs a
hypothesis $h$ such that $\err(h,\Unif(S_V))\geq \frac{1}{2}-\eps/8
> \frac{1}{2} - 7\epsilon/40$, and in this case the algorithm accepts.

\textbf{Soundness ($\dist_{\cF}(\refdist{P},P)>\eps$).}
If $\TL$ rejects, then our tester rejects and we are done. Suppose $\TL$ outputs
a hypothesis $h$. By \cref{lemma:mixtureopt}, we have
\begin{equation}\label{eq:optbound}
    \opt(\cF,M)=\frac{1}{2}\left(1-\dist_{\cF}(\refdist{P},P)\right)< \frac{1}{2}\left(1-\eps\right).
\end{equation}
By the soundness of $\TL$ and \eqref{eq:optbound},
with probability at least $1-\delta/2$ we have
\[
    \err(h,M)\leq \opt(\cF,M)+\eps/5 < \frac{1}{2}\left(1-\eps \right)+\eps/5.
\]
By Hoeffding's inequality and the union bound as before, with probability at
least $1-\delta$ we have
\[
    \err(h,\Unif(S_V))< \frac{1}{2}-\eps/2+\eps/5+\eps/8=\frac{1}{2}-7\eps/40 \,,
\]
and in this case the tester rejects.
\end{proof}

\subsection{\texorpdfstring{$\cF$}{F}-Identity testing implies testable query learning}
Before proceeding, we also define distribution-specific agnostic learning, as follows:

\begin{definition}[Distribution-specific agnostic query learning]
    We say algorithm $\cL$ is an \emph{$(\epsilon, \delta)$-agnostic query
    learner} for family $\cF$ under marginal $\refdist{P}$ if, given query
    access to input function $f : \cX \to \zo$, algorithm $\cL$ outputs a
    function $h : \cX \to \zo$ satisfying $\err(h, \refdist{P}^f) \le \opt(\cF,
    \refdist{P}^f) + \epsilon$ with probability at least $1-\delta$. Algorithm
    $\cL$ is called \emph{proper} if its outputs satisfies $h \in \cF$.
\end{definition}

The class $f \oplus \cF$ of error distinguishers, defined implicitly by the
labeling function $f$ to which the algorithm has query access, will be crucial
in our reductions. Below, we define this class and the formal model of identity
testing against it.

\begin{definition}[Error distinguishers]
    Let $f : \cX \to \zo$ be a Boolean function over $\cX$. Write $f \oplus \cF
    \define \{ f \oplus g : g \in \cF \}$ for the family of \emph{error
    distinguishers} defined by function $f$ and family $\cF$.
\end{definition}

\begin{definition}[Distribution testing against implicit error distinguishers]
    We say algorithm $\cT$ is an $(\epsilon, \delta)$-identity tester to
    $\refdist{P}$ against \emph{implicit $\cF$-error distinguishers} if,
    given query access to input function $f : \cX \to \zo$, algorithm $\cT$ is
    an $(\epsilon, \delta)$-identity tester to $\refdist{P}$ against family $f
    \oplus \cF$, that is, it uses queries to $f$ and samples from input
    distribution $P$ to distinguish between the cases $P = \refdist{P}$ and
    $\dist_{f \oplus \cF}(P, \refdist{P}) > \epsilon$ with success probability
    at least $1-\delta$.
\end{definition}

The next lemma shows that, if we have a distribution-specific agnostic query
learner for a family $\cF$ under marginal $\refdist{P}$, as well as an identity
tester to $\refdist{P}$ against implicit $\cF$-error distinguishers, then we
obtain a testable query learner for $\cF$ with respect to marginal
$\refdist{P}$.

\begin{lemma}[Testable query learning via identity testing against implicit error distinguishers]
    \label{lemma:testable-query-learning-from-implicit}
    Suppose $\cT$ is an $(\epsilon, \delta)$-identity tester to $\refdist{P}$
    against implicit $\cF$-error distinguishers using $s_\cT$ samples, $q_\cT$
    queries, and $t_\cT$ time. Suppose $\cL$ is an $(\epsilon, \delta)$-agnostic
    query learner for $\cF$ under marginal $\refdist{P}$ using $q_\cL$ queries
    and $t_\cL$ time. Suppose a sample from $\refdist{P}$ may be produced in
    time $t_{\refdist{P}}$. Then there exists a $(3\epsilon, 2\delta)$-testable
    query learner $\cA$ for $\cF$ with respect to marginal $\refdist{P}$ using
    $s_\cT + O\left(\frac{\log(1/\delta)}{\epsilon^2}\right)$ samples from $P$,
    $q_\cT + q_\cL + O\left(\frac{\log(1/\delta)}{\epsilon^2}\right)$ queries to
    $f$, and $\poly(t_\cT) + \poly(t_\cL) + \poly(t_{\refdist{P}}
    \log(1/\delta)/\epsilon)$ time. Also, $\cA$ is proper if $\cL$ is.
\end{lemma}
\begin{proof}
    The algorithm $\cA$ first simulates tester $\cT$ on input distribution $P$.
    If $\cT$ rejects, algorithm $\cA$ rejects. Otherwise, $\cA$ simulates
    learner $\cL$ by taking samples from (the explicit) distribution
    $\refdist{P}$ and querying $f$ to obtain hypothesis $h$. It then uses
    $O(\log(1/\delta)/\epsilon^2)$ samples from $P$ and queries to $f$ to
    estimate the quantities $\E_P[f \oplus h]$ and $\E_{\refdist{P}}[f \oplus
    h]$ to additive error $\epsilon/4$ with overall success probability
    $1-\delta$. Finally, $\cA$ returns the hypothesis $h$ if these two estimates
    are within $\epsilon/2$ of each other, and rejects otherwise.

    The sample, query and time complexity claims are immediate, and so is
    properness when $\cL$ is proper. Completeness follows from the fact that
    $\cT$ accepts with probability at least $1-\delta$ when $P=\refdist{P}$, and
    that in this case we have $\E_P[f \oplus h] = \E_{\refdist{P}}[f \oplus h]$,
    so the second step only rejects with probability at most $\delta$. It
    remains to verify soundness.

    First, suppose $\dist_{f \oplus \cF}(P, \refdist{P}) > \epsilon$. Then, with
    probability at least $1-\delta$, tester $\cT$ rejects and hence so does
    $\cA$, so soundness is satisfied. Similarly, if the hypothesis $h$
    produced\footnote{For simplicity of the analysis, we assume that $\cA$
    always simulates both $\cT$ and $\cL$, so that $h$ is available even if it
    is ultimately discarded.} by $\cL$ is such that $\abs{\E_P[f \oplus h] -
    \E_{\refdist{P}}[f \oplus h]} > \epsilon$, then with probability at least
    $1-\delta$ the estimates of these expectations are more than $\epsilon/2$
    apart, so $\cA$ rejects and soundness is satisfied.

    On the other hand, suppose $\dist_{f \oplus \cF}(P, \refdist{P}) \le
    \epsilon$ and that the hypothesis $h$ is such that $\abs{\E_P[f \oplus h] -
    \E_{\refdist{P}}[f \oplus h]} \le \epsilon$. Suppose that
    $\err(h, \refdist{P}^f) \le \opt(\cF, \refdist{P}^f) + \epsilon$, which
    occurs with probability at least $1-\delta$. If $\cA$ rejects, then
    soundness is satisfied. Otherwise, $\cA$ outputs hypothesis $h$, and we have
    \begin{align*}
        \err(h, P_f)
        &= \E_P[f \oplus h]
        \le \E_{\refdist{P}}[f \oplus h] + \epsilon
        = \err(h, \refdist{P}^f) + \epsilon \\
        &\le \inf_{g \in \cF} \E_{\refdist{P}}[f \oplus g] + 2\epsilon
        \le \inf_{g \in \cF} \E_P[f \oplus g] + 3\epsilon
        = \opt(\cF, P_f) + 3\epsilon \,,
    \end{align*}
    where the first inequality is by the assumption that $\E_P[f \oplus h]$ and
    $\E_{\refdist{P}}[f \oplus h]$ are close to each other, the second inequality
    is by the assumption on $\err(h, \refdist{P}^f)$, and the third inequality
    is by the assumption that $\dist_{f \oplus \cF}(P, \refdist{P}) \le
    \epsilon$ and the definition of family $f \oplus \cF$. Thus soundness is
    satisfied.
\end{proof}

The lemma above raises the question: why should one hope to have testers against
implicit error distinguishers? Next, we show that one may in fact obtain such a
tester from an $\cF$-identity tester. The idea is to use samples from $P$ and queries to
$f$ to create two ``synthetic'' distributions which, depending on the value of
$f$, mix components of $P$ and $\refdist{P}$. Then, we can show by the triangle
inequality that at least one such distribution must be such that family $\cF$
already witnesses the distance captured by family $f \oplus \cF$.

\begin{lemma}[Tester against implicit $\cF$-error distinguishers from tester against $\cF$]
    \label{lemma:implicit-tester-reduction}
    Suppose $\cT$ is an $(\epsilon, \delta)$-identity tester to $\refdist{P}$
    against $\cF$ using $s_\cT$ samples and $t_\cT$ time. Also, suppose a sample
    from $\refdist{P}$ may be produced in time $t_{\refdist{P}}$. Then
    there exists a $(7\epsilon, 4\delta)$-identity tester to $\refdist{P}$
    against implicit $\cF$-error distinguishers using $s = O\left(s_\cT +
    \frac{\log(1/\delta)}{\epsilon^2}\right)$ samples from $P$, $s$ queries to
    $f$, and $\poly(s \cdot t_{\refdist{P}}) + \poly(t_\cT)$ time.
\end{lemma}
\begin{proof}
    Let $\mu_P \define \E_P[f]$ and $\mu_{\refdist{P}} \define
    \E_{\refdist{P}}[f]$. For each $b \in \zo$, let $P|b$ denote the
    distribution of $x \sim P$ conditional on $f(x) = b$, set arbitrarily in the
    pathological case where it is not uniquely defined, and define
    $\refdist{P}|b$ analogously. Given $\alpha \in [0,1]$, we define the
    distributions $P^\alpha_0$, $P^\alpha_1$ as follows: for each $b \in \zo$, a
    sample $x \sim P^\alpha_b$ is produced by sampling $b' \sim \bern(\alpha)$,
    and then $x \sim P|b'$ if $b' = b$ or $x \sim \refdist{P}|b'$ if $b' \ne b$.

    \textbf{Algorithm.} The tester works as follows:
    \begin{enumerate}
        \item Take $s$ samples from $P$ and, using queries to $f$, divide the
        samples into tuples $S_0^P$ and $S_1^P$ by putting each sampled element
        $x$ into tuple $S_{f(x)}^P$. Do the same using samples from
        $\refdist{P}$ to construct the tuples $S_0^{\refdist{P}}$ and
        $S_1^{\refdist{P}}$.

        \item\label{item:step-mu} Compute $\hat{\mu}_P \define
        \frac{\abs{S_1^P}}{s}$ and $\hat{\mu}_{\refdist{P}} \define
        \frac{\abs{S_1^{\refdist{P}}}}{s}$. If $\abs{\hat{\mu}_P -
        \hat{\mu}_{\refdist{P}}} > \epsilon/2$, \textbf{reject}.

        \item For each $b \in \zo$, do the following:
            \begin{enumerate}
                \item If $\min\left\{\hat{\mu}_{\refdist{P}},
                1-\hat{\mu}_{\refdist{P}}\right\} \ge \epsilon/2$, let $R_b =
                P^{\mu_{\refdist{P}}}_b$. Otherwise, let $R_b = P$.

                \item\label{item:step-simulation} Simulate tester $\cT$ on input
                distribution $R_b$. When $R_b=P$, a sample from $R_b$ is
                obtained by sampling from $P$. When $R_b =
                P^{\mu_{\refdist{P}}}_b$, draw a sample from $R_b$ as follows.
                First, draw $b' \sim \bern(\mu_{\refdist{P}})$ by drawing $z
                \sim \refdist{P}$ and letting $b' = f(z)$. Then, if $b' = b$,
                provide a fresh element from $S_{b'}^P$, and if $b' \ne b$,
                provide a fresh element from $S_{b'}^{\refdist{P}}$; if this is
                impossible because the corresponding tuple has run out of fresh
                elements, \textbf{reject}.

                \item If the simulation of $\cT$ rejects, \textbf{reject}.
            \end{enumerate}
        \item If no previous step has rejected, \textbf{accept}.
    \end{enumerate}

    The sample, query, and time complexity claims are immediate, so we now prove
    correctness. We show completeness and soundness separately.

    \textbf{Completeness.} Suppose $P=\refdist{P}$. Then $\mu_P =
    \mu_{\refdist{P}}$, and hence $\abs{\hat{\mu}_P - \hat{\mu}_{\refdist{P}}}
    \le \epsilon/2$, with probability at least $1-\delta$ by Hoeffding's
    inequality, and in this case Step~\ref{item:step-mu} does not reject.

    We now claim that, with probability at least $1-\delta$, the algorithm has
    enough elements in the tuples $S_b^P, S_b^{\refdist{P}}$ to simulate $\cT$
    without rejecting in Step~\ref{item:step-simulation}. If
    $\min\left\{\mu_{\refdist{P}}, 1-\mu_{\refdist{P}}\right\} < \epsilon/4$,
    then $\min\left\{\hat{\mu}_{\refdist{P}}, 1-\hat{\mu}_{\refdist{P}}\right\}
    < \epsilon/2$ with probability at least $1-\delta$ by Hoeffding's
    inequality, and in this case we set $R_b=P$ and
    Step~\ref{item:step-simulation} does not reject. Therefore we may assume
    that $\min\left\{\mu_{\refdist{P}}, 1-\mu_{\refdist{P}}\right\} \ge
    \epsilon/4$.

    Since $P=\refdist{P}$ and hence $\mu_P=\mu_{\refdist{P}}$, we may without
    loss of generality fix $b=1$ and consider the case of drawing samples from
    $S_1^P$ when $b'=1$. By Hoeffding's inequality, we have, for all $t > 0$,
    \begin{equation}
        \label{eq:tail-bound-1}
        \Pr\left[\abs{\abs{S_1^P} - s\mu_P} \ge t\right] \le 2\exp(-2t^2 / s) \,.
    \end{equation}
    Let $N$ denote the number of fresh elements from $S_1^P$ requested during
    the simulation, so that $N \sim \bin(s_\cT, \mu_P)$ (recall that $\mu_P =
    \mu_{\refdist{P}}$ here). Then, similarly, we have
    \begin{equation}
        \label{eq:tail-bound-2}
        \Pr\left[\abs{N - s_\cT \mu_P} \ge t\right]
        \le 2\exp(-2 t^2 / s_\cT)
        \le 2\exp(-2 t^2 / s) \,.
    \end{equation}
    The tuple $S_1^P$ has enough fresh elements when $\abs{S_1^P} \ge N$. Let
    $\ell \define \ln(16/\delta)$ and $t \define \sqrt{\frac{\ell s}{2}}$, so
    that the RHS of \eqref{eq:tail-bound-1}, \eqref{eq:tail-bound-2} is
    $\delta/8$. A union bound yields that, with probability at least
    $1-\delta/4$, we have
    \[
        N \le s_\cT \mu_P + t
        \qquad \text{and} \qquad
        s \mu_P - t \le \abs{S_1^P} \,.
    \]
    Thus it suffices to show that $s_\cT \mu_P + t \le s \mu_P - t$, or equivalently that
    \begin{equation}
        \label{eq:equiv-ineq}
        \mu_P (s - s_\cT) \ge \sqrt{2 \ell s} \,.
    \end{equation}
    Thanks to the bound $\mu_P \ge \epsilon/4$, we have that \eqref{eq:equiv-ineq} is implied by
    \[
        \epsilon (s - s_\cT) \ge 4 \sqrt{2 \ell s} \,.
    \]
    By choosing $s \ge 2s_\cT$, so that $s - s_\cT \ge s/2$, it suffices to have
    $s \ge 128\ell/\epsilon^2$, which indeed holds when we set $s =
    \Theta\left(s_\cT + \frac{\log(1/\delta)}{\epsilon^2}\right)$ appropriately.
    We conclude that \eqref{eq:equiv-ineq} holds and that the tuple $S_1^P$ runs
    out of fresh elements with probability at most $\delta/4$. By a union bound,
    Step~\ref{item:step-simulation} rejects with probability at most $\delta$,
    as claimed.

    Finally, since $\refdist{P} = P = P^{\mu_{\refdist{P}}}_b$ for each $b \in
    \zo$, in each simulation of $\cT$ we have $R_b=\refdist{P}$. Since each
    fresh element from the tuples $S_b^P, S_b^{\refdist{P}}$ is an independent
    sample from $P|b, \refdist{P}|b$ respectively, the probability that a
    simulation of $\cT$ rejects is at most the probability that a simulation of
    $\cT$ with fresh independent samples would reject, plus the probability that
    the simulation fails because one of the tuples $S_b^P, S_b^{\refdist{P}}$
    does not have enough elements. The latter is accounted for above, and the
    former is at most $\delta$ by the soundness of $\cT$. By a union bound over
    the events considered before and the two simulations of $\cT$, the algorithm
    accepts with probability at least $1-4\delta$.

    \textbf{Soundness.} Suppose $\dist_{f \oplus \cF}(P, \refdist{P}) >
    7\epsilon$. If $\abs{\mu_P - \mu_{\refdist{P}}} > \epsilon$, then
    Step~\ref{item:step-mu} rejects with probability at least $1-\delta$ by
    Hoeffding's inequality, so we may assume that $\abs{\mu_P -
    \mu_{\refdist{P}}} \le \epsilon$.

    Since each fresh element from the tuples $S_b^P, S_b^{\refdist{P}}$ is an
    independent sample from $P|b, \refdist{P}|b$ respectively, while running out
    of fresh elements only causes the algorithm to reject, it suffices to show
    the claim that $\dist_\cF(R_b, \refdist{P}) > \epsilon$ for some $b \in \zo$
    with probability at least $1-\delta$: in this case, a perfect simulation of
    $\cT$ (without ever running out of fresh elements) would reject with
    probability at least $1-\delta$, and the event of running out of elements
    only increases the probability of rejection; thus by a union bound the
    tester rejects with probability at least $1-2\delta$. The rest of the proof
    is devoted to establishing the claim.

    We first show that $\dist_\cF(P_b^{\mu_{\refdist{P}}}, \refdist{P}) >
    \epsilon$ for some $b \in \zo$. Let $g \in \cF$ be a function such that
    $\abs{\E_P[f \oplus g] - \E_{\refdist{P}}[f \oplus g]} > 7\epsilon >
    4\epsilon$, which exists by hypothesis. Write $\bar{f} \define 1-f$, and
    similarly $\bar g \define 1-g$. We have
    \begin{align*}
        4\epsilon
        &< \Abs{\E_P[f \oplus g] - \E_{\refdist{P}}[f \oplus g]}
        = \Abs{\E_P[f \cdot \bar g + \bar f \cdot g]
            - \E_{\refdist{P}}[f \cdot \bar g + \bar f \cdot g]} \\
        &\le \Abs{\E_P[f \cdot \bar g] - \E_{\refdist{P}}[f \cdot \bar g]}
            + \Abs{\E_P[\bar f \cdot g] - \E_{\refdist{P}}[\bar f \cdot g]} \,.
    \end{align*}
    \sloppy
    Thus one of the two terms in the last line above is bigger than $2\epsilon$.
    Suppose that $\Abs{\E_P[f \cdot \bar g] - \E_{\refdist{P}}[f \cdot \bar g]}
    > 2\epsilon$. We have
    \begin{align*}
        \dist_\cF(P_1^{\mu_{\refdist{P}}}, \refdist{P})
        &\ge \Abs{\E_{P_1^{\mu_{\refdist{P}}}}[\bar g] - \E_{\refdist{P}}[\bar g]}
        = \Abs{\E_{P_1^{\mu_{\refdist{P}}}}[f \cdot \bar g + \bar f \cdot \bar g]
            - \E_{\refdist{P}}[f \cdot \bar g + \bar f \cdot \bar g]} \\
        &= \Abs{\E_{P_1^{\mu_{\refdist{P}}}}[f \cdot \bar g]
            - \E_{\refdist{P}}[f \cdot \bar g]} \,,
    \end{align*}
    \sloppy
    where the last equality holds because $\E_{P_1^{\mu_{\refdist{P}}}}[\bar f
    \cdot \bar g] = \E_{\refdist{P}}[\bar f \cdot \bar g] =
    (1-\mu_{\refdist{P}}) \E_{\refdist{P}|0}[\bar g]$ by the construction of
    $P_1^{\mu_{\refdist{P}}}$. Since $\Abs{\E_P[f \cdot \bar g] -
    \E_{\refdist{P}}[f \cdot \bar g]} > 2\epsilon$, if we show that $\Abs{\E_P[f
    \cdot \bar g] - \E_{P_1^{\mu_{\refdist{P}}}}[f \cdot \bar g]} \le \epsilon$,
    then we will conclude via the triangle inequality that
    $\dist_\cF(P_1^{\mu_{\refdist{P}}}, \refdist{P}) \ge
    \Abs{\E_{P_1^{\mu_{\refdist{P}}}}[f \cdot \bar g] - \E_{\refdist{P}}[f \cdot
    \bar g]} > \epsilon$, as desired. Using $\abs{\mu_P - \mu_{\refdist{P}}} \le
    \epsilon$, we have
    \[
        \Abs{\E_P[f \cdot \bar g] - \E_{P_1^{\mu_{\refdist{P}}}}[f \cdot \bar g]}
        = \Abs{\mu_P \cdot \E_{P|1}[\bar g]
            - \mu_{\refdist{P}} \cdot \E_{P|1}[\bar g]}
        \le \Abs{\mu_P - \mu_{\refdist{P}}}
        \le \epsilon \,,
    \]
    as needed. Similarly, suppose $\Abs{\E_P[\bar f \cdot g] -
    \E_{\refdist{P}}[\bar f \cdot g]} > 2\epsilon$ instead. We now have
    \begin{align*}
        \dist_\cF(P_0^{\mu_{\refdist{P}}}, \refdist{P})
        &\ge \Abs{\E_{P_0^{\mu_{\refdist{P}}}}[g] - \E_{\refdist{P}}[g]}
        = \Abs{\E_{P_0^{\mu_{\refdist{P}}}}[f \cdot g + \bar f \cdot g]
            - \E_{\refdist{P}}[f \cdot g + \bar f \cdot g]} \\
        &= \Abs{\E_{P_0^{\mu_{\refdist{P}}}}[\bar f \cdot g]
            - \E_{\refdist{P}}[\bar f \cdot g]} \,,
    \end{align*}
    where the last equality holds because $\E_{P_0^{\mu_{\refdist{P}}}}[f \cdot
    g] = \E_{\refdist{P}}[f \cdot g] = \mu_{\refdist{P}} \E_{\refdist{P}|1}[g]$
    by the construction of $P_0^{\mu_{\refdist{P}}}$. Again it suffices to show
    that $\Abs{\E_P[\bar f \cdot g] - \E_{P_0^{\mu_{\refdist{P}}}}[\bar f \cdot
    g]} \le \epsilon$, which holds by
    \[
        \Abs{\E_P[\bar f \cdot g] - \E_{P_0^{\mu_{\refdist{P}}}}[\bar f \cdot g]}
        = \Abs{(1-\mu_P) \cdot \E_{P|0}[g] - (1-\mu_{\refdist{P}}) \cdot \E_{P|0}[g]}
        \le \Abs{\mu_P - \mu_{\refdist{P}}}
        \le \epsilon \,.
    \]

    We have now established that $\dist_\cF(P_b^{\mu_{\refdist{P}}},
    \refdist{P}) > \epsilon$ for some $b \in \zo$, and our goal is to show that
    $\dist_\cF(R_b, \refdist{P}) > \epsilon$ for some $b \in \zo$ with
    probability at least $1-\delta$. To conclude the proof, we consider two
    cases.

    First, suppose $\min\left\{\mu_P, 1-\mu_P\right\} \ge 2\epsilon$. Then
    $\min\left\{\mu_{\refdist{P}}, 1-\mu_{\refdist{P}}\right\} \ge \epsilon$
    (since $\Abs{\mu_P-\mu_{\refdist{P}}} \le \epsilon$), and hence
    $\min\left\{\hat{\mu}_{\refdist{P}}, 1-\hat{\mu}_{\refdist{P}}\right\} \ge
    \epsilon/2$ with probability at least $1-\delta$ by Hoeffding's inequality.
    Moreover, in this case $R_b = P_b^{\mu_{\refdist{P}}}$ in each iteration $b
    \in \zo$, and the claim follows.

    Second, suppose $\min\left\{\mu_P, 1-\mu_P\right\} < 2\epsilon$. If
    $\min\left\{\hat{\mu}_{\refdist{P}}, 1-\hat{\mu}_{\refdist{P}}\right\} \ge
    \epsilon/2$, then again we are done, so suppose that
    $\min\left\{\hat{\mu}_{\refdist{P}}, 1-\hat{\mu}_{\refdist{P}}\right\} <
    \epsilon/2$, so that $R_b = P$ for each $b \in \zo$. Suppose without loss of
    generality that $\mu_P < 2\epsilon$, so that $\mu_{\refdist{P}} <
    3\epsilon$. We claim that $\dist_\cF(P, \refdist{P}) > \epsilon$. Let $g \in
    \cF$ be such that $\abs{\E_P[f \oplus g] - \E_{\refdist{P}}[f \oplus g]} >
    7\epsilon$, which exists by hypothesis. We have
    \begin{align*}
        \dist_\cF(P, \refdist{P})
        &\ge \Abs{\argEx{P}{g} - \argEx{\refdist{P}}{g}}
        = \Abs{\argEx{P}{f \cdot g + \bar f \cdot g}
            - \argEx{\refdist{P}}{f \cdot g + \bar f \cdot g}} \\
        &= \Abs{\mu_P \argEx{P|1}{g} - \mu_{\refdist{P}} \argEx{\refdist{P}|1}{g}
                + (1-\mu_P) \argEx{P|0}{g}
                - (1-\mu_{\refdist{P}}) \argEx{\refdist{P}|0}{g}} \\
        &\ge \Abs{(1-\mu_P) \argEx{P|0}{g}
            - (1-\mu_{\refdist{P}}) \argEx{\refdist{P}|0}{g}}
            - \underbrace{\Abs{\mu_P \argEx{P|1}{g}
                - \mu_{\refdist{P}} \argEx{\refdist{P}|1}{g}}}_{< 3\epsilon} \\
        &> \Abs{(1-\mu_P) \argEx{P|0}{g}
            - (1-\mu_{\refdist{P}}) \argEx{\refdist{P}|0}{g}
            + \mu_P \argEx{P|1}{\bar g}
            - \mu_{\refdist{P}} \argEx{\refdist{P}|1}{\bar g}} \\
            &\qquad - \underbrace{\Abs{\mu_P \argEx{P|1}{\bar g}
                - \mu_{\refdist{P}} \argEx{\refdist{P}|1}{\bar g}}}_{
                < 3\epsilon} - 3\epsilon \\
        &> \Abs{\argEx{P}{f \cdot \bar g + \bar f \cdot g}
                - \argEx{\refdist{P}}{f \cdot \bar g + \bar f \cdot g}}
            - 6\epsilon \\
        &= \Abs{\argEx{P}{f \oplus g} - \argEx{\refdist{P}}{f \oplus g}} - 6\epsilon
        > \epsilon \,,
    \end{align*}
    as claimed. The case when $1-\mu_P < 2\epsilon$ proceeds similarly by
    starting with the function $\bar g$ instead, and again we obtain that
    $\dist_\cF(P, \refdist{P}) > \epsilon$. This concludes the proof.
\end{proof}

Combining \cref{lemma:testable-query-learning-from-implicit,lemma:implicit-tester-reduction} yields
the main result of this section.

\begin{theorem}[$\cF$-identity testing implies testable query learning; refinement of
    \cref{thm:intro-testable-query-learning-from-identity-testing}]
    \label{thm:testable-query-learning-from-identity-testing}
    Suppose $\cT$ is an $(\epsilon,\delta)$-identity tester to $\refdist{P}$
    against $\cF$ using $s_\cT$ samples and $t_\cT$ time. Suppose $\cL$ is an
    $(\epsilon, \delta)$-agnostic query learner for $\cF$ under marginal
    $\refdist{P}$ using $q_\cL$ queries and $t_\cL$ time. Also, suppose a sample
    from $\refdist{P}$ may be produced in time $t_{\refdist{P}}$. Then there
    exists an $(O(\epsilon), O(\delta))$-testable query learner $\cA$ for $\cF$
    with respect to marginal $\refdist{P}$ using $s = O\left(s_\cT +
    \frac{\log(1/\delta)}{\epsilon^2}\right)$ samples from the input marginal
    $P$, $s + q_\cL$ queries to $f$, and $\poly(s \cdot t_{\refdist{P}}) +
    \poly(t_\cT) + \poly(t_\cL)$ time. Also, $\cA$ is proper if $\cL$ is.
\end{theorem}
\begin{proof}
    Apply \cref{lemma:implicit-tester-reduction} to obtain from $\cT$ an
    $(O(\epsilon), O(\delta))$-identity tester $\cT'$ to $\refdist{P}$ against
    implicit $\cF$-error distinguishers using $s$ samples and queries, and
    $\poly(s \cdot t_{\refdist{P}}) + \poly(t_\cT)$ time. Then, apply
    \cref{lemma:testable-query-learning-from-implicit} to obtain from $\cT'$ and
    $\cL$ an $(O(\epsilon), O(\delta))$-testable query learner for $\cF$ with
    respect to marginal $\refdist{P}$ with the announced sample, query, and time
    complexity bounds, which is proper if $\cL$ is.
\end{proof}

\subsection{Implications for testable proper learning with queries}

Agnostic proper learning is a notoriously challenging problem, and indeed the
attempt to design efficient algorithms for this task runs up against
complexity-theoretic barriers
\cite{feldman2006new,guruswami2009hardness,daniely2016complexity}.
On the other hand, two
standard modifications to the general model can often make agnostic proper
learning more tractable: 1)~fixing an underlying distribution, i.e., moving to
the \emph{distribution-specific} setting; and, possibly,
2)~also\footnote{Feldman showed that, in the distribution-free setting
of agnostic learning, queries essentially do not help \cite{Feldman09}.}
allowing the algorithm to query the input function at arbitrary points, i.e.,
giving it \emph{membership queries}. For example, the class of halfspaces over
$\bbR^n$ admits an agnostic proper learning algorithm under the standard
Gaussian distribution with running time $n^{\poly(1/\epsilon)}$
\cite{DiakonikolasKKT21}, and in fact the dependence on $n$ can be made fully
polynomial if membership queries are allowed \cite{DiakonikolasKKT24}; and the
class of size-$s$ decision trees over $\bits^n$ admits an agnostic proper
learning algorithm with membership queries under the uniform distribution with
running time $\poly(n) \cdot s^{O(\log \log s)}$ \cite{BlancLQT22}.

We show that, similarly, testable \emph{proper} learning algorithms may be obtained by allowing
membership queries. Specifically, we obtain testable proper learning algorithms using membership
queries for the classes of \emph{halfspaces} (with respect to the standard Gaussian distribution)
and \emph{size-$s$ decision trees} (with respect to the uniform distribution over the hypercube),
while matching the computational complexity of the corresponding improper, sample-based testable
learners.

We give a general recipe of the form ``agnostic proper query learning and
testable learning imply testable proper query learning.'' That is, we obtain
our results by combining two ingredients:
\begin{enumerate}
    \item An efficient agnostic proper learner for family $\cF$ under a
    ``good marginal distribution'' (e.g., Gaussian or uniform over the
    hypercube), possibly using membership queries.

    \item An efficient (improper) testable learner with respect to the good
    marginal. This implies, via \cref{thm:reduction}, an identity tester to the
    good marginal against $\cF$, which suffices for an application of our
    general \cref{thm:testable-query-learning-from-identity-testing}.
\end{enumerate}
Given these two ingredients, our results follow from the framework established
earlier in this section via black-box reductions, rather than requiring tailored
arguments to address the combined challenges of testable and proper learning.

\subsubsection{General result}

The following theorem encodes the general scheme outlined above, which underlies the results in this
subsection.

\begin{lemma}
    \label{lemma:general-testable-proper-query-learning}
    Let $\cX$ be a domain, $\cF$ be a family of Boolean functions over $\cX$,
    and $\refdist{P}$ be a probability distribution over $\cX$. Suppose a
    sample from $\refdist{P}$ may be produced in time $t_{\refdist{P}}$. Suppose
    we have the following algorithms $\cL$ and $\TL$:
    \begin{enumerate}
        \item $\cL$ is an $(\epsilon, \delta)$-agnostic proper query learner for
        $\cF$ under marginal $\refdist{P}$ using $q_\cL$ queries and $t_\cL$
        time.
        \item $\TL$ is an $(\epsilon, \delta)$-testable learner for $\cF$ with
        respect to marginal $\refdist{P}$ using $s_{\TL}$ labeled samples and
        $t_{\TL}$ time.
    \end{enumerate}
    Then there exists an $(O(\epsilon), O(\delta))$-testable proper query
    learner for $\cF$ with respect to marginal $\refdist{P}$ using $s =
    O\left(s_{\TL} + \frac{\log(1/\delta)}{\epsilon^2}\right)$ samples from the
    input marginal $P$, $s + q_\cL$ queries to the input function $f$, and
    $\poly(s \cdot t_{\refdist{P}}) + \poly(t_{\TL}) + \poly(t_\cL)$ time.
\end{lemma}
\begin{proof}
    Combine \cref{thm:testable-query-learning-from-identity-testing,thm:reduction}.
\end{proof}

\subsubsection{Halfspaces}

The first ingredient we require is an agnostic proper learner.
\cite{DiakonikolasKKT21} gave an agnostic proper learner for halfspaces under
the standard Gaussian distribution with running time $n^{O(1/\epsilon^4)}$. To
more closely match the $n^{\tilde O(1/\epsilon^2)}$ bound for testable learning
first attained by \cite{GollakotaKK23}, we use the following result by
\cite{DiakonikolasKKT24} instead, which achieves an improved upper bound by
allowing membership queries (which are inherent to
\cref{lemma:general-testable-proper-query-learning} in any case).

\begin{theorem}[Theorem~A.1 of \cite{DiakonikolasKKT24}]
    \sloppy
    There exists an $(\epsilon, \delta)$-agnostic proper query learner for the class of halfspaces
    over $\bbR^n$ under the standard Gaussian distribution using $\poly(n/\epsilon)$ samples and
    queries, and $\left(\poly(n/\epsilon) + (1/\epsilon)^{O(1/\epsilon^6)}\right) \log(1/\delta)$
    time.
\end{theorem}

The second ingredient we require is a testable learner. Here we cite the result of
\cite{GollakotaKK23}, with suitably-amplified success probability via standard arguments:

\begin{theorem}[\cite{GollakotaKK23}]
    There exists an $(\epsilon, \delta)$-testable learner for the class of halfspaces over $\bbR^n$
    with respect to the standard Gaussian distribution using $n^{\tilde O(1/\epsilon^2)}
    \log(1/\delta)$ labeled samples and time.
\end{theorem}

Combining these results via \cref{lemma:general-testable-proper-query-learning} yields our main
result for halfspaces:

\begin{restatable}[Testable proper query learner for halfspaces w.r.t.\
    Gaussian]{theorem}{thmqueryhalfspaces}
    \label{thm:proper-halfspaces}
    \sloppy
    Suppose a sample from the standard Gaussian distribution over $\bbR^n$ may be produced in time
    $t_\cN$. Then, there exists an $(\epsilon, \delta)$-testable proper query learner for the class
    of halfspaces over $\bbR^n$ with respect to the standard Gaussian distribution using $n^{\tilde
    O(1/\epsilon^2)} \log(1/\delta)$ samples and queries, and $\left(n^{\tilde O(1/\epsilon^2)}
    \poly(t_\cN) + (1/\epsilon)^{O(1/\epsilon^6)}\right) \poly\log(1/\delta)$ time.
\end{restatable}

\subsubsection{Size-$s$ decision trees}

We rely on the following result of \cite{BlancLQT22} on agnostic proper learning of decision trees
using membership queries, with suitably-amplified success probability:

\begin{theorem}[\cite{BlancLQT22}]
    There exists an $(\epsilon, \delta)$-agnostic proper query learner for the class of size-$s$
    decision trees over $\bits^n$ under the uniform distribution with running time $\poly(n) \cdot
    (s/\epsilon)^{O\left(\log\left((\log s)/\epsilon\right)\right)} \cdot \poly\log(1/\delta)$.
\end{theorem}

The existence of an $n^{O(\log(s/\epsilon))}$-time algorithm for testable learning of size-$s$
decision trees with respect to the uniform distribution over $\bits^n$ is implied by
\cite[Theorem~4.2]{GollakotaKK23} and \cite{Baz09} via the existence of degree-$O(\log(s/\epsilon))$
sandwiching polynomials for this class \cite[Lemma~34]{KlivansSV24}, and it is also directly stated
(with an additional distributional tolerance guarantee) in \cite{GoelSSV24}.

\begin{theorem}[\cite{GollakotaKK23,KlivansSV24,GoelSSV24}]\label{thm:tl-dt}
    There exists an $(\epsilon, \delta)$-testable learner for the class of size-$s$ decision trees
    over $\bits^n$ with respect to the uniform distribution with running time
    $n^{O(\log(s/\epsilon))} \cdot \poly\log(1/\delta)$.
\end{theorem}

Combining these results via \cref{lemma:general-testable-proper-query-learning} yields our main
result for decision trees:

\begin{restatable}[Testable proper query learner for size-$s$ decision trees w.r.t\
    uniform]{theorem}{thmquerydecisiontrees}
    \label{thm:proper-decision-trees}
    There exists an $(\epsilon, \delta)$-testable proper query learner for the class of size-$s$
    decision trees over $\bits^n$ with respect to the uniform distribution with running time
    $(n/\epsilon)^{O(\log(s/\epsilon))} \cdot \poly\log(1/\delta)$.
\end{restatable}

\subsection{Implications for testable semiagnostic proper learning}
\label{sec:semiagnostic}

In the previous subsection, we obtained our testable proper learning results by
studying the membership query model. In certain settings, this model may be
unrealistic, as only sample access to the labeled points may be available. In
this subsection, we observe that a weaker form of testable proper learning,
namely up to $3 \cdot \opt +\ \epsilon$ error, is attainable even in the
sample-based setting.

To accomplish this, we start with the same assumptions as the previous
subsection (the existence of a testable learner and an agnostic proper query
learner for the class) and resort to a two-step process: first, apply standard
testable learning to the input function $f$ to learn, using samples, a
hypothesis $\hat{f}$. Then, apply
\cref{lemma:general-testable-proper-query-learning} to $\hat{f}$ to learn a
\emph{proper} approximation of $\hat{f}$, providing that algorithm with query
access to the learned $\hat{f}$ using its explicit description. Due to the
triangle inequality, this procedure ensures that, if the algorithm does not
reject the input distribution, then it produces (with high probability) a
hypothesis that is at most $(3 \cdot \opt +\ \epsilon)$-far from $f$. The
conclusion is the following general result:

\newcommand{\teval}{\text{eval}}

\begin{lemma}
    Let $\cX$ be a domain, $\cF$ be a family of Boolean functions over $\cX$,
    and $\refdist{P}$ be a probability distribution over $\cX$. Suppose a sample
    from $\refdist{P}$ may be produced in time $t_{\refdist{P}}$, and that a
    hypothesis from $\cF$ may be evaluated at any point $x \in \cX$ in time
    $t_{\teval}$. Suppose we have the following algorithms and $\cL$ and $\TL$:
    \begin{enumerate}
        \item $\cL$ is an $(\epsilon, \delta)$-agnostic proper query learner for $\cF$ under $\refdist{P}$
            using $q_\cL$ queries and $t_\cL$ time.
        \item $\TL$ is an $(\epsilon, \delta)$-testable learner for $\cF$ with respect to $\refdist{P}$ using
            $s_{\TL}$ labeled samples and $t_{\TL}$ time.
    \end{enumerate}
    Then there exists an $(O(\epsilon), O(\delta))$-testable $3$-semiagnostic
    proper learner for $\cF$ with respect to $\refdist{P}$ using $s =
    O\left(s_{\TL} + \frac{\log(1/\delta)}{\epsilon^2}\right)$ labeled samples
    from the input joint distribution $\cD$ and $\poly(s \cdot t_{\refdist{P}})
    + \poly(t_{\TL}) + \poly(t_\cL \cdot t_{\teval})$ time.
\end{lemma}
\begin{proof}
    Let $P = \cD_\cX$ be the marginal on $\cX$ of the input joint distribution
    $\cD$. The algorithm $\cA$ simulates $\TL$, which either rejects or produces
    a hypothesis $\hat{f}$. If $\TL$ rejects, then $\cA$ rejects. Otherwise,
    $\cA$ simulates the testable proper query learner from
    \cref{lemma:general-testable-proper-query-learning}, call it $\cB$, on
    marginal $P$ and labeling function $\hat{f}$; then, $\cA$ rejects if $\cB$
    rejects, or otherwise outputs the hypothesis $h \in \cF$ produced by $\cB$.

    The complexity bounds and properness are immediate. Completeness is also clear because, when
    $P=\refdist{P}$, each of the algorithms $\TL$ and $\cB$ only reject with probability $O(\delta)$ each, and hence
    so does $\cA$ by a union bound. As for soundness, suppose $\TL$ and $\cB$ produced correct
    outputs (i.e., satisfying their own soundness conditions), which occurs with probability at
    least $1 - O(\delta)$. Then either one of the steps rejected, in which case $\cA$ rejects and
    soundness is satisfied, or $\cA$ produces a hypothesis $h$. In this case, letting $f^* \define
    \arg\min_{f' \in \cF} \err(f', \cD)$ (which we assume is attained for simplicity), we have
    \begin{align*}
        \err(h, \cD)
        &= \Pr_{(x,y) \sim \cD}[h(x) \ne y] \\
        &\le \Pr_{x \sim P}[h(x) \ne \hat{f}(x)] + \Pr_{(x,y) \sim \cD}[\hat{f}(x) \ne y] \\
        &= \err(h, P_{\hat{f}}) + \err(\hat{f}, \cD) \\
        &\le \opt(\cF, P_{\hat{f}}) + \opt(\cF, \cD) + O(\epsilon)
            \tag{Soundness of $\cB$ and $\TL$} \\
        &\le \err(f^*, P_{\hat{f}}) + \opt(\cF, \cD) + O(\epsilon) \\
        &= \Pr_{x \sim P}[f^*(x) \ne \hat{f}(x)] + \opt(\cF, \cD) + O(\epsilon) \\
        &\le \Pr_{(x,y) \sim \cD}[f^*(x) \ne y] + \Pr_{(x,y) \sim \cD}[\hat{f}(x) \ne y]
            + \opt(\cF, \cD) + O(\epsilon) \\
        &= \err(f^*, \cD) + \err(\hat{f}, \cD) + \opt(\cF, \cD) + O(\epsilon) \\
        &\le 3 \cdot \opt(\cF, \cD) + O(\epsilon)
            \tag{Optimality of $f^*$ and soundness of $\TL$} \,.
    \end{align*}
\end{proof}

As a consequence, we conclude

\begin{corollary}
    There exist testable $3$-semiagnostic proper learners from the classes of halfspaces with
    respect to the standard Gaussian distribution over $\bbR^n$, and size-$s$ decision trees with
    respect to the uniform distribution over $\bits^n$, with running times matching
    \cref{thm:proper-halfspaces,thm:proper-decision-trees}, respectively.
\end{corollary}

\section{Applications to PAC Verification of Learning Algorithms}
\label{sec:verification}
The abstract connection between verification and testing has been known in the literature. The work
of \cite{MutrejaS23} implicitly reduced uniformity testing to interactive verification of VC
classes. Meanwhile, the work of \cite{chiesa2018proofs} reduce identity testing of distributions to
interactive verification of distributions. Building on some of the ideas presented in these papers,
we prove our results relating $\cF$-identity testing to PAC verification
of learning algorithms.

Recall the definition of PAC verification from \cref{def:ver}.
Before we prove our mild equivalence result, we first set up some additional notation. 

\begin{definition}[Lifted class $\tilde{\cF}$]
\label{def:Lifted function class}
Given $\cF\subseteq \{0,1\}^{\cX}$, define the \emph{lifted} function class $\tilde{\cF}: \cX\times\{0,1\}\to \{0,1\}$ as
\[
\tilde{\cF} \define \left\{\tilde{f}(x,y)=\Ind[f(x)\neq y]:\;f\in \cF\right\}.
\]
Note that for any $\cD\in\Delta(\cX\times \{0,1\})$, $\ee_{(x,y)\sim\cD}[\tilde{f}(x,y)]=\err(f,\cD)$.
\end{definition}

\subsection{Identity testers from verification protocols}
We first show how to obtain an $\tilde{\cF}$-identity tester to a specific reference
distribution (whose marginal on the last bit is uniform) from a
testable interactive proof system. Although the resulting tester
does not have a direct interpretation in terms of $\cF$-identity testing,
a Rademacher complexity argument will allow us to
conclude from this a verifier sample complexity lower bound for
testable verification.

\begin{algorithm}[!ht]
\caption{$\tilde{\cF}$-identity tester to distribution
$\refdist{\cD}=D_{\cX}\times\Unif(\bits)$}
    \begin{algorithmic}[1]
    \State \textbf{Input:} Samples from distribution $\cD\in\Delta(\cX\times \{0,1\})$. 
    \State Run the interactive verification protocol where $P$ is instantiated on $\refdist{\cD}$ and $V$ is given sample access to $\cD$.
    \State If verifier $V$ rejects, output ``Reject".
    \State Draw a tuple $S_m$ of $m$ i.i.d.\ samples $(x_i,y_i)\sim\cD$.
    If the output hypothesis $h_V$ of the verifier satisfies $\err(h_V,\Unif(S_m))= \frac1m\sum_{i=1}^m \Ind[h_V(x_i)\neq y_i] < \frac{1}{2}-\epsilon/4$,
    output ``Reject".
    \State Otherwise output ``Accept".
  \end{algorithmic}
  \label{alg:tester-from-protocol}
\end{algorithm}

\begin{theorem}\label{thm:IP-to-IT}
    Let $\cF\subseteq \{0,1\}^{\cX}$ be a function class closed under complement. Fix any marginal
    $\cD_{\cX}$ on $\cX$ and let the reference distribution be $\refdist{\cD} = \cD_{\cX} \times
    \Unif(\{0,1\})$. Let $\langle P,V\rangle$ be a testable interactive proof system that
    verifies $\cF$ to error $\alpha$ and confidence $1-\beta$ with respect to marginal $\cD_\cX$.

    Then there is an $(\epsilon, \delta)$-identity tester to $\refdist{\cD}$ against
    $\tilde{\cF}$ (\cref{alg:tester-from-protocol}) which simulates $\langle P, V \rangle$
    internally with parameters $\alpha = \epsilon/4, \beta = \delta/2$ and uses $m =
    O(\eps^{-2}\log(1/\delta))$ extra samples from the input joint distribution $\cD$.
\end{theorem}

\begin{proof}
\smallskip
\textbf{Completeness $(\cD=\refdist{\cD})$:}
Following the construction of the tester, we see that
both the prover and the verifier get sample access to $\refdist{\cD}$.
Since
the marginal distribution on the labels is the uniform distribution, it is easy
to see that for every $h\in \cF$, we have $\err(h,\refdist{\cD})=\frac12$.

The algorithm checks the empirical loss of the output hypothesis $h_V$
of the verifier
by sampling $m \ge C \eps^{-2}\log(1/\delta)$ samples, for
sufficiently large constant $C$. By a Hoeffding bound, with probability at least
$1-\delta/2$ we have
\begin{equation}\label{eq:hoeff}
    \lvert\err(h_V,\Unif(S_m))-\err(h_V,\refdist{\cD})\rvert\leq \eps/4.    
\end{equation}

    Since $P$ is the honest prover and the prover and verifier get samples from $\cD =
    \refdist{\cD}$ with marginal $\cD_\cX$, the completeness of the $\langle P,V\rangle$
protocol with parameters $(\eps/4,\delta/2)$
ensures that, with probability at least
$1-\delta/2$, the verifier $V$ does not reject.
Since every hypothesis in $\cF$ has error at
least $1/2$ on the distribution $\refdist{\cD}$, taking the union bound
over the events that $V$ rejects or \eqref{eq:hoeff} fails to hold,
we get that, with probability at least $1-\delta$, the verifier $V$ does not
reject and $\err(h_V,\Unif(S_m))\geq \frac{1}{2}-\eps/4$, and in this case the
tester accepts.

    \textbf{Soundness}\;$(\dist_{\tilde{\cF}}(\refdist{\cD},\cD)>\varepsilon)$:
    Then there exists a hypothesis $\tilde{h}\in\tilde{\cF}$ such that
    $\lvert\ee_{\refdist{\cD}}[\tilde{h}]
    - \ee_{\cD}[\tilde{h}]\rvert>
    \varepsilon$. Since
    $\ee_{(x,y)\sim\cD}[\tilde{h}(x,y)]=\err(h,\cD)$
    and similarly for $\refdist{\cD}$, it follows that
    $\lvert\err(h,\refdist{\cD})-\err(h,\cD)\rvert>\eps$
    (where $h \in \cF$ is the hypothesis whose lifted version is $\tilde{h}$).

    Since the marginal distribution of $\refdist{\cD}$ on the labels is
    the uniform distribution, we have
    $\err(h,\refdist{\cD})=\frac{1}{2}$ and hence
    $\lvert \err(h,\cD)-1/2\rvert>\eps$.
    It follows that $\opt(\cF,\cD) < 1/2-\eps$.

    By the soundness of the $\langle P,V \rangle$ protocol (where $V$ gets samples from the
    arbitrary distribution $\cD$ and we may view $P$ instantiated on $\refdist{\cD}$ as
    a possibly dishonest prover $P'$),
    the verifier outputs a hypothesis $h$ such that 
    \begin{equation}
    \label{eq:verifier-soundness}
    \Pr\left[h = \text{reject}\lor \left(\err(h,\cD)\leq \opt(\cF,\cD)+\eps/4\right)\right]\geq 1-\delta/2.
    \end{equation}
  If the verifier $V$ rejects, then our identity tester rejects and we are done.
  Suppose $V$ outputs hypothesis $h$. When the condition in \eqref{eq:verifier-soundness} holds,
  we obtain
  \[
    \err(h, \cD) \le \opt(\cF, \cD) + \epsilon/4
    < \frac{1}{2} - \epsilon + \epsilon/4
    = \frac{1}{2} - 3\epsilon/4 \,.
  \]
  \sloppy
  In this case,
  a Hoeffding bound implies that, with probability at least $1-\delta/2$,
  we have $\err(h,\Unif(S_m)) < \frac{1}{2} - \epsilon/4$,
  in which case the tester rejects. The proof follows by a union bound.
\end{proof}

In order to use our Rademacher-based sample complexity lower bound for
$\cF$-identity testing from \cref{thm:rademacher-lower-bound}, we first observe that the
Rademacher complexity of $\cF$ over $\cD_\cX$ is equal to the Rademacher
complexity of the lifted $\tilde{\cF}$ over $\refdist{\cD}$. Intuitively,
this holds because both $\tilde{\cF}$ and $\refdist{\cD}$ are
``uninformative'' in the last coordinate compared to $\cF$ and $\cD_\cX$.

\begin{claim}
    \label{claim:rademacher-equal}
    For each $m \ge 1$, it holds that $\cR_m(\cF, \cD_\cX) = \cR_m(\tilde{\cF},
    \refdist{\cD})$.
\end{claim}
\begin{proof}
    Recall the notation $f^{\pm} = 2f-1$ for the $\sbits$-valued version of $f$.
    Note that each lifted function $\tilde{f}$ satisfies $\tilde{f}^{\pm}(x,y) =
    (-1)^y f^{\pm}(x)$. Now, by definition of Rademacher complexity and using
    the ongoing assumption that $\cF$ is closed under complementation, we have
    \begin{align*}
        \cR_m(\tilde{\cF}, \refdist{\cD})
        &= \argEx{(x_i, y_i)_i \sim \refdist{\cD}^m}{
            \argEx{\sigma \sim \sbits^m}{
                \sup_{\tilde{f} \in \tilde{\cF}} \Abs{
                    \frac{1}{m} \sum_{i=1}^m \sigma_i \tilde{f}^{\pm}(x_i, y_i)
                }
            }
        } \\
        &= \argEx{(x_i)_i \sim \cD_\cX^m}{
            \argEx{(y_i)_i \sim \bits^m}{
                \argEx{\sigma \sim \sbits^m}{
                    \sup_{f \in \cF}
                    \frac{1}{m} \sum_{i=1}^m \sigma_i (-1)^{y_i} f^{\pm}(x_i)
                }
            }
        } \,.
    \end{align*}
    Note that the random variable $\left(\sigma_i (-1)^{y_i}\right)_{i \in [m]}$
    is distributed as $\Unif(\sbits)^m$. Hence, using again that $\cF$ is closed
    under complementation, we have
    \[
        \cR_m(\tilde{\cF}, \refdist{\cD})
        = \argEx{(x_i)_i \sim \cD_\cX^m}{
            \argEx{\sigma \sim \sbits^m}{
                \sup_{f \in \cF} \Abs{
                    \frac{1}{m} \sum_{i=1}^m \sigma_i f^{\pm}(x_i)
                }
            }
        }
        = \cR_m(\cF, \cD_\cX) \,. \qedhere
    \]
\end{proof}

Combining \cref{thm:IP-to-IT,claim:rademacher-equal,thm:rademacher-lower-bound},
we obtain

\thmintroverificationlowerbound*
\begin{proof}
    If there was a testable verification protocol for $\cF$ with respect to marginal
    $\cD_\cX$ with verifier sample complexity $o(\sqrt{m})$, then by
    \cref{thm:IP-to-IT} we would have an $\tilde{\cF}$-identity tester to
    $\refdist{\cD}$ with sample complexity $o(\sqrt{m}) +
    \frac{1}{\epsilon^2} \le c\sqrt{m}$ for any suitably chosen sufficiently
    small constant $c > 0$, since $m \ge C/\epsilon^4$ for a sufficiently large
    constant $C$. But \cref{claim:rademacher-equal} implies that
    $\cR_m(\tilde{\cF}, \refdist{\cD}) = \cR_m(\cF, \cD_\cX) > 4\epsilon$,
    so the existence of this tester contradicts
    \cref{thm:rademacher-lower-bound}.
\end{proof}

\subsection{Verification protocols from identity testers}
In the opposite direction, we show how to obtain distribution-free PAC
verification protocols from $\cF$-identity testers for all reference distributions. We
first reduce PAC verification to $\tilde{\cF}$-identity testing, and then reduce
$\tilde{\cF}$-identity testing to $\cF$-identity testing.

\begin{algorithm}[!ht]
\caption{Prover-Verifier protocol $\langle P,V \rangle$ for function class $\cF$}
    \begin{algorithmic}[1]
    \State \textbf{Input:} Explicit description of the distribution $\cD$ for prover $P$ and sample access to distribution $\cD$ for the verifier $V$. 
    \State The honest prover $P$ sends an explicit description of a distribution $\cD'=\cD$ on $\cX\times\{0,1\}$ and a hypothesis $f'\in \arg\min_{f\in\cF}\err(f,\cD)$. 
    \State The verifier $V$ iterates over all functions in $f\in\cF$ and calculates $\opt(\cF,\cD')$. 
    \If{$\err(f',\cD')\neq \opt(\cF,\cD')$}
    \State Output ``Reject". \label{line:reject1}
    \EndIf
    \State Run $\tilde{\cF}$-identity tester $A_{\tilde{\cF}}$ with parameters $(\eps/2,\delta/2)$
        on reference distribution $\cD'$ and samples from $\cD$.
    \If{$A_{\tilde{\cF}}$ accepts}
    \State Output $f'$
    \Else
    \State Output ``Reject".
    \EndIf
  \end{algorithmic}
\label{alg:protocol}
\end{algorithm}

\smallskip
\begin{theorem}[$\tilde{\cF}$-identity testing for all distributions $\Longrightarrow$ Interactive verification]\label{thm:IT-to-IP}
Let $\cF\subseteq\{0,1\}^{\cX}$ be some
function class with lifted class
$\tilde{\cF}$. Suppose there exists an $(\epsilon/2, \delta/2)$-identity tester
$\cA_{\tilde{\cF}}$
for all reference distributions against $\tilde{\cF}$ with sample complexity
$m_{\tilde{\cF}}$.
Then there exists a distribution-free verification protocol (\cref{alg:protocol})
that verifies $\cF$ to excess error
$\varepsilon$ and confidence $1-\delta$ with verifier sample complexity
$m_{\tilde{\cF}}$.
\end{theorem}

\begin{proof}
\sloppy
\textbf{Completeness:}
Since the prover $P$ is honest, it sends $\cD'=\cD$ and $f'\in
\arg\min_{f\in \cF}\err(h,\cD)$.
\cref{line:reject1} does not reject since
$\err(f',\cD')=\opt(\cF,\cD')$, and
$\cA_{\tilde{\cF}}$ accepts with probability at least $1-\delta/2$ since
the tester sees identical distributions.
Hence the verifier outputs $h$ that achieves error $\opt(\cF,\cD)$
with probability at least $1-\delta$.

\textbf{Soundness:} If the protocol rejects at \cref{line:reject1}, then the
verifier rejects and we are done. Suppose that step does not reject, so that the
verifier receives a hypothesis $f'\in\arg\min_{f\in
\cF}\err(f,\cD').$ In the subsequent steps, if the tester $\cA_{\tilde{\cF}}$
rejects then the verifier rejects as well, and again we are done. With
probability at least $1-\delta/2$, if $\cA_{\tilde{\cF}}$ accepts we have that,
for all $f \in \cF$,
\begin{equation}\label{eq:close-in-error}
    \abs{\err(f, \cD') - \err(f, \cD)}
    = \lvert\ee_{\cD'}[\tilde{f}] - \ee_{\cD}[\tilde{f}]\rvert
    \leq\varepsilon/2.
\end{equation}
This implies that
\begin{equation}
    \label{eq:close-in-opt}
    \abs{\opt(\cF, \cD') - \opt(\cF, \cD)} \le \epsilon/2 \,.
\end{equation}
But then, if the protocol returns hypothesis $f'$, we conclude that
\begin{align*}
    \err(f', \cD)
    &\le \err(f', \cD') + \epsilon/2 \tag{By \eqref{eq:close-in-error}} \\
    &= \opt(\cF, \cD') + \epsilon/2 \tag{By \cref{line:reject1}} \\
    &\le \opt(\cF, \cD) + \epsilon \tag{By \eqref{eq:close-in-opt}} \,.
\end{align*}
\end{proof}

\subsubsection{Lifting identity testers from $\cF$ to $\tilde{\cF}$}
\label{subsec:lifting}
\newcommand{\ak}{\cA_k}
In this section, we show a reduction lifting an identity tester against a
function class $\cF$ to the function class $\tilde{\cF}$.

\paragraph*{Notation.} Let $D\in\Delta(\cX\times \cY)$ be a distribution over
some set $\cX$ and $\cY=\bits$. For each $b \in \bits$, let $D[y=b] \define
\Pr_{y\sim D_{\cY}}[y=b]$ where $D_{\cY}$ denotes the marginal distribution on
$\cY$. Write $D_b \in \Delta(\cX)$ for the distribution on $\cX$ of $D$
conditioned on $Y=b$.

Our reduction relies on the following definition of synthetic distributions.

\begin{definition}[Synthetic distributions for $\tilde{P}$ and $\tilde{Q}$]
    Given distributions $\tilde{P}$ and $\tilde{Q}$ over $\cX \times \bits$, we define the
    distributions $\tilde{R}(\tilde{P}, \tilde{Q})$ over $\cX \times \bits$ and $R^b(\tilde{P}, \tilde{Q})$ over
    $\cX$, for each $b \in \bits$, as follows:
    \begin{enumerate}
        \item To draw a sample $(x,y) \sim \tilde{R}(\tilde{P}, \tilde{Q})$, we first draw $y
        \sim \tilde{P}_\cY$ and then $x \sim \tilde{Q}_y$.

        \item  For each $b \in \bits$, to draw a sample $x \sim R^b(\tilde{P},\tilde{Q})$, we
        first draw $y \sim \tilde{P}_\cY$. If $y = b$, we let $x \sim \tilde{P}_y$. Otherwise,
        we let $x \sim \tilde{Q}_y$.
    \end{enumerate}
    Note that a sample from $R^b(\tilde{P},\tilde{Q})$ may be produced using one
    sample from $\tilde{P}$ and one sample from $\tilde{Q}_{1-b}$.
\end{definition}

It turns out that $\dist_{\tilde{\cF}}(\tilde{P}, \tilde{Q})$ is captured by the
distance between the marginals $\tilde{P}_\cY$ and $\tilde{Q}_\cY$, along with
the distances between the marginals $\tilde{P}_\cX$ and $R^b$ against $\cF$:

\begin{lemma}
    \label{lemma:synthetic}
    For all distributions $\tilde{P}, \tilde{Q} \in \Delta(\cX \times \bits)$, letting
    $\tilde{R} \define \tilde{R}(\tilde{P}, \tilde{Q})$ and $R^b \define R^b(\tilde{P}, \tilde{Q})$ for each $b
    \in \bits$, for each $f : \cX \to \bits$ we have
    \[
        \dist_{\tilde{f}}(\tilde{P}, \tilde{Q})
        \le 2\abs{\tilde{P}[y=1] - \tilde{Q}[y=1]} + \dist_f(\tilde{P}_\cX, R^0) + \dist_f(\tilde{P}_\cX, R^1)
        \,,
    \]
    and therefore, for any family $\cF$ of Boolean functions over $\cX$,
    \[
        \dist_{\tilde{\cF}}(\tilde{P}, \tilde{Q})
        \le 2\abs{\tilde{P}[y=1] - \tilde{Q}[y=1]}
            + \dist_\cF(\tilde{P}_\cX, R^0) + \dist_\cF(\tilde{P}_\cX, R^1) \,.
    \]
\end{lemma}
\begin{proof}
    By the triangle inequality,
    \[
        \dist_{\tilde{f}}(\tilde{P}, \tilde{Q})
        \le \dist_{\tilde{f}}(\tilde{P}, \tilde{R}) + \dist_{\tilde{f}}(\tilde{Q}, \tilde{R}) \,.
    \]
    By the definition of $\tilde{R}$, the second term is
    \begin{align*}
        \dist_{\tilde{f}}(\tilde{Q}, \tilde{R})
        &= \abs{\E_{\tilde{Q}}[\tilde{f}] - \E_{\tilde{R}}[\tilde{f}]} \\
        &= \abs{\tilde{Q}[y=1]\E_{x \sim \tilde{Q}_1}[\tilde{f}(x,1)]
                + \tilde{Q}[y=0]\E_{x \sim \tilde{Q}_0}[\tilde{f}(x,0)] \\
        &\qquad \qquad
                - \tilde{P}[y=1]\E_{x \sim \tilde{Q}_1}[\tilde{f}(x,1)]
                - \tilde{P}[y=0]\E_{x \sim \tilde{Q}_0}[\tilde{f}(x,0)]} \\
        &\le \abs{\tilde{Q}[y=1] - \tilde{P}[y=1]} + \abs{\tilde{Q}[y=0] - \tilde{P}[y=0]} \\
        &= 2\abs{\tilde{P}[y=1] - \tilde{Q}[y=1]} \,,
    \end{align*}
    where we used the triangle inequality and the fact that each
    expectation is bounded between $0$ and $1$. As for the first term,
    writing $\mu \define \tilde{P}[y=1]$ and $\nu \define 1-\mu = \tilde{P}[y=0]$, we have
    \begin{align*}
        \dist_{\tilde{f}}(\tilde{P}, \tilde{R})
        &= \abs{\mu \E_{x \sim \tilde{P}_1}[\tilde{f}(x,1)]
                + \nu \E_{x \sim \tilde{P}_0}[\tilde{f}(x,0)]
                - \mu \E_{x \sim \tilde{Q}_1}[\tilde{f}(x,1)]
                - \nu \E_{x \sim \tilde{Q}_0}[\tilde{f}(x,0)]} \\
        &\le \abs{\mu \E_{x \sim \tilde{P}_1}[\tilde{f}(x,1)]
                - \mu \E_{x \sim \tilde{Q}_1}[\tilde{f}(x,1)]}
            + \abs{\nu \E_{x \sim \tilde{P}_0}[\tilde{f}(x,0)]
                - \nu \E_{x \sim \tilde{Q}_0}[\tilde{f}(x,0)]} \\
        &= \abs{\mu \E_{x \sim \tilde{P}_1}[1 - f(x)]
                - \mu \E_{x \sim \tilde{Q}_1}[1 - f(x)]}
            + \abs{\nu \E_{x \sim \tilde{P}_0}[f(x)]
                - \nu \E_{x \sim \tilde{Q}_0}[f(x)]} \\
        &= \abs{\mu \E_{\tilde{P}_1}[f] - \mu \E_{\tilde{Q}_1}[f]}
            + \abs{\nu \E_{\tilde{P}_0}[f] - \nu \E_{\tilde{Q}_0}[f]} \,,
    \end{align*}
    where the inequality is by the triangle inequality, the next equality
    is by the definition of $\tilde{f}$, and the last equality is by extracting
    the constant $1$ from the expectations. By adding and subtracting the
    same quantity inside each absolute value and applying the law of total
    expectation, we obtain
    \begin{align*}
        \dist_{\tilde{f}}(\tilde{P}, \tilde{R})
        &\le \abs{\mu \E_{\tilde{P}_1}[f] + \nu \E_{\tilde{P}_0}[f]
                - \mu \E_{\tilde{Q}_1}[f] - \nu \E_{\tilde{P}_0}[f]}
            + \abs{\nu \E_{\tilde{P}_0}[f] + \mu \E_{\tilde{P}_1}[f]
                - \nu \E_{\tilde{Q}_0}[f] - \mu \E_{\tilde{P}_1}[f]} \\
        &= \abs{\E_{\tilde{P}_\cX}[f] - \E_{R^0}[f]}
            + \abs{\E_{\tilde{P}_\cX}[f] - \E_{R^1}[f]} \\
        &= \dist_f(\tilde{P}_\cX, R^0) + \dist_f(\tilde{P}_\cX, R^1) \,. \qedhere
    \end{align*}
\end{proof}

Equipped with this lemma, we may reduce $\tilde{\cF}$-identity testing to $\cF$-identity testing.
The strategy and proof are similar to \cref{lemma:implicit-tester-reduction}.

\begin{theorem}\label{thm:lifting tester}
    Let $\cF$ be a class of Boolean functions over $\cX$ with lifted class
    $\tilde{\cF}$. Let $\cA_\cF$ be an $(\epsilon,\delta)$-identity tester for
    all reference distributions (over $\cX$) against $\cF$ with sample
    complexity $m_\cF$. Then there exists a $(3\epsilon,4\delta)$-identity
    tester $\cA_{\tilde{\cF}}$ for all reference distributions (over $\cX \times
    \bits$) against $\tilde{\cF}$ with sample complexity $m_{\tilde{\cF}} =
    2m_\cF + O\left(\frac{\log(1/\delta)}{\epsilon^2}\right)$.
\end{theorem}
\begin{proof}
    Let $\refdist{\tilde{P}}$ be any reference distribution over $\cX \times
    \bits$, given explicitly as input. Let $\tilde{P}$ be the unknown input
    distribution over $\cX \times \bits$. For each $b \in \bits$, let $R^b
    \define R^b(\refdist{\tilde{P}}, \tilde{P})$. The tester $\cA_{\tilde{\cF}}$
    proceeds as follows:
    \begin{enumerate}
        \item \label{item:lift-1}
        Use $O\left(\frac{\log(1/\delta)}{\epsilon^2}\right)$ samples to
        compute an estimate $\hat{\mu}$ of $\tilde{P}[y=1]$ to additive error at
        most $\epsilon/10$ with probability at least $1-\delta$ by a Hoeffding
        bound. If $\abs{\hat{\mu} - \refdist{\tilde{P}}[y=1]} > \epsilon/10$,
        \textbf{reject}.
        \item \label{item:lift-2}
        Draw $m_{\tilde{\cF}}$ samples $(x_i, y_i) \sim \tilde{P}$.
        For each sample $(x_i, y_i)$, add $x_i$ to the ordered tuple
        $S_0$ if $y_i=0$, or otherwise add $x_i$ to ordered tuple $S_1$
        if $y_i=1$.
        \item \label{item:lift-3}
        For each $b \in \bits$, simulate the tester $\cA_\cF$ with
        reference distribution $\left(\refdist{\tilde{P}}\right)_\cX$ and input distribution
        $R^b$, using a fresh element from $S_{1-b}$ whenever a sample from
        $\tilde{P}_{1-b}$ is required to draw a sample from $R^b$; if the tuple
        $S_{1-b}$ has no more fresh elements to allow this, abort the simulation.
        \item \label{item:lift-4}
        If either simulation of $\cA_\cF$ did not abort and rejected,
        \textbf{reject}. Otherwise, \textbf{accept}.
    \end{enumerate}
    The sample complexity is immediate, so it remains to argue correctness. By a
    union bound, with probability at least $1-3\delta$ the following events
    occur:
    \begin{enumerate}
        \item The estimate $\hat{\mu}$ satisfies $\abs{\hat{\mu} -
        \tilde{P}[y=1]} \le \epsilon/10$.
        \item For each $b \in \bits$, \emph{either} the tuple $S_{1-b}$ runs out
        of fresh elements for the simulation of $\cA_\cF$, \emph{or} that
        simulation (obtains fresh independent samples from the correct
        distribution and therefore) produces a correct output, i.e.\ satisfying
        its completeness and soundness conditions.
    \end{enumerate}
    Suppose that these events occur. We now separately argue completeness
    and soundness.

    \textbf{Completeness.} Suppose $\tilde{P} = \refdist{\tilde{P}}$. Then
    Step~\ref{item:lift-1} does not reject because the two distributions also have
    the same marginals on $\cY$ and by the accuracy of $\hat{\mu}$.
    Step~\ref{item:lift-4} does not reject because $R^b = R^b(\refdist{\tilde{P}},
    \tilde{P}) = \tilde{P}_\cX$ and by the completeness of $\cA_\cF$. Therefore
    $\cA_{\tilde{\cF}}$ accepts.

    \textbf{Soundness.} Suppose $\dist_{\tilde{\cF}}(\tilde{P}, \refdist{\tilde{P}}) > 3\epsilon$. By \cref{lemma:synthetic}, at least
    one of the quantities $2\abs{\refdist{\tilde{P}}[y=1] - \tilde{P}[y=1]}$,
    $\dist_\cF\left(\left(\refdist{\tilde{P}}\right)_\cX, R^0\right)$,
    and $\dist_\cF\left(\left(\refdist{\tilde{P}}\right)_\cX, R^1\right)$ is
    greater than $\epsilon$. We consider two cases.

    First, suppose $\abs{\refdist{\tilde{P}}[y=1] - \tilde{P}[y=1]} >
    \epsilon/4$. Then, by the triangle inequality and the accuracy of $\hat{\mu}$,
    we have $\abs{\hat{\mu} - \refdist{\tilde{P}}[y=1]} > \epsilon/4 - \epsilon/10
    > \epsilon/10$, so Step~\ref{item:lift-1} rejects.

    Second, suppose $\abs{\refdist{\tilde{P}}[y=1] - \tilde{P}[y=1]} \le \epsilon/4$.
    Then we must have
    $\dist_\cF\left(\left(\refdist{\tilde{P}}\right)_\cX, R^0\right) > \epsilon$ or
    $\dist_\cF\left(\left(\refdist{\tilde{P}}\right)_\cX, R^1\right) > \epsilon$.
    Suppose without loss of generality that
    $\dist_\cF\left(\left(\refdist{\tilde{P}}\right)_\cX, R^0\right) > \epsilon$
    (the other case proceeds analogously). If the tuple $S_1$ contains enough
    fresh elements to enable the simulation of $\cA_\cF$ on $b=0$,
    then by the soundness of $\cA_\cF$ we know that Step~\ref{item:lift-4}
    rejects. Therefore, it remains to bound the probability that the tuple $S_1$
    runs out of fresh elements.

    Let $\mu \define \refdist{\tilde{P}}[y=1]$ and $\mu' \define
    \tilde{P}[y=1]$. By the definition of $R^0$, the number of elements
    requested from $S_1$ is $N \sim \bin(m_{\cF}, \mu)$, while $\abs{S_1} \sim
    \bin(m_{\tilde{\cF}}, \mu')$. The tuple $S_1$ has enough fresh elements when
    $\abs{S_1} \ge N$.

    We first observe that it must hold that $\mu > \epsilon$. Indeed, if $\mu
    \le \epsilon$, then the definition of $R^0$ implies that
    $\dist_\varvarTV\left(\left(\refdist{\tilde{P}}\right)_\cX, R^0\right) \le
    \epsilon$, contradicting that
    $\dist_\cF\left(\left(\refdist{\tilde{P}}\right)_\cX, R^0\right) >
    \epsilon$.
    %Therefore, we also conclude that $\mu' > \epsilon - \epsilon/4 = 3\epsilon/4$.

    Now, by Hoeffding's inequality, for all $t > 0$ we have
    \[
        \Pr\left[\abs{\abs{S_1} - m_{\tilde{\cF}} \mu'} \ge t\right]
        \le 2\exp(-2t^2 / m_{\tilde{\cF}})
    \]
    and
    \[
        \Pr\left[\abs{N - m_\cF \mu} \ge t\right]
        \le 2\exp(-2t^2 / m_\cF)
        \le 2\exp(-2t^2 / m_{\tilde{\cF}}) \,.
    \]
    Note that, by the bounds $\abs{\mu-\mu'} \le \epsilon/4$, $m_{\tilde{\cF}}
    \ge 2m_\cF$, and $\mu > \epsilon$, we have
    \[
        m_{\tilde{\cF}} \mu' - m_\cF \mu
        \ge m_{\tilde{\cF}}\left(\mu - \frac{\epsilon}{4}\right)
            - \frac{m_{\tilde{\cF}}}{2} \mu
        = \frac{m_{\tilde{\cF}}}{2} \left(\mu - \frac{\epsilon}{2}\right)
        > \frac{m_{\tilde{\cF}} \epsilon}{4} \,.
    \]
    Therefore, if $\abs{\abs{S_1} - m_{\tilde{\cF}} \mu'} \le t$ and
    $\abs{N - m_\cF \mu} \le t$ for $t \define m_{\tilde{\cF}} \epsilon / 8$,
    we will have $\abs{S_1} \ge N$ as desired. By the above, these inequalities
    hold except with probability at most
    \[
        4\exp(-2t^2 / m_{\tilde{\cF}})
        = 4\exp\left(
            -\frac{2 \cdot m_{\tilde{\cF}}^2 \epsilon^2}{64 m_{\tilde{\cF}}}\right)
        = 4\exp\left(-\frac{m_{\tilde{\cF}} \epsilon^2}{32}\right)
        \le \delta \,,
    \]
    the last inequality by the choice of $m_{\tilde{\cF}} \ge C
    \frac{\log(1/\delta)}{\epsilon^2}$ for suitable constant $C > 0$. Therefore,
    with probability at least $1-\delta$ we have $\abs{S_1} \ge N$ and hence the
    simulation of $\cA_\cF$ rejects, which concludes the proof.
\end{proof}

\begin{corollary}($\cF$-identity testing for all reference distributions to verification of $\cF$ classes.)\label{cor:F-testing-to-F-IP}
    Let $\cF\subseteq\{0,1\}^{\cX}$ be some
    %symmetric
    function class. Suppose there exists an $(\epsilon, \delta)$-identity
    tester for all reference distributions against $\cF$ with sample
    complexity $m_\cF$.
    Then there exists a distribution-free verification protocol that verifies $\cF$ to
    excess error $O(\varepsilon)$ and confidence $1-O(\delta)$ with verifier sample
    complexity $2m_{\cF}+O\left(\frac{\log(1/\delta)}{\eps^2}\right)$
\end{corollary}

\begin{proof}
    Combine \cref{thm:IT-to-IP,thm:lifting tester}.
\end{proof}

\subsubsection{Application: verification protocols for unions of rectangles}
In this section, we provide applications that follow from the reduction from interactive proofs to
identity testing. First, we prove a statistically efficient protocol for the task of verification of
union of $k$ intervals. For this, we use an algorithm from \cite{DKN14}.
As already discussed, their notion of $\cA_k$ distance is precisely a specialization of our fooling
distance to the class of unions of $k$ intervals.
For the sake of completeness we define $\ak$ distance.

\begin{definition}(Union of $k$ intervals)
    A function $f:[0,1]\to\{0,1\}$ is a \emph{union of $k$ intervals}
    if there exists $k$ intervals
    $[a_1,b_1],\dots ,[a_k,b_k]$
    such that $f(x)=1$ iff $x$ is contained in one of the intervals
    $[a_i,b_i]$. Let $\cF_k$ denote the class of such functions. 
\end{definition}

\begin{definition}($\ak$ distance)
Let $k\in\N$ and let $\cF_k$
denote the class of unions of $k$ intervals. For two distributions $P,Q\in\Delta([0,1])$, the $\cA_k$ distance between $P$ and $Q$ is defined as
\begin{align*}
    \cA_k(P,Q)=\dist_{\cF_k}(P,Q)=\sup_{f\in\cF_k}\lvert \ee_P[f]-\ee_Q[f]\rvert
\end{align*}
\end{definition}

We require the following result on identity testing against the $\cA_k$ distance.

\begin{theorem}[\cite{DKN14}]
\label{thm:ak-testing-1d}
Given $\eps>0$, an integer $k$ with $k\geq 2$, sample access to a distribution $\tgtdist{P}$ over $[0,1]$,
and an explicit distribution $\refdist{P}$ over $[0,1]$,
there is an algorithm which uses $O(\sqrt{k}/\eps^2)$ samples from $\tgtdist{P}$,
and with probability at least $2/3$ distinguishes between $\tgtdist{P} =
\refdist{P}$ and $\dist_{\cF_k}(P,\refdist{P})> \eps$.
\end{theorem}

Combining \cref{cor:F-testing-to-F-IP,thm:ak-testing-1d}, we obtain the following result for
verifying the class of unions of $k$ intervals.

\begin{theorem}
    Let $k \in \N$, and let $\cF_k$ be the class of unions of $k$ intervals on $[0,1]$. There exists
    a distribution-free verification protocol for the class $\cF_k$ with verifier sample complexity
    $O\left(\frac{\sqrt{k}\log(1/\delta)}{\eps^2}\right)$.
\end{theorem}

The work of \cite{diakonikolas2023testingclosenessmultivariatedistributions} extended the notion of
$\cA_k$ distance to higher dimensions. In this setting, the class of distinguishers consists of
unions of $k$ disjoint $n$-dimensional axis-aligned rectangles in $\R^n$. They studied the problem
of closeness (or equivalence) testing under this notion, which is a harder problem than identity
testing. We now state the result they achieve.

\begin{theorem}\label{thm:ak-testing-d}
    \sloppy
    Let $\eps>0$ and integer $k>2$. There exists a computationally efficient algorithm which, given
    sample access to distributions with density functions $\refdist{p},p:\R^n\to \R^+$, draws
    $C\cdot2^{n/3}k^{6/7}\log^{3n}(k)/\eps^{\alpha_n}$ samples from $\refdist{p},p$ for a
    sufficiently large constant $C>0$ and $\alpha_n=O\left(n^2\cdot2^{2^{n+1}}\right)$, and with
    probability at least $2/3$ correctly distinguishes whether $\refdist{p}=p$ or
    $\cA^n_k(\refdist{p}, p) > \eps$, where $\cA^n_k$ is the natural generalization of $\cA_k$ to
    $n$ dimensions described above.
\end{theorem}

In constant-dimensional space, combining \cref{cor:F-testing-to-F-IP,thm:ak-testing-d}, we obtain
the following result for verifying the class of unions of $k$ disjoint $n$-dimensional axis-aligned
rectangles.

\thmrect*

\begin{remark}
    \label{rem:verification-efficiency}
    Our interactive proof protocols, in their current state derived from $\cF$-identity testing,
    completely ignore the cost of the work performed by the prover and the communication cost
    between the prover and verifier. While this might seem problematic, we remark that the
    prover only needs to send a description of the distribution which is close in fooling distance,
    and that natural classes of distributions enjoy succinct descriptions. The upside of our
    approach is that we obtain improved verifier sample complexity via black-box reductions, namely
    $O(\sqrt{k}/\epsilon^2)$ for unions of $k$ intervals (compared to $O(\sqrt{k}/\eps^{2.5})$
    obtained by \cite{MutrejaS23}) and $\tilde{O}(k^{6/7})$ for disjoint unions of $k$
    multidimensional rectangles (for which no protocols were known before). We suspect that one
    should be able to obtain doubly efficient protocols by inspecting these reductions in their
    concrete settings.
\end{remark}

\subsubsection{VC lower bounds for $\cF$-identity testing}
\label{sec:vc-lower-bounds}
In this section, we give sample complexity lower bounds for algorithms that
perform $\cF$-identity testing for all reference distributions (\cref{def:testing-for-all}),
which follow from the weak equivalence between $\cF$-identity testing and verification of
learning algorithms along with known sample complexity lower bounds in the
verification framework.

\begin{theorem}[\cite{GoldwasserRSY21}]
    \label{thm:interactive-proof-vc-lower-bound}
    For every $\eps,\delta\in\left(0,\frac{1}{8}\right)$, there exist constants $c_0,c_1,c_2>0$ and a sequence of classes $\cH_1,\cH_2,\dots$ such that:
    \begin{itemize}
        \item For all $d\in\N$, the class $\cH_d$ has VC dimension at most $c_0\cdot d\log d$.
        \item The sample complexity of proper PAC verifying $\cH_d$ is
        $\Omega(d)$. That is, if $(V_1,P_1),(V_2,P_2),\dots $ is a sequence such
        that for all $d\in\N,\;(V_d,P_d)$ is an interactive proof system that
        verifies $\cH_d$,
        using random examples such that the output is either ``reject" or in $\cH_d$,
        then for all $d\geq c_1$, $V_d$ uses at least
        $c_2\cdot d$ samples when executed
        %on input
        with parameters $(\eps,\delta)$.
    \end{itemize}
\end{theorem}

\begin{corollary}
    \label{cor:vc-all-distributions-lower-bound}
For every $\eps,\delta\in\left(0,\frac{1}{8}\right)$, there exists a sequence of classes $\cF_1,\cF_2,\dots$ such that:
\begin{itemize}
    \item For all $d\in\N$, the class $\cF_d$ has VC dimension at most $c_0\cdot d\log d$.
    \item Every identity tester for all reference distributions against $\cF_d$
    has sample complexity $\Omega(d)$
\end{itemize}
\end{corollary}
\begin{proof}
    Immediate consequence of
    \cref{thm:interactive-proof-vc-lower-bound,cor:F-testing-to-F-IP}.
\end{proof}

\begin{remark}[On the tightness of the Rademacher sample complexity characterization]
    \label{rem:rademacher-tight}
    \cref{cor:vc-all-distributions-lower-bound} shows that the Rademacher complexity-based sample
    complexity upper bound for $\cF$-identity testing from \cref{thm:intro-testing-upper-bounds} can
    be tight up to a logarithmic factor for some families $\cF$. For example, suppose for a
    contradiction we could improve that upper bound to $O(m^{1-c})$, for constant $c > 0$ and small
    constants $\epsilon, \delta$. We have the following bound for the Rademacher complexity of
    family $\cF_d$ with respect to any reference distribution $\refdist{P}$, see e.g.\ \cite[p.\
    192]{Bousquet2004}:
    \[
        \cR_m(\cF, \refdist{P}) \le C \cdot \sqrt{\frac{\mathrm{VC}(\cF_d)}{m}} \,,
    \]
    where $\mathrm{VC}(\cF_d)$ is the VC dimension of $\cF_d$ and $C > 0$ is a constant. Setting $m
    = \Theta(d \log d) > \mathrm{VC}(\cF_d)$ sufficiently large, we obtain $\cR_m(\cF, \refdist{P})
    \le \epsilon/16$. Thus the improved result would give an identity tester to any reference
    distribution $\refdist{P}$ against $\cF_d$ with sample complexity $O(m^{1-c}) = o(d)$,
    contradicting \cref{cor:vc-all-distributions-lower-bound}. Hence there can be no such
    improvement to \cref{thm:intro-testing-upper-bounds}.

    \sloppy
    Conversely, the lower bound from \cref{thm:rademacher-lower-bound} can also be tight for some
    classes. Specifically, the class $\cF_k$ of unions of $k$ intervals on $[0,1]$ has Rademacher
    complexity satisfying $\cR_k(\cF_k, \Unif([0,1])) = 1$, so \cref{thm:rademacher-lower-bound}
    gives a lower bound of $\Omega(\sqrt{k})$ for uniformity testing against $\cF_k$, which is
    matched by the upper bound from \cref{thm:ak-testing-1d}.
\end{remark}

\begin{theorem}[\cite{MutrejaS23}]
    \label{thm:ip-sqrt-vc-lb}
    Let $\eps\in(0,1),\delta =1/3$, let $\cX$ be a set and let $\cF\subseteq \{0,1\}^{\cX}$ be a hypothesis class with VC dimension $d$. Assume that
    $\langle V,P \rangle$ is an interactive proof system that PAC verifies $\cF$ with parameters $(\eps,\delta)$ with respect to the set of all distributions $\mathbb{D}=\Delta(\cX\times\bits)$, and the verifier uses $m_V(d,\eps)$ i.i.d. labeled samples. Then there exists constants $C,c>0$ such that  $m_V(d,\eps)\geq (C\cdot \sqrt{d}-c)/\eps^2$.
\end{theorem}

\begin{corollary}
    Let $\eps\in(0,1),\delta =1/3$, let $\cX$ be a set and let $\cF\subseteq
    \{0,1\}^{\cX}$ be a hypothesis class with VC dimension $d$. Then any
    $(\epsilon, \delta)$-identity tester for all reference distributions against
    $\cF$ has sample complexity $\Omega(\sqrt{d}/\eps^2)$.
\end{corollary}
\begin{proof}
    Immediate consequence of \cref{thm:ip-sqrt-vc-lb,cor:F-testing-to-F-IP}.
\end{proof}

%\section{TV Testing of Structured and Continuous Distributions via \texorpdfstring{$\cF$}{F}-Identity Testing}
\section{TV Testing of Structured Distributions via \texorpdfstring{$\cF$}{F}-Identity Testing}
\label{sec:structured}
The problem of identity testing in full generality requires samples that scale
with the square root of the domain size \cite{paninski2008coincidence}. This
implies that identity testing is impossible for the setting where the domain is
infinite, so in such cases, some assumption is needed in order to get any
testing algorithm. 

Inspired by the work of \cite{DKN14}, we consider the setting of identity
testing under the promise that the reference distribution and the distribution
being tested belong to a class of structured distributions. We provide
sufficient conditions on the structure of distributions that either make
identity testing possible at all, or yield computational and statistical
improvements over bounds attainable by standard identity testing (without
structural assumptions). We show how the notion of fooling distance is a useful
primitive that yields identity testers for structured distributions under TV
distance, and also lies at the core of reductions that leverage insights from
the testable learning literature to obtain computationally efficient TV testers
in the present setting.

Recall the definition of the problem of testing structured distributions from
\cref{def:identity-testing-structured-intro}.

\subsection{Statistically efficient identity testing}

In this section, we formalize the scheme by which $\cF$-identity testers may be used to
obtain identity testers against TV distance for structured classes of
distributions. We start by defining \emph{Scheff\'{e} sets}, which serve to connect
the structure of the class $\bbD$ with the structure of a related function
family $\cF_{\bbD}$.

\begin{definition}[Scheff\'{e} set]\label{def:sset}
    For any two distributions $P,Q\in\Delta(\cX)$, the \emph{Scheff\'{e} set}
    $\sset$ is $\sset \define \{x\in \cX\mid P(x)\geq Q(x)\}$.
\end{definition}

\begin{definition}[Scheff\'{e} function class]\label{def:ssclass}
    Let $\bbD$ be any class of distributions over domain $\cX$. For a Scheff\'{e}
    set $\sset$ defined with respect to distributions $P,Q\in\bbD$, let
    $f_{P,Q}:\cX\to \{0,1\}$ denote the indicator function of $\sset$, that is,
    $f_{P,Q} (x)=\ind\left[x\in \sset\right]$. We call
    $\cF_{\bbD }=\cup_{P\neq Q} \{f_{P,Q}\}$ the Scheff\'{e} function class for the
    class of distributions $\bbD$. We drop the subscript from $\cF_{\bbD}$ when
    the class of distributions is obvious. 
\end{definition}

\begin{remark}[Scheff\'{e} sets and total variation]\label{remark:eq}
    Let $\sset$ denote the Scheff\'{e} set corresponding to the distributions
    $P,Q\in\Delta(\cX)$. Note that $\sset\in\arg\max_{A\subseteq \cX}
    \Abs{P(A)-Q(A)}$. It follows that
    $\dist_\tv(P,Q)=\dist_{f_{P,Q}}(P,Q)$.  
\end{remark}

Therefore, given a distribution class $\bbD$ with Scheff\'{e} function class
$\cF$, identity testing distributions in $\bbD$ against TV distance reduces to
$\cF$-identity testing:

\begin{lemma}
    \label{lemma:f-identity-testing-implies-structured-testing}
    Let $\bbD$ be a  set of distributions over domain $\cX$, and let $\cF_{\bbD}$ be the Scheff\'{e}
    function class for class $\bbD$. Let $\cA$ be an $(\epsilon, \delta)$-identity tester for all
    reference distributions against $\cF_{\bbD}$ (see \cref{def:testing-for-all}). Then $\cA$ is
    also an $(\epsilon,\delta)$-identity tester for distribution class $\bbD$.

    In fact, the same conclusion holds even if the algorithm $\cA$ specified by
    \cref{def:testing-for-all} is only promised to work when the reference distribution
    $\refdist{P}$ given to it satisfies $\refdist{P} \in \bbD$.
\end{lemma}
\begin{proof}
    Let $\refdist{P} \in \bbD$ be a reference distribution and $P \in \bbD$ be an input
    distribution. By \cref{remark:eq}, we have $\dist_\tv(\refdist{P}, P) = \dist_\cF(\refdist{P},
    P)$. The claim follows.
\end{proof}

Therefore, following \cref{sec:vc-bounds}, we can obtain a simple sample complexity upper bound in
terms of the VC dimension of the Scheff\'{e} function class:

\begin{lemma}\label{lemma:structuredit}
    Let $\bbD$ be a (possibly infinite) set of distributions over (a possibly
    infinite) domain $\cX$. Let $\cF_{\bbD}$ be the Scheff\'{e} function class for
    class $\bbD$. Suppose $\cF_{\bbD}$ has finite VC dimension $d$. Then
    \cref{alg:idtesting} is an $(\eps,\delta)$-identity tester for distribution
    class $\bbD$ with sample complexity
    $O\left(\tfrac{d + \log(1/\delta)}{\eps^2}\right)$.
\end{lemma}
\begin{proof}
    Combine \cref{lemma:vc-upper-bound} with
    \cref{lemma:f-identity-testing-implies-structured-testing}.
\end{proof}

\begin{corollary}
    If $\bbD=\{D_1,D_2,\dots,D_n\}$ is a finite set of distributions, then there exists an identity tester for the class of distributions $\bbD$.
\end{corollary}
\begin{proof}
     As shown in \cref{lemma:structuredit}, if the VC dimension of $\cF$ is
     finite, then \cref{alg:idtesting} is an identity tester for the
     distribution class $\bbD$. Hence, it suffices to show that the function
     class $\cF$ has a finite VC dimension. Since the number of distributions in
     $\bbD$ is finite, the number of Scheff\'{e} sets is at most $\binom{n}{2}$, so
     the size of $\cF$ is at most $O(n^2)$. Hence the VC dimension of $\cF$ is
     at most $O(\log n)$. This proves our claim.
\end{proof}

As an immediate application, we show identity testers for the class of
degree-$d$ polynomial distributions and size-$s$ decision tree distributions
over the solid cube $[0,1]^n$. We define these classes of distributions below. 

\begin{definition}[Degree-$d$ polynomial distributions over {$[0,1]^n$}]
\label{def:degree-d-pdf}
Let $\cX=[0,1]^n$ and let $\cP_d$ denote the class of degree-$d$ multivariate
polynomials over $[0,1]^n$. Let $\bbD_d$ be the set of absolutely continuous
probability distributions with probability density function in $\cP_d$, that is,
\begin{align*}
    \bbD_d = \left\{
        p \in \cP_d : \int_{x\in \cX} p(x) \odif x = 1, \, p \ge 0
        \right\} \,.
\end{align*}
\end{definition}

We recall the following standard bound on the VC dimension of the class of
degree-$d$ polynomial threshold functions (PTFs) on $[0,1]^n$:

\begin{fact}[See, e.g., \cite{AB02}]
    \label{fact:vck}
    The VC dimension of class $\bbD_d$ over $[0,1]^n$ is
    $O\left(\binom{n+d}{d}\right)$.
\end{fact}

\begin{lemma}\label{lemma:degree-d-it}
    \cref{alg:idtesting} is an $(\epsilon, \delta)$-identity tester for
    distribution class $\bbD_d$ with sample complexity
    $O\left(\frac{\binom{n+d}{d} + \log(1/\delta)}{\eps^2}\right)$.
\end{lemma}
\begin{proof}
\sloppy
For $\refdist{P},P\in\bbD_{d}$, let $\refdist{p},p$ be the corresponding
probability density functions, which are degree-$d$ polynomials. Define
$t\left(x\right)=\refdist{p}(x)-p\left(x\right)$ which is also a degree-$d$
polynomial. We now look at the Scheff\'{e} set for the distributions
$\refdist{P},P$. The Scheff\'{e} set $\cH_{\refdist{P},P}=\left\{ x\in \cX\mid
\refdist{p}(x)\geq p(x)\right\} =\left\{ x\in \cX\mid t\left(x\right)\geq
0\right\}.$ Observe that $f_{\refdist{P},P}\left(x\right)=\ind\left[x\in
\cH_{\refdist{P},P}\right]$ is a degree-$d$ PTF.
Thus the Scheff\'{e} function class for the distribution class
$\bbD_d$ is a subset of the class of degree-$d$ PTFs.
The proof follows from \cref{lemma:structuredit,fact:vck}.
\end{proof}

\begin{definition}[Decision tree distributions over {$[0,1]^n$}]\label{def:dtd}
    For any parameters $n,s\in \N$, let $\bbT_{n,s}$ denote the set of absolutely continuous
    probability distributions $P$ whose density functions $p:[0,1]^n\to\R^+$, with $\int_{x\in
    [0,1]^n} p(x) \odif x = 1$, are such that $p$ can be computed by some decision tree with at most
    $s$ nodes where each internal node is of the form ``$x_i\leq t?$" for some $t\in \R$, and the
    density value at each leaf is a non-negative real number. We call such $P$ a (size-$s$)
    \emph{decision tree distribution}. For a decision tree distribution $P$, we write $L(P)$ for the
    set of leaves of a size-$s$ decision tree computing the density of $P$ (we implicitly fix a
    single arbitrary decision tree representation for $P$), where each leaf $\ell \in L(P)$ is
    identified with the subset of $[0,1]^n$ obtained by the intersection of the threshold
    restrictions in the path from the root to leaf $\ell$.
\end{definition}

We will use the following bound on the VC dimension of the class of size-$s$
decision trees on $[0,1]^n$:

\begin{fact}\label{fact:vctree}(\cite{leboeuf2020decisiontreespartitioningmachines})
    Let $\cT_{n,s}$ be the class of Boolean functions represented by size-$s$ decision
    trees over $[0,1]^n$. The VC dimension of the function class $\cT_{n,s}$ is
    $O(s\log (ns))$.
\end{fact}

\begin{lemma}(Identity tester for decision tree distributions over $[0,1]^n$)\label{lem:it-for-dt}
    %Let $\refdist{P}\in \bbT_{n,s}$ be a decision tree distribution of size at most $s$ over
    %$\cX=[0,1]^n$. Then
    \cref{alg:idtesting} is an $(\eps,\delta)$-identity tester for distribution
    class $\bbT_{n,s}$ with sample complexity $O\left(\tfrac{s^2\log (ns) + \log
    (1/\delta)}{\eps^2}\right)$.
\end{lemma}

\begin{proof}
    Let $\refdist{P},P$ be any size-$s$ decision tree distributions and let $\refdist{p}$ and $p$ be
    the corresponding probability density functions. Consider the Scheff\'{e} set for the distributions
    $\refdist{P},P$, namely $\cH_{\refdist{P},P}=\{x\in X\mid \refdist{p}(x)\geq p(x)\}$. We
    show that the function $f_{\refdist{P},P}(x)=\ind[x\in \cH_{\refdist{P},P}]$ can be represented
    by a decision tree of size at most $s^2$. 
    
    %Consider $\Pi=\left\{ R=\refdist{P}(\ell)\cap P(\ell)\neq\emptyset:\refdist{P}(\ell)\in
    %L(\refdist{P}),P(\ell)\in L(P)\right\} $.
    Consider $\Pi=\left\{  \ell_1 \cap \ell_2 \neq \emptyset : \ell_1 \in
    L(\refdist{P}), \ell_2 \in L(P)\right\} $.
Observe that each $R\in\Pi$ is a subcube and $\Pi$ forms a valid
partition of $[0,1] ^{n}$. Since the distributions $\refdist{P}$ and $P$ are uniform over each
    subcube, we know that for all $R\in \Pi$, it holds that $\refdist{p}(x)\geq p(x)$ for all $x \in
    R$ or $\refdist{p}(x)\leq p(x)$ for all $x \in R$. Thus $\Pi$, which can be represented by a
    decision tree of size at most $s^2$ (formed by composing the decision trees for
    $\refdist{P},P$), computes (by assigning values to the leaves appropriately) the function
    $f_{\refdist{P},P}$. Hence the Scheff\'{e} function class for the distribution class $\bbT_{n,s}$ is
    a subset of the class of binary decision trees $\cT_{n,s^2}$. The proof follows from
    \cref{lemma:structuredit,fact:vctree}.
\end{proof}

\subsection{Computationally efficient identity testing}
\label{sec:structured-computationally-efficient}

In the previous section, we designed identity testers for classes of structured distributions by
leveraging the fact that the total variation distance is equal to the $\cF$-distance for suitably
chosen function classes $\cF$. In this section we show how we can use algorithms from the testable
learning framework (defined in \cref{subsec:TL}) to design computationally efficient testers.

%\paragraph{Testable Learning algorithms to Identity Testing for Structured Distributions}

\begin{theorem}\label{thm:tl-to-structured}
    Let $\bbD$ be a class of distributions over $\cX$ containing the uniform distribution over
    $\cX$. Let $\cF$ denote the Scheff\'{e} function class corresponding to $\bbD$. If there exists a
    testable learning algorithm for $\cF$ with respect to the uniform distribution on $\cX$,
    then there exists a uniformity tester for the class of distributions $\bbD$. The time
    complexity of the constructed uniformity tester is polynomial in that of the testable learning
    algorithm. The sample complexity also matches that of the testable learning algorithm plus an
    additional $O\left(\frac{\log(1/\delta)}{\eps^2}\right)$ samples.
\end{theorem}
\begin{proof}
    By \cref{thm:reduction}, there is a uniformity tester against class $\cF$ with the announced
    time and sample complexity. Since class $\bbD$ contains the uniform distribution, that algorithm
    is then also a uniformity tester for distribution class $\bbD$ by \cref{remark:eq}.
\end{proof}

We use the following result on testable learning of degree-$2$ PTFs.

\begin{theorem}[\cite{GollakotaKK23}]\label{thm:tl-degree-2ptf}
    Let $\cF$ be the class of degree-$2$ polynomial threshold functions over $\{0,1\}^n$. Let $\eps>0$. Then there exists a testable learning algorithm with respect to the uniform distribution over $\{0,1\}^n$ up to excess error $\eps$ with time and sample complexity $n^{\tilde{O}(1/\eps^9)}$.
\end{theorem}

\begin{corollary}(Uniformity Tester for $\bbD_2$ distributions)
    Let $\bbD_2$ be the set of degree-$2$ polynomial distributions over $\{0,1\}^n$ (defined
    analogously to \cref{def:degree-d-pdf}). Then, there exists a uniformity testing algorithm for
    class $\bbD_2$ with time and sample complexity $n^{\tilde{O}(1/\eps^9)}$. 
\end{corollary}
\begin{proof}
    First, observe that the uniform distribution belongs in the class of degree-$2$ polynomial
    distributions since it has constant mass function. Following the proof of
    \cref{lemma:degree-d-it}, we can similarly argue that the Scheff\'{e} function class of $\bbD_2$ is
    contained in $\cF_2$, which is the set of degree-$2$ polynomial threshold functions.
    \cref{thm:tl-degree-2ptf} gives a testable learning algorithm for the class $\cF_2$. Following
    \cref{thm:tl-to-structured}, we thus get a uniformity tester for distribution class $\bbD_2$
    with time and sample complexity $n^{\tilde{O}(1/\eps^9)}$.
\end{proof}

\begin{corollary}(Uniformity Tester for $\bbT_{n,s}$ distributions)
    \label{cor:uniformity-decision-trees}
    Let $\bbT_{n,s}$ be the set of size-$s$ decision tree distributions over $\{0,1\}^n$ (defined
    analogously to \cref{def:dtd}). Then, there exists a uniformity testing algorithm for class
    $\bbT_{n,s}$ with time and sample complexity $n^{O(\log(s/\epsilon))}$. 
\end{corollary}
\begin{proof}
    The proof follows the proof of \cref{lem:it-for-dt}. First, observe that the uniform
    distribution can be represented as a depth $0$ decision tree. Similar to \cref{lem:it-for-dt},
    we can similarly argue that the Scheff\'{e} function class for $\bbT_{n,s}$ is a subset of
    size-$s^2$ decision trees. \cref{thm:tl-dt} gives a testable learning algorithm for size-$s^2$
    decision trees. From \cref{thm:tl-to-structured} we get a uniformity tester for the class of
    size-$s$ decision tree distributions with time and sample complexity $n^{O(\log(s/\epsilon))}$. 
\end{proof}

\subsubsection{Testing identity to decision tree distributions}

While \cref{cor:uniformity-decision-trees} gives a computationally efficient \emph{uniformity}
tester for decision tree distributions, the reduction to testable learning does not immediately lend
itself to the broader problem of \emph{identity} testing. In this section, we give computationally
efficient identity testers to any target decision tree distribution, against the class of decision
tree distributions -- and hence TV testers when the input distribution is also a decision tree
distribution -- by combining the reduction to testable learning with testing-specific techniques.

As above, write $\bbT_{n,s}$ for the class of size-$s$ decision tree distributions over
$\bits^n$, and write $\cT_{n,s}$ for the class of Boolean functions computed by size-$s$ decision
trees over $\bits^n$.

\paragraph*{Testing against decision tree distributions.} We first give an identity tester to any
reference distribution in $\bbT_{n,s}$ against the class $\cT_{n,s}$ of distinguishers. The idea is
to use a standard identity tester to test the marginal input distribution over the leaves of the
reference distribution, and then test uniformity of the conditional distribution within each leaf
against $\cT_{n,s}$.

\begin{lemma}
    \label{lemma:identity-dt-dist-against-dt}
    There exists an algorithm which, given as input the explicit description of any decision tree
    distribution $\refdist{P} \in \bbT_{n,s}$, is an $\epsilon$-identity tester to $\refdist{P}$
    against $\cT_{n,s}$ with time and sample complexity $n^{O(\log(s/\epsilon))}$.
\end{lemma}
\begin{proof}
    Recall that we write $\ell \in L(\refdist{P})$ for the leaves of decision tree distribution
    $\refdist{P}$. Abusing notation, for each $x \in \bits^n$, write $L(x) \in L(\refdist{P})$ for
    the leaf of $\refdist{P}$ containing $x$. For any distribution $Q$ over $\bits^n$, write
    $Q^{L(\refdist{P})}$ for its marginal over the leaves of $\refdist{P}$, and $Q \mid \ell$ for
    its conditional distribution on the leaf (subcube) $\ell \subseteq \bits^n$ -- set arbitrarily
    if $Q$ places zero mass on $\ell$. Write $Q(\ell)$ for the mass of $Q$ on $\ell$.

    Given sample access to input distribution $P$, the tester proceeds as follows:
    \begin{enumerate}
        \item Draw $O(\sqrt{s}/\epsilon^2)$ samples and use an identity tester (e.g.\
            \cite{valiant2017automatic}) to test the identity of input distribution
            $P^{L(\refdist{P})}$ to reference distribution $\refdist{P}^{L(\refdist{P})}$ with
            proximity parameter $\epsilon/10$ and failure probability $1/10$. Note that a sample
            from $P^{L(\refdist{P})}$ may be produced by drawing $x \sim P$ and outputting $L(x)$.
            If that tester rejects, \textbf{reject}.

        \item For each leaf $\ell \in L(\refdist{P})$ satisfying $\refdist{P}(\ell) \ge
            \frac{\epsilon}{10s}$, do the following. Test uniformity of input distribution $P \mid
            \ell$ against $\cT_{n,s}$ with proximity parameter $\epsilon/10$ and failure probability
            $\delta = \frac{1}{10s}$, which requires $n^{O(\log(s/\epsilon))}$ samples and time by
            \cref{thm:tl-dt,thm:reduction}. To draw a sample from $P \mid \ell$, repeatedly draw
            samples $x \sim P$ until one satisfies $L(x) = \ell$. If either the tester rejects, or
            more than $\frac{20s}{\epsilon} \cdot n^{C \cdot \log(s/\epsilon)}$ samples $x \sim P$
            are drawn before we have enough samples to run the uniformity tester, where $C > 0$ is a
            sufficiently large absolute constant, \textbf{reject}.

        \item If none of the previous steps rejected, \textbf{accept}.
    \end{enumerate}

    It is clear that the time and sample complexity is indeed $n^{O(\log(s/\epsilon))} \cdot
    \poly(s/\epsilon) = n^{O(\log(s/\epsilon))}$, so now we argue correctness.

    \textbf{Completeness.} Suppose $P = \refdist{P}$. The first step rejects with probability at
    most $1/10$ because $P^{L(\refdist{P})} = \refdist{P}^{L(\refdist{P})}$. We now consider the
    second step. Let $\ell \in L(\refdist{P})$ satisfy $\refdist{P}(\ell) \ge \frac{\epsilon}{10s}$.
    If the algorithm successfully draws enough samples $x \sim P \mid \ell$, then the uniformity
    tester rejects with probability at most $\delta = \frac{1}{10s}$ since $P \mid \ell =
    \refdist{P} \mid \ell$ is the uniform distribution over $\ell$ by the definition of decision
    tree distribution. On the other hand, by a Chernoff bound, the tester fails to draw enough
    samples with probability
    \[
        \exp\left(-\Theta\left(n^{C \cdot \log(s/\epsilon)}\right)\right)
        \le \frac{1}{10s} \,.
    \]
    Hence, by a union bound the tester rejects with probability at most $\frac{1}{10} + s \cdot 2
    \cdot \frac{1}{10s} = \frac{3}{10}$.

    \textbf{Soundness.} Suppose $\dist_{\cT_{n,s}}(P, \refdist{P}) > \epsilon$. Let $P_*$ be the
    distribution on $\bits^n$ obtained by first sampling a leaf $\ell \sim
    \refdist{P}^{L(\refdist{P})}$ and then producing $x \sim P \mid \ell$. By the triangle
    inequality, we have
    \[
        \epsilon < \dist_{\cT_{n,s}}(P, \refdist{P})
        \le \dist_{\cT_{n,s}}(P, P_*) + \dist_{\cT_{n,s}}(P_*, \refdist{P})
        \le \dist_\tv(P, P_*) + \dist_{\cT_{n,s}}(P_*, \refdist{P}) \,.
    \]
    Suppose $\dist_\tv(P, P_*) > \epsilon/2$. By the definition of $P_*$, we have that $P \mid \ell
    = P_* \mid \ell$ for each $\ell \in L(\refdist{P})$, and hence $\dist_\tv(P, P_*) =
    \dist_\tv(P^{L(\refdist{P})}, P_*^{L(\refdist{P})}) = \dist_\tv(P^{L(\refdist{P})},
    \refdist{P}^{L(\refdist{P})})$, the last equality since $P_*^{L(\refdist{P})} =
    \refdist{P}^{L(\refdist{P})}$ by construction. It follows that the tester rejects in the first
    step except with probability $1/10$, and we are done.

    Otherwise, we have $\dist_{\cT_{n,s}}(P_*, \refdist{P}) > \epsilon/2$. Hence, for some Boolean
    function $f \in \cT_{n,s}$ we have
    \begin{align*}
        &\frac{\epsilon}{2}
        < \Abs{\E_{P_*}[f] - \E_{\refdist{P}}[f]}
        = \Abs{
            \sum_{\ell \in L(\refdist{P})} \refdist{P}(\ell) \left(
            \E_{P \mid \ell}[f] - \E_{\refdist{P} \mid \ell}[f]
            \right)
        }
        \le \E_{\ell \sim \refdist{P}^{L(\refdist{P})}}\left[
            \Abs{\E_{P \mid \ell}[f] - \E_{\refdist{P} \mid \ell}[f]}
        \right] \\
        &= \E_{\ell \sim \refdist{P}^{L(\refdist{P})}}\left[
            \Abs{\E_{P \mid \ell}[f] - \E_{\refdist{P} \mid \ell}[f]}
            \cdot \Ind[\refdist{P}(\ell) \ge \frac{\epsilon}{10s}]
        \right]
    + \E_{\ell \sim \refdist{P}^{L(\refdist{P})}}\left[
        \Abs{\E_{P \mid \ell}[f] - \E_{\refdist{P} \mid \ell}[f]}
        \cdot \Ind[\refdist{P}(\ell) < \frac{\epsilon}{10s}]
    \right] \\
    &\le \E_{\ell \sim \refdist{P}^{L(\refdist{P})}}\left[
        \Abs{\E_{P \mid \ell}[f] - \E_{\refdist{P} \mid \ell}[f]}
        \cdot \Ind[\refdist{P}(\ell) \ge \frac{\epsilon}{10s}]
        \right]
    + \Pr_{\ell \sim \refdist{P}^{L(\refdist{P})}}[\refdist{P}(\ell) < \frac{\epsilon}{10s}] \,.
    \end{align*}
    The second term in the last line is at most $s \cdot \frac{\epsilon}{10s} =
    \frac{\epsilon}{10}$. Therefore there exists $\ell \in L(\refdist{P})$ for which
    \[
        \Abs{\E_{P \mid \ell}[f] - \E_{\refdist{P} \mid \ell}[f]}
            \cdot \Ind\left[\refdist{P}(\ell) \ge \frac{\epsilon}{10s}\right]
        > \frac{\epsilon}{4} \,,
    \]
    which implies that $\refdist{P}(\ell) \ge \frac{\epsilon}{10s}$ and $\dist_{\cT_{n,s}}(P \mid
    \ell, \Unif(\ell)) > \frac{\epsilon}{4}$. Therefore the iteration corresponding to such $\ell$
    in the second step of the algorithm either rejects due to failing to obtain enough samples, in
    which case we are done, or it obtains enough samples to run the uniformity tester against
    $\cT_{n,s}$ on $P \mid \ell$, which rejects except with probability $\frac{1}{10s} \le
    \frac{1}{10}$.
\end{proof}

\paragraph*{TV identity testing for decision tree distributions.} In the setting of testing
structured distributions, where the input distribution is promised to satisfy $P \in \bbT_{n,s}$,
the result above implies a computationally efficient identity tester against TV distance -- hence
extending \cref{cor:uniformity-decision-trees} from uniformity to identity testing:

\begin{corollary}
    There exists an $\epsilon$-identity tester for distribution class $\bbT_{n,s}$ with time and
    sample complexity $n^{O(\log(s/\epsilon))}$.
\end{corollary}
\begin{proof}
    By \cref{lemma:f-identity-testing-implies-structured-testing}, it suffices to give an algorithm
    which, given the explicit description of any $\refdist{P} \in \bbT_{n,s}$, is an
    $\epsilon$-identity tester to $\refdist{P}$ against $\cT_{n,s^2}$ (recall this class contains
    the Scheff\'{e} function class for $\bbT_{n,s}$). We get such an algorithm by applying
    \cref{lemma:identity-dt-dist-against-dt} with size parameter $s^2$.
\end{proof}

\section*{Acknowledgments}

We thank Arnab Bhattacharyya, Eric Blais, and Adam Smith for useful discussions and suggestions.
We also thank the anonymous referees for many valuable comments and suggestions.
This project was initiated while R.F.\ was visiting Boston University hosted by Sofya Raskhodnikova,
who also provided feedback on part of this work as a graduate course project by R.D.

R.F.\ is supported by an NSERC Postdoctoral Fellowship and by NSF grant CCF-2106429. M.B.\ and R.D.\ are supported by NSF grant CNS-2046425.

\bibliographystyle{alpha}
\bibliography{references}

\newcommand{\etalchar}[1]{$^{#1}$}
\begin{thebibliography}{HJKRR18}

\bibitem[AAC{\etalchar{+}}23]{DBLP:conf/icml/AamandACINS23}
Anders Aamand, Alexandr Andoni, Justin~Y. Chen, Piotr Indyk, Shyam Narayanan,
  and Sandeep Silwal.
\newblock Data structures for density estimation.
\newblock In Andreas Krause, Emma Brunskill, Kyunghyun Cho, Barbara Engelhardt,
  Sivan Sabato, and Jonathan Scarlett, editors, {\em International Conference
  on Machine Learning, {ICML} 2023, 23-29 July 2023, Honolulu, Hawaii, {USA}},
  Proceedings of Machine Learning Research, pages 1--18. {PMLR}, 2023.

\bibitem[AB02]{AB02}
Martin Anthony and Peter~L. Bartlett.
\newblock {\em Neural Network Learning - Theoretical Foundations}.
\newblock Cambridge University Press, 2002.

\bibitem[ACB17]{arjovsky2017wasserstein}
Martin Arjovsky, Soumith Chintala, and L{\'e}on Bottou.
\newblock Wasserstein generative adversarial networks.
\newblock In {\em International conference on machine learning}, pages
  214--223. Pmlr, 2017.

\bibitem[ACS10]{AdamaszekCS10}
Michal Adamaszek, Artur Czumaj, and Christian Sohler.
\newblock Testing monotone continuous distributions on high-dimensional real
  cubes.
\newblock In Moses Charikar, editor, {\em Proceedings of the Twenty-First
  Annual {ACM-SIAM} Symposium on Discrete Algorithms, {SODA} 2010, Austin,
  Texas, USA, January 17-19, 2010}, pages 56--65. {SIAM}, 2010.

\bibitem[AEMM22]{amit2022integral}
Ron Amit, Baruch Epstein, Shay Moran, and Ron Meir.
\newblock Integral probability metrics pac-bayes bounds.
\newblock {\em Advances in Neural Information Processing Systems},
  35:3123--3136, 2022.

\bibitem[AFL24]{AdarFL24a}
Tomer Adar, Eldar Fischer, and Amit Levi.
\newblock Improved bounds for high-dimensional equivalence and product testing
  using subcube queries.
\newblock In Amit Kumar and Noga Ron{-}Zewi, editors, {\em Approximation,
  Randomization, and Combinatorial Optimization. Algorithms and Techniques,
  {APPROX/RANDOM} 2024, August 28-30, 2024, London School of Economics, London,
  {UK}}, volume 317 of {\em LIPIcs}, pages 48:1--48:21. Schloss Dagstuhl -
  Leibniz-Zentrum f{\"{u}}r Informatik, 2024.

\bibitem[Baz09]{Baz09}
Louay~MJ Bazzi.
\newblock Polylogarithmic independence can fool dnf formulas.
\newblock {\em SIAM Journal on Computing}, 38(6):2220--2272, 2009.

\bibitem[BBL04]{Bousquet2004}
Olivier Bousquet, St{\'e}phane Boucheron, and G{\'a}bor Lugosi.
\newblock Introduction to statistical learning theory.
\newblock In Olivier Bousquet, Ulrike von Luxburg, and Gunnar R{\"a}tsch,
  editors, {\em Advanced Lectures on Machine Learning: ML Summer Schools 2003,
  Canberra, Australia, February 2 - 14, 2003, T{\"u}bingen, Germany, August 4 -
  16, 2003, Revised Lectures}, pages 169--207. Springer Berlin Heidelberg,
  Berlin, Heidelberg, 2004.

\bibitem[BC18]{Bhattacharyya_2018}
Rishiraj Bhattacharyya and Sourav Chakraborty.
\newblock Property testing of joint distributions using conditional samples.
\newblock {\em ACM Transactions on Computation Theory}, 10(4):1–20, August
  2018.

\bibitem[BCG19]{BlaisCG19}
Eric Blais, Cl{\'{e}}ment~L. Canonne, and Tom Gur.
\newblock Distribution testing lower bounds via reductions from communication
  complexity.
\newblock {\em {ACM} Trans. Comput. Theory}, 11(2):6:1--6:37, 2019.

\bibitem[BCSV25]{BlancaCSV25}
Antonio Blanca, Zongchen Chen, Daniel Stefankovic, and Eric Vigoda.
\newblock Complexity of high-dimensional identity testing with coordinate
  conditional sampling.
\newblock {\em {ACM} Trans. Algorithms}, 21(1):7:1--7:58, 2025.

\bibitem[BFR{\etalchar{+}}00]{BatuFRSW00}
Tugkan Batu, Lance Fortnow, Ronitt Rubinfeld, Warren~D. Smith, and Patrick
  White.
\newblock Testing that distributions are close.
\newblock In {\em 41st Annual Symposium on Foundations of Computer Science,
  {FOCS} 2000, Redondo Beach, California, USA, November 12-14, 2000}, pages
  259--269. {IEEE} Computer Society, 2000.

\bibitem[BKM19]{DBLP:conf/colt/BousquetKM19}
Olivier Bousquet, Daniel Kane, and Shay Moran.
\newblock The optimal approximation factor in density estimation.
\newblock In Alina Beygelzimer and Daniel Hsu, editors, {\em Conference on
  Learning Theory, {COLT} 2019, 25-28 June 2019, Phoenix, AZ, {USA}},
  Proceedings of Machine Learning Research, pages 318--341. {PMLR}, 2019.

\bibitem[BLMT23]{BlancLMT23}
Guy Blanc, Jane Lange, Ali Malik, and Li{-}Yang Tan.
\newblock Lifting uniform learners via distributional decomposition.
\newblock In Barna Saha and Rocco~A. Servedio, editors, {\em Proceedings of the
  55th Annual {ACM} Symposium on Theory of Computing, {STOC} 2023, Orlando, FL,
  USA, June 20-23, 2023}, pages 1755--1767. {ACM}, 2023.

\bibitem[BLQT22]{BlancLQT22}
Guy Blanc, Jane Lange, Mingda Qiao, and Li{-}Yang Tan.
\newblock Properly learning decision trees in almost polynomial time.
\newblock {\em J. {ACM}}, 69(6):39:1--39:19, 2022.

\bibitem[BNNR11]{BaNNR11}
Khanh~Do Ba, Huy~L. Nguyen, Huy~N. Nguyen, and Ronitt Rubinfeld.
\newblock Sublinear time algorithms for earth mover's distance.
\newblock {\em Theory Comput. Syst.}, 48(2):428--442, 2011.

\bibitem[Can20]{gs009}
Cl{\'{e}}ment~L. Canonne.
\newblock {\em A Survey on Distribution Testing: Your Data is Big. But is it
  Blue?}
\newblock Number~9 in Graduate Surveys. Theory of Computing Library, 2020.

\bibitem[Can22]{Canonne22}
Cl{\'{e}}ment~L. Canonne.
\newblock Topics and techniques in distribution testing: {A} biased but
  representative sample.
\newblock {\em Found. Trends Commun. Inf. Theory}, 19(6):1032--1198, 2022.

\bibitem[CCK{\etalchar{+}}21]{CanonneCKLW21}
Cl{\'{e}}ment~L. Canonne, Xi~Chen, Gautam Kamath, Amit Levi, and Erik
  Waingarten.
\newblock Random restrictions of high dimensional distributions and uniformity
  testing with subcube conditioning.
\newblock In D{\'{a}}niel Marx, editor, {\em Proceedings of the 2021 {ACM-SIAM}
  Symposium on Discrete Algorithms, {SODA} 2021, Virtual Conference, January 10
  - 13, 2021}, pages 321--336. {SIAM}, 2021.

\bibitem[CCR{\etalchar{+}}25]{chakrabarty2025monotonicitytestinghighdimensionaldistributions}
Deeparnab Chakrabarty, Xi~Chen, Simeon Ristic, C.~Seshadhri, and Erik
  Waingarten.
\newblock Monotonicity testing of high-dimensional distributions with subcube
  conditioning.
\newblock In Michal Kouck{\'{y}} and Nikhil Bansal, editors, {\em Proceedings
  of the 57th Annual {ACM} Symposium on Theory of Computing, {STOC} 2025,
  Prague, Czechia, June 23-27, 2025}, pages 1019--1030. {ACM}, 2025.

\bibitem[CDKS20]{CanonneDKS20}
Cl{\'{e}}ment~L. Canonne, Ilias Diakonikolas, Daniel~M. Kane, and Alistair
  Stewart.
\newblock Testing bayesian networks.
\newblock {\em {IEEE} Trans. Inf. Theory}, 66(5):3132--3170, 2020.

\bibitem[CDV24]{casacuberta2024complexity}
S{\'{\i}}lvia Casacuberta, Cynthia Dwork, and Salil~P. Vadhan.
\newblock Complexity-theoretic implications of multicalibration.
\newblock In Bojan Mohar, Igor Shinkar, and Ryan O'Donnell, editors, {\em
  Proceedings of the 56th Annual {ACM} Symposium on Theory of Computing, {STOC}
  2024, Vancouver, BC, Canada, June 24-28, 2024}, pages 1071--1082. {ACM},
  2024.

\bibitem[CG18]{chiesa2018proofs}
Alessandro Chiesa and Tom Gur.
\newblock Proofs of proximity for distribution testing.
\newblock In Anna~R. Karlin, editor, {\em 9th Innovations in Theoretical
  Computer Science Conference, {ITCS} 2018, Cambridge, MA, USA, January 11-14,
  2018}, volume~94 of {\em LIPIcs}, pages 53:1--53:14. Schloss Dagstuhl -
  Leibniz-Zentrum f{\"{u}}r Informatik, 2018.

\bibitem[CGKR25]{casacuberta2025global}
S{\'{\i}}lvia Casacuberta, Parikshit Gopalan, Varun Kanade, and Omer Reingold.
\newblock How global calibration strengthens multiaccuracy.
\newblock In {\em 66th {IEEE} Annual Symposium on Foundations of Computer
  Science, {FOCS} 2025, Sydney, Australia, December 14-17, 2025}, pages
  1198--1227. {IEEE}, 2025.

\bibitem[CJLW21]{ChenJLW21}
Xi~Chen, Rajesh Jayaram, Amit Levi, and Erik Waingarten.
\newblock Learning and testing junta distributions with sub cube conditioning.
\newblock In Mikhail Belkin and Samory Kpotufe, editors, {\em Conference on
  Learning Theory, {COLT} 2021, 15-19 August 2021, Boulder, Colorado, {USA}},
  volume 134 of {\em Proceedings of Machine Learning Research}, pages
  1060--1113. {PMLR}, 2021.

\bibitem[CKK{\etalchar{+}}24]{chandrasekaran2024efficient}
Gautam Chandrasekaran, Adam Klivans, Vasilis Kontonis, Konstantinos
  Stavropoulos, and Arsen Vasilyan.
\newblock Efficient discrepancy testing for learning with distribution shift.
\newblock {\em Advances in Neural Information Processing Systems},
  37:137263--137308, 2024.

\bibitem[CM24]{ChenM24}
Xi~Chen and Cassandra Marcussen.
\newblock Uniformity testing over hypergrids with subcube conditioning.
\newblock In David~P. Woodruff, editor, {\em Proceedings of the 2024 {ACM-SIAM}
  Symposium on Discrete Algorithms, {SODA} 2024, Alexandria, VA, USA, January
  7-10, 2024}, pages 4338--4370. {SIAM}, 2024.

\bibitem[CRS15]{canonne2015testingprobabilitydistributionsusing}
Cl{\'{e}}ment~L. Canonne, Dana Ron, and Rocco~A. Servedio.
\newblock Testing probability distributions using conditional samples.
\newblock {\em {SIAM} J. Comput.}, 44(3):540--616, 2015.

\bibitem[Dan16]{daniely2016complexity}
Amit Daniely.
\newblock Complexity theoretic limitations on learning halfspaces.
\newblock In Daniel Wichs and Yishay Mansour, editors, {\em Proceedings of the
  48th Annual {ACM} {SIGACT} Symposium on Theory of Computing, {STOC} 2016,
  Cambridge, MA, USA, June 18-21, 2016}, pages 105--117. {ACM}, 2016.

\bibitem[DDK19]{DaskalakisD019}
Constantinos Daskalakis, Nishanth Dikkala, and Gautam Kamath.
\newblock Testing ising models.
\newblock {\em {IEEE} Trans. Inf. Theory}, 65(11):6829--6852, 2019.

\bibitem[DDS{\etalchar{+}}13]{daskalakis2013testing}
Constantinos Daskalakis, Ilias Diakonikolas, Rocco~A. Servedio, Gregory
  Valiant, and Paul Valiant.
\newblock Testing \emph{k}-modal distributions: Optimal algorithms via
  reductions.
\newblock In Sanjeev Khanna, editor, {\em Proceedings of the Twenty-Fourth
  Annual {ACM-SIAM} Symposium on Discrete Algorithms, {SODA} 2013, New Orleans,
  Louisiana, USA, January 6-8, 2013}, pages 1833--1852. {SIAM}, 2013.

\bibitem[DKK{\etalchar{+}}21]{DiakonikolasKKT21}
Ilias Diakonikolas, Daniel~M. Kane, Vasilis Kontonis, Christos Tzamos, and
  Nikos Zarifis.
\newblock Agnostic proper learning of halfspaces under gaussian marginals.
\newblock In Mikhail Belkin and Samory Kpotufe, editors, {\em Conference on
  Learning Theory, {COLT} 2021, 15-19 August 2021, Boulder, Colorado, {USA}},
  volume 134 of {\em Proceedings of Machine Learning Research}, pages
  1522--1551. {PMLR}, 2021.

\bibitem[DKK{\etalchar{+}}23]{diakonikolas2023efficient}
Ilias Diakonikolas, Daniel Kane, Vasilis Kontonis, Sihan Liu, and Nikos
  Zarifis.
\newblock Efficient testable learning of halfspaces with adversarial label
  noise.
\newblock {\em Advances in Neural Information Processing Systems},
  36:39470--39490, 2023.

\bibitem[DKK{\etalchar{+}}24]{DiakonikolasKKT24}
Ilias Diakonikolas, Daniel~M. Kane, Vasilis Kontonis, Christos Tzamos, and
  Nikos Zarifis.
\newblock Agnostically learning multi-index models with queries.
\newblock In {\em 65th {IEEE} Annual Symposium on Foundations of Computer
  Science, {FOCS} 2024, Chicago, IL, USA, October 27-30, 2024}, pages
  1931--1952. {IEEE}, 2024.

\bibitem[DKL24]{diakonikolas2023testingclosenessmultivariatedistributions}
Ilias Diakonikolas, Daniel~M. Kane, and Sihan Liu.
\newblock Testing closeness of multivariate distributions via ramsey theory.
\newblock In Bojan Mohar, Igor Shinkar, and Ryan O'Donnell, editors, {\em
  Proceedings of the 56th Annual {ACM} Symposium on Theory of Computing, {STOC}
  2024, Vancouver, BC, Canada, June 24-28, 2024}, pages 340--347. {ACM}, 2024.

\bibitem[DKN15a]{DiakonikolasKN15}
Ilias Diakonikolas, Daniel~M. Kane, and Vladimir Nikishkin.
\newblock Optimal algorithms and lower bounds for testing closeness of
  structured distributions.
\newblock In Venkatesan Guruswami, editor, {\em {IEEE} 56th Annual Symposium on
  Foundations of Computer Science, {FOCS} 2015, Berkeley, CA, USA, 17-20
  October, 2015}, pages 1183--1202. {IEEE} Computer Society, 2015.

\bibitem[DKN15b]{DKN14}
Ilias Diakonikolas, Daniel~M. Kane, and Vladimir Nikishkin.
\newblock Testing identity of structured distributions.
\newblock In Piotr Indyk, editor, {\em Proceedings of the Twenty-Sixth Annual
  {ACM-SIAM} Symposium on Discrete Algorithms, {SODA} 2015, San Diego, CA, USA,
  January 4-6, 2015}, pages 1841--1854. {SIAM}, 2015.

\bibitem[DKN17]{DiakonikolasKN17}
Ilias Diakonikolas, Daniel~M. Kane, and Vladimir Nikishkin.
\newblock Near-optimal closeness testing of discrete histogram distributions.
\newblock In Ioannis Chatzigiannakis, Piotr Indyk, Fabian Kuhn, and Anca
  Muscholl, editors, {\em 44th International Colloquium on Automata, Languages,
  and Programming, {ICALP} 2017, July 10-14, 2017, Warsaw, Poland}, volume~80
  of {\em LIPIcs}, pages 8:1--8:15. Schloss Dagstuhl - Leibniz-Zentrum
  f{\"{u}}r Informatik, 2017.

\bibitem[DKP19]{DiakonikolasKP19}
Ilias Diakonikolas, Daniel~M. Kane, and John Peebles.
\newblock Testing identity of multidimensional histograms.
\newblock In Alina Beygelzimer and Daniel Hsu, editors, {\em Conference on
  Learning Theory, {COLT} 2019, 25-28 June 2019, Phoenix, AZ, {USA}}, volume~99
  of {\em Proceedings of Machine Learning Research}, pages 1107--1131. {PMLR},
  2019.

\bibitem[DKW18]{DaskalakisKW18}
Constantinos Daskalakis, Gautam Kamath, and John Wright.
\newblock Which distribution distances are sublinearly testable?
\newblock In Artur Czumaj, editor, {\em Proceedings of the Twenty-Ninth Annual
  {ACM-SIAM} Symposium on Discrete Algorithms, {SODA} 2018, New Orleans, LA,
  USA, January 7-10, 2018}, pages 2747--2764. {SIAM}, 2018.

\bibitem[DL01]{devroye2001combinatorial}
Luc Devroye and G{\'a}bor Lugosi.
\newblock {\em Combinatorial methods in density estimation}, volume~1.
\newblock Springer, 2001.

\bibitem[DLLT23]{dwork2023pseudorandomness}
Cynthia Dwork, Daniel Lee, Huijia Lin, and Pranay Tankala.
\newblock From pseudorandomness to multi-group fairness and back.
\newblock In {\em The Thirty Sixth Annual Conference on Learning Theory}, pages
  3566--3614. PMLR, 2023.

\bibitem[Fel09]{Feldman09}
Vitaly Feldman.
\newblock On the power of membership queries in agnostic learning.
\newblock {\em J. Mach. Learn. Res.}, 10:163--182, 2009.

\bibitem[FGKP06]{feldman2006new}
Vitaly Feldman, Parikshit Gopalan, Subhash Khot, and Ashok~Kumar Ponnuswami.
\newblock New results for learning noisy parities and halfspaces.
\newblock In {\em 47th Annual {IEEE} Symposium on Foundations of Computer
  Science, {FOCS} 2006, Berkeley, California, USA, October 21-24, 2006,
  Proceedings}, pages 563--574. {IEEE} Computer Society, 2006.

\bibitem[FH24]{FH24}
Renato {Ferreira Pinto Jr} and Nathaniel Harms.
\newblock Distribution testing with a confused collector.
\newblock In Venkatesan Guruswami, editor, {\em 15th Innovations in Theoretical
  Computer Science Conference, {ITCS} 2024, January 30 to February 2, 2024,
  Berkeley, CA, {USA}}, volume 287 of {\em LIPIcs}, pages 47:1--47:14. Schloss
  Dagstuhl - Leibniz-Zentrum f{\"{u}}r Informatik, 2024.

\bibitem[FOS08]{FeldmanOS08}
Jon Feldman, Ryan O'Donnell, and Rocco~A. Servedio.
\newblock Learning mixtures of product distributions over discrete domains.
\newblock {\em {SIAM} J. Comput.}, 37(5):1536--1564, 2008.

\bibitem[GJK{\etalchar{+}}24]{GurJKRSS24}
Tom Gur, Mohammad~Mahdi Jahanara, Mohammad~Mahdi Khodabandeh, Ninad Rajgopal,
  Bahar Salamatian, and Igor Shinkar.
\newblock On the power of interactive proofs for learning.
\newblock In Bojan Mohar, Igor Shinkar, and Ryan O'Donnell, editors, {\em
  Proceedings of the 56th Annual {ACM} Symposium on Theory of Computing, {STOC}
  2024, Vancouver, BC, Canada, June 24-28, 2024}, pages 1063--1070. {ACM},
  2024.

\bibitem[GKK23]{GollakotaKK23}
Aravind Gollakota, Adam~R. Klivans, and Pravesh~K. Kothari.
\newblock A moment-matching approach to testable learning and a new
  characterization of rademacher complexity.
\newblock In Barna Saha and Rocco~A. Servedio, editors, {\em Proceedings of the
  55th Annual {ACM} Symposium on Theory of Computing, {STOC} 2023, Orlando, FL,
  USA, June 20-23, 2023}, pages 1657--1670. {ACM}, 2023.

\bibitem[GKR{\etalchar{+}}22]{gopalan2021omnipredictors}
Parikshit Gopalan, Adam~Tauman Kalai, Omer Reingold, Vatsal Sharan, and Udi
  Wieder.
\newblock Omnipredictors.
\newblock In Mark Braverman, editor, {\em 13th Innovations in Theoretical
  Computer Science Conference, {ITCS} 2022, Berkeley, CA, USA, January 31 -
  February 3, 2022}, LIPIcs, pages 79:1--79:21. Schloss Dagstuhl -
  Leibniz-Zentrum f{\"{u}}r Informatik, 2022.

\bibitem[GKSV24]{gollakota2023efficient}
Aravind Gollakota, Adam~R. Klivans, Konstantinos Stavropoulos, and Arsen
  Vasilyan.
\newblock An efficient tester-learner for halfspaces.
\newblock In {\em The Twelfth International Conference on Learning
  Representations, {ICLR} 2024, Vienna, Austria, May 7-11, 2024}.
  OpenReview.net, 2024.

\bibitem[GKSV25]{goel2025testing}
Surbhi Goel, Adam~R Klivans, Konstantinos Stavropoulos, and Arsen Vasilyan.
\newblock Testing noise assumptions of learning algorithms.
\newblock {\em arXiv preprint arXiv:2501.09189}, 2025.

\bibitem[GR00]{GR00}
Oded Goldreich and Dana Ron.
\newblock On testing expansion in bounded-degree graphs.
\newblock {\em Electron. Colloquium Comput. Complex.}, {TR00-020}, 2000.

\bibitem[GR09]{guruswami2009hardness}
Venkatesan Guruswami and Prasad Raghavendra.
\newblock Hardness of learning halfspaces with noise.
\newblock {\em {SIAM} J. Comput.}, 39(2):742--765, 2009.

\bibitem[GRSY21]{GoldwasserRSY21}
Shafi Goldwasser, Guy~N. Rothblum, Jonathan Shafer, and Amir Yehudayoff.
\newblock Interactive proofs for verifying machine learning.
\newblock In James~R. Lee, editor, {\em 12th Innovations in Theoretical
  Computer Science Conference, {ITCS} 2021, January 6-8, 2021, Virtual
  Conference}, volume 185 of {\em LIPIcs}, pages 41:1--41:19. Schloss Dagstuhl
  - Leibniz-Zentrum f{\"{u}}r Informatik, 2021.

\bibitem[GSSV24]{GoelSSV24}
Surbhi Goel, Abhishek Shetty, Konstantinos Stavropoulos, and Arsen Vasilyan.
\newblock Tolerant algorithms for learning with arbitrary covariate shift.
\newblock In Amir Globersons, Lester Mackey, Danielle Belgrave, Angela Fan,
  Ulrich Paquet, Jakub~M. Tomczak, and Cheng Zhang, editors, {\em Advances in
  Neural Information Processing Systems 38: Annual Conference on Neural
  Information Processing Systems 2024, NeurIPS 2024, Vancouver, BC, Canada,
  December 10 - 15, 2024}, 2024.

\bibitem[HJKRR18]{hebert2018multicalibration}
Ursula H{\'e}bert-Johnson, Michael Kim, Omer Reingold, and Guy Rothblum.
\newblock Multicalibration: Calibration for the (computationally-identifiable)
  masses.
\newblock In {\em International Conference on Machine Learning}, pages
  1939--1948. PMLR, 2018.

\bibitem[IT03]{IT03}
Piotr Indyk and Nitin Thaper.
\newblock Fast image retrieval via embeddings.
\newblock In {\em International Workshop on Statistical and Computational
  Theories of Vision (ICCV)}, 2003.

\bibitem[KKMS08]{kalai2008agnostically}
Adam~Tauman Kalai, Adam~R. Klivans, Yishay Mansour, and Rocco~A. Servedio.
\newblock Agnostically learning halfspaces.
\newblock {\em {SIAM} J. Comput.}, 37(6):1777--1805, 2008.

\bibitem[KL18]{kothari2018improper}
Pravesh~K. Kothari and Roi Livni.
\newblock Improper learning by refuting.
\newblock In Anna~R. Karlin, editor, {\em 9th Innovations in Theoretical
  Computer Science Conference, {ITCS} 2018, Cambridge, MA, USA, January 11-14,
  2018}, LIPIcs, pages 55:1--55:10. Schloss Dagstuhl - Leibniz-Zentrum
  f{\"{u}}r Informatik, 2018.

\bibitem[KPJ{\etalchar{+}}23]{kong2023covariate}
Insung Kong, Yuha Park, Joonhyuk Jung, Kwonsang Lee, and Yongdai Kim.
\newblock Covariate balancing using the integral probability metric for causal
  inference.
\newblock In {\em International Conference on Machine Learning}, pages
  17430--17461. PMLR, 2023.

\bibitem[KSV24]{KlivansSV24}
Adam~R. Klivans, Konstantinos Stavropoulos, and Arsen Vasilyan.
\newblock Testable learning with distribution shift.
\newblock In Shipra Agrawal and Aaron Roth, editors, {\em The Thirty Seventh
  Annual Conference on Learning Theory, June 30 - July 3, 2023, Edmonton,
  Canada}, volume 247 of {\em Proceedings of Machine Learning Research}, pages
  2887--2943. {PMLR}, 2024.

\bibitem[LLM20]{leboeuf2020decisiontreespartitioningmachines}
Jean{-}Samuel Leboeuf, Fr{\'{e}}d{\'{e}}ric Leblanc, and Mario Marchand.
\newblock Decision trees as partitioning machines to characterize their
  generalization properties.
\newblock In Hugo Larochelle, Marc'Aurelio Ranzato, Raia Hadsell,
  Maria{-}Florina Balcan, and Hsuan{-}Tien Lin, editors, {\em Advances in
  Neural Information Processing Systems 33: Annual Conference on Neural
  Information Processing Systems 2020, NeurIPS 2020, December 6-12, 2020,
  virtual}, 2020.

\bibitem[LQ26]{lange2026limitations}
Jane Lange and Mingda Qiao.
\newblock Limitations of membership queries in testable learning.
\newblock In Shubhangi Saraf, editor, {\em 17th Innovations in Theoretical
  Computer Science Conference, {ITCS} 2026, Bocconi University, Milan, Italy,
  January 27-30, 2026}, LIPIcs, pages 91:1--91:23. Schloss Dagstuhl -
  Leibniz-Zentrum f{\"{u}}r Informatik, 2026.

\bibitem[MMR09]{mansour2009domain}
Yishay Mansour, Mehryar Mohri, and Afshin Rostamizadeh.
\newblock Domain adaptation: Learning bounds and algorithms.
\newblock In {\em {COLT} 2009 - The 22nd Conference on Learning Theory,
  Montreal, Quebec, Canada, June 18-21, 2009}, 2009.

\bibitem[MPV25]{marcussen2025characterizing}
Cassandra Marcussen, Aaron Putterman, and Salil~P. Vadhan.
\newblock Characterizing the distinguishability of product distributions
  through multicalibration.
\newblock In Srikanth Srinivasan, editor, {\em 40th Computational Complexity
  Conference, {CCC} 2025, Toronto, Canada, August 5-8, 2025}, LIPIcs, pages
  19:1--19:19. Schloss Dagstuhl - Leibniz-Zentrum f{\"{u}}r Informatik, 2025.

\bibitem[MS08]{DBLP:conf/colt/MahalanabisS08}
Satyaki Mahalanabis and Daniel Stefankovic.
\newblock Density estimation in linear time.
\newblock In Rocco~A. Servedio and Tong Zhang, editors, {\em 21st Annual
  Conference on Learning Theory - {COLT} 2008, Helsinki, Finland, July 9-12,
  2008}, pages 503--512. Omnipress, 2008.

\bibitem[MS23]{MutrejaS23}
Saachi Mutreja and Jonathan Shafer.
\newblock {PAC} verification of statistical algorithms.
\newblock In Gergely Neu and Lorenzo Rosasco, editors, {\em The Thirty Sixth
  Annual Conference on Learning Theory, {COLT} 2023, 12-15 July 2023,
  Bangalore, India}, volume 195 of {\em Proceedings of Machine Learning
  Research}, pages 5021--5043. {PMLR}, 2023.

\bibitem[M{\"u}l97]{muller1997integral}
Alfred M{\"u}ller.
\newblock Integral probability metrics and their generating classes of
  functions.
\newblock {\em Advances in applied probability}, 29(2):429--443, 1997.

\bibitem[Pan08]{paninski2008coincidence}
Liam Paninski.
\newblock A coincidence-based test for uniformity given very sparsely sampled
  discrete data.
\newblock {\em {IEEE} Trans. Inf. Theory}, 54(10):4750--4755, 2008.

\bibitem[RS09]{RubinfeldS09}
Ronitt Rubinfeld and Rocco~A. Servedio.
\newblock Testing monotone high-dimensional distributions.
\newblock {\em Random Struct. Algorithms}, 34(1):24--44, 2009.

\bibitem[RV23]{rubinfeld2022testingdistributionalassumptionslearning}
Ronitt Rubinfeld and Arsen Vasilyan.
\newblock Testing distributional assumptions of learning algorithms.
\newblock In Barna Saha and Rocco~A. Servedio, editors, {\em Proceedings of the
  55th Annual {ACM} Symposium on Theory of Computing, {STOC} 2023, Orlando, FL,
  USA, June 20-23, 2023}, pages 1643--1656. {ACM}, 2023.

\bibitem[Sch47]{scheffe1947useful}
Henry Scheff{\'e}.
\newblock A useful convergence theorem for probability distributions.
\newblock {\em The Annals of Mathematical Statistics}, 18(3):434--438, 1947.

\bibitem[SFG{\etalchar{+}}12]{sriperumbudur2012empirical}
Bharath~K. Sriperumbudur, Kenji Fukumizu, Arthur Gretton, Bernhard
  Sch{\"o}lkopf, and Gert R.~G. Lanckriet.
\newblock {On the empirical estimation of integral probability metrics}.
\newblock {\em Electronic Journal of Statistics}, 6:1550--1599, 2012.

\bibitem[STW24]{slot2024testably}
Lucas Slot, Stefan Tiegel, and Manuel Wiedmer.
\newblock Testably learning polynomial threshold functions.
\newblock {\em Advances in Neural Information Processing Systems},
  37:3781--3831, 2024.

\bibitem[TTV09]{trevisan2009regularity}
Luca Trevisan, Madhur Tulsiani, and Salil~P. Vadhan.
\newblock Regularity, boosting, and efficiently simulating every high-entropy
  distribution.
\newblock In {\em Proceedings of the 24th Annual {IEEE} Conference on
  Computational Complexity, {CCC} 2009, Paris, France, 15-18 July 2009}, pages
  126--136. {IEEE} Computer Society, 2009.

\bibitem[Vad12]{vadhan2012pseudorandomness}
Salil~P Vadhan.
\newblock Pseudorandomness.
\newblock {\em Foundations and Trends{\textregistered} in Theoretical Computer
  Science}, 7(1-3):1--336, 2012.

\bibitem[Vad17]{vadhan2017learning}
Salil Vadhan.
\newblock On learning vs. refutation.
\newblock In {\em Conference on Learning Theory}, pages 1835--1848. PMLR, 2017.

\bibitem[VC15]{vapnik2015uniform}
Vladimir~N Vapnik and A~Ya Chervonenkis.
\newblock On the uniform convergence of relative frequencies of events to their
  probabilities.
\newblock In {\em Measures of complexity: festschrift for alexey chervonenkis},
  pages 11--30. Springer, 2015.

\bibitem[VV17]{valiant2017automatic}
Gregory Valiant and Paul Valiant.
\newblock An automatic inequality prover and instance optimal identity testing.
\newblock {\em {SIAM} J. Comput.}, 46(1):429--455, 2017.

\bibitem[Yat85]{yatracos1985rates}
Yannis~G Yatracos.
\newblock Rates of convergence of minimum distance estimators and kolmogorov's
  entropy.
\newblock {\em The Annals of Statistics}, 13(2):768--774, 1985.

\bibitem[Zol84]{zolotarev1984probability}
Vladimir~Mikhailovich Zolotarev.
\newblock Probability metrics.
\newblock {\em Theory of Probability \& Its Applications}, 28(2):278--302,
  1984.

\end{thebibliography}

\end{document}